\documentclass{article}
\usepackage[british]{babel}
\usepackage{amsmath,amssymb,graphicx,xcolor,latexsym,theorem}
\usepackage[a4paper,margin=1.8cm]{geometry}
\usepackage{hyperref}

\newcommand{\Kappa}{\mathrm{K}}
\newcommand{\assign}{:=}
\newcommand{\backassign}{=:}
\newcommand{\barsuchthat}{|}
\newcommand{\cdummy}{\cdot}
\newcommand{\longdownarrow}{{\mbox{\rotatebox[origin=c]{-90}{$\longrightarrow$}}}}
\newcommand{\mathd}{\mathrm{d}}
\newcommand{\nin}{\not\in}
\newcommand{\nobracket}{}
\newcommand{\nospace}{}

\newcommand{\tmcolor}[2]{{\color{#1}{#2}}}

\newcommand{\tmmathbf}[1]{\ensuremath{\boldsymbol{#1}}}
\newcommand{\tmop}[1]{\ensuremath{\operatorname{#1}}}
\newcommand{\tmscript}[1]{\text{\scriptsize{$#1$}}}
\newcommand{\tmtextbf}[1]{\text{{\bfseries{#1}}}}
\newcommand{\tmtextit}[1]{\text{{\itshape{#1}}}}

\newenvironment{proof}{\noindent\textbf{Proof\ }}{\hspace*{\fill}$\Box$\medskip}
\newcounter{nnacknowledgments}

{\theorembodyfont{\rmfamily}\newtheorem{acknowledgments*}[nnacknowledgments]{Acknowledgments}}
\newtheorem{theorem}{Theorem}[section]
\newtheorem{corollary}[theorem]{Corollary}
\newtheorem{definition}[theorem]{Definition}
{\theorembodyfont{\rmfamily}\newtheorem{example}[theorem]{Example}}
\newtheorem{lemma}[theorem]{Lemma}
\newtheorem{proposition}[theorem]{Proposition}
{\theorembodyfont{\rmfamily}\newtheorem{remark}[theorem]{Remark}}

\numberwithin{equation}{section}
\definecolor{amber}{rgb}{0.0, 0.26, 0.15}

\newcommand{\paragraphtoc}[1]{}
\title{Schauder estimates for germs of distributions on smooth manifolds}
\author{Beatrice Costeri \thanks{BC: Dipartimento di Fisica "Alessandro Volta'',
		Universit\`a degli Studi di Pavia \& INFN, Sezione di Pavia, 
		Via Bassi 6,
		I-27100 Pavia,
		Italia; Istituto Nazionale di Alta Matematica, Sezione di Pavia, via Ferrata 5, 27100 Pavia, Italia
		\mbox{beatrice.costeri01@universitadipavia.it}
	}
	\and Claudio Dappiaggi \thanks{CD: Dipartimento di Fisica ''Alessandro Volta'',
		Universit\`a degli Studi di Pavia \& INFN, Sezione di Pavia, 
		Via Bassi 6,
		I-27100 Pavia,
		Italia; Istituto Nazionale di Alta Matematica, Sezione di Pavia, via Ferrata 5, 27100 Pavia, Italia
		\mbox{claudio.dappiaggi@unipv.it}
	}
	\and
	Paolo Rinaldi \thanks{PR: Dipartimento di Fisica ''Alessandro Volta'',
		Universit\`a degli Studi di Pavia \& INFN, Sezione di Pavia, 
		Via Bassi 6,
		I-27100 Pavia,
		Italia; Istituto Nazionale di Alta Matematica, Sezione di Pavia, via Ferrata 5, 27100 Pavia, Italia
		\mbox{paolo.rinaldi@unipv.it}
	}
	\and
	Matteo Savasta \thanks{MS: Department of Mathematics, Heriot-Watt University, Edinburgh and Dipartimento di Matematica ''Federigo Enriques'', Università degli Studi di Milano, Via Saldini 50, 20123 Milano, Italia \mbox{ matteo.savasta@unimi.it, ms3031@hw.ac.uk}
}}

\begin{document}

\date{February 23, 2026}

\maketitle

\begin{abstract}
 We discuss germs of distributions on $d-$dimensional smooth Riemannian manifolds and, in particular, we derive \emph{multi-level Schauder estimates} without making any further assumptions on the underlying geometry. As a preliminary step, we define the notions of coherence and homogeneity for germs of distributions on open subsets of $\mathbb{R}^d$, $d \ge 1$. Subsequently, we formulate both the \emph{reconstruction theorem}, \textit{cf.}, \cite{CZ20}, and the Schauder estimates, \textit{cf.}, \cite{BCZ}, in this setting. Leveraging the properties of the exponential map, we extend these results to Riemannian manifolds. Specifically, we devise a counterpart of the reconstruction theorem previously established in the literature \cite{RS21}, while additionally proving the regularity of the reconstructed distribution in suitable H\"older-Zygmund spaces. Finally, by introducing a novel concept of $\beta$-regularizing kernels on Riemannian manifolds, we establish Schauder estimates for coherent and homogeneous germs in this context.
\end{abstract}

\

{\tableofcontents}

\section{Introduction}
\label{Introduction}

A major advance in the solution theory for singular stochastic partial differential equations (SPDEs) on the Euclidean space $\mathbb{R}^d$, $d \ge 1$, is the \emph{theory of regularity structures}, see \cite{Hai}. In this seminal paper, Hairer develops a highly abstract framework which allows to establish local existence and uniqueness of solutions to a broad class of stochastic PDEs, including the Kardar–Parisi–Zhang equation, {\it cf.}, \cite{KPZ}. Two cornerstones of Hairer’s regularity structures are the \emph{reconstruction theorem} and the \emph{multi-level Schauder estimates}, see \cite[Thms 1.3 \& 1.6]{Hai}. The latter, in particular, play a crucial r\^ole in the analysis of elliptic and parabolic PDEs, as they serve as an essential tool in fixed-point arguments. \\

Yet, despite its high level of generality, the theory of regularity structures does not incorporate the degree of concreteness required for establishing a close connection to several physically inspired models. Most notably, in sharp contrast to quantum field theory, {\emph{cf.}}, \cite{BFDY, BF09}, this approach does not provide a concrete mean to construct explicitly the solution and the correlation functions of a nonlinear singular SPDE. This motivated the development of a complementary approach, inspired by algebraic quantum field theory, {\it cf.}, \cite{FR} and designed to perform explicit, perturbative computations even when the underlying background is not flat, see \cite{DDRZ, BDR23, BDR24, BCD24}. Yet this novel method suffers of a severe limitation, namely the impossibility of proving convergence of the underlying perturbative series. Hence, it is highly desirable to establish a connection with the theory of regularity structures and, to this end, to develop counterparts of both the reconstruction theorem and Schauder estimates on generic backgrounds.

In order to achieve this goal, it is necessary to translate these two key results in a language better suited to a formulation on manifolds. This issue is naturally addressed by working with germs of distributions. Since these are intrinsically local objects, this formulation yields a framework that can be naturally extended to smooth manifolds. 
Specifically, given a {\it germ of distributions} on $\mathbb{R}^d$, namely, a family $(F_x)_{x\in \mathbb{R}^d} \subset \mathcal{D}'(\mathbb{R}^d)$, one might wonder whether there exists a globally defined distribution $f \in \mathcal{D}'(\mathbb{R}^d)$, such that, for all $x \in \mathbb{R}^d$,
\begin{equation*}
	f \approx F_x \;\;\text{in a neighborhood of}\;\; x,
\end{equation*}
where the approximation is to be understood in a suitable sense, see \cite{CZ20}. In this instance, the germ $(F_x)_{x\in \mathbb{R}^d}$ is said to be a \emph{local approximation} of $f$. For instance, if $f$ is generated by a $C^r$ function, $r \ge 0$, the associated germ is the Taylor polynomial of degree $r$ centered at a point $x$, see Example \ref{Taylor}.  
However, under milder regularity assumptions on the global distribution $f$, the existence and uniqueness of a reconstruction in terms of $(F_x)_{x \in \mathbb{R}^d}$ is far less trivial. Indeed, in \cite{CZ20}, it was proven that, under the additional requirement that the germ is {\emph{coherent}}, such a reconstruction exists. The counterpart of this result on smooth manifolds has been already discussed in \cite{RS21}, yet in this work we shall prove a version of the reconstruction theorem on Riemannian manifolds, see Theorem \ref{Thm: Reconstruction Theorem}, tailored to a slightly more general notion of germ. This will allow us to assess in addition the continuity of the reconstruction operator in suitable H\"older-Zygmund spaces, see Section \ref{ReconstructionisholderManifold}. 

More recently, in \cite{BCZ}, also multi-level Schauder estimates have been established within this distributional framework. These stem from the observation that the convolution between a distribution and a kernel that has an integrable singularity at the origin improves the regularity of the resulting object. Examples of kernels abiding by this property include the heat kernel as well as the Green's functions of several linear differential operators. Most notably, such a regularization property emerges also when working at the level of germs of distributions, {\it cf.}, \cite[Thm. 1.6]{Hai}
In this context, the main idea of Schauder estimates, is that, given a singular kernel $\mathsf{K}$ associated with a differential operator, there exists an integration map  $\mathcal{K}$ between germs such that, locally, the following approximation property holds true in a suitable sense: 
\begin{equation*}
	f \approx F_x \;\;\text{at}\;\; x\in \mathbb{R}^d \Rightarrow \mathsf{K}*f \approx \mathcal{K}F_x \;\;\text{at}\;\; x \in \mathbb{R}^d.
\end{equation*}

Expanding on these concepts, we establish the main results of this paper. These are codified in the following two theorems, concerning Schauder estimates for germs of distributions on Riemannian manifolds, see Theorems \ref{maintheoremonmanifolds} and \ref{maintheoremonmanifolds2} as well as Section \ref{sectionSchaudermanifold} for notations and conventions. 

\begin{theorem}[Main Theorem 1]
Let $\{\Omega_n\}_{n \in \mathbb{N}}$ be an
  exhaustion by compact sets of $M$ as per Definition \ref{Def: Compact Exhaustion} and let $\gamma, \beta > 0$. Let $F$ be an
  $(\tmmathbf{\alpha}, \gamma)$-coherent germ of distributions of order
  $\tmmathbf{r}$ and range $\tmmathbf{R}$, as per Definition \ref{Def:
  Coherence on M} and $\Kappa$ a $\beta$-regularising kernel of order $(m, r)$
  and range $\rho_n$ subordinated to $\{\Omega_n\}_{n \in \mathbb{N}}$ as per Definition
  \ref{regularisingkernelonmanifold}. Assume, moreover, that, $\rho_n$ is small enough, see Equation~\eqref{Eq:rhoN}, and that, for
  every compact set $K \subset M$,
  \begin{equation*}
    \alpha_K + \beta > 0,
    \hspace{0.27em} \hspace{0.27em} \hspace{0.27em} \gamma + \beta \nin
    \mathbb{N}_0, \hspace{0.27em} \hspace{0.27em} \hspace{0.27em} m > \gamma +
    \beta, \hspace{0.27em} \hspace{0.27em} \hspace{0.27em} r > - \alpha_K .
  \end{equation*}
  Then, there exists a germ of distributions, denoted $\mathcal{K}^{\gamma,
  \beta} F$, such that
  \begin{equation*}
     \mathcal{K}^{\gamma, \beta} F -
    \mathsf{K} (\mathcal{R}^{\gamma} F) \in \mathcal{G}^{\gamma + \beta,
    \tmmathbf{\alpha}'_{\gamma + \beta}}_{\bar{\tmmathbf{r}},
    \bar{\tmmathbf{R}}} (M),
  \end{equation*}
  \textit{i.e.}, $\mathcal{K}^{\gamma, \beta} F - \mathsf{K} (\mathcal{R}^{\gamma} F)$
  is $(\gamma + \beta)$-homogeneous and $(\tmmathbf{\alpha}', \gamma +
  \beta)$-coherent on $M$, of order $\tmmathbf{\bar{r}}$, with $\bar{r}_K
  \assign \lfloor - \alpha_K \rfloor$, $\alpha_K' \assign \alpha_K - d$. Here $\bar{R}_K$ is a suitable parameter codifying the range of the new
  germ. Furthermore, the following estimate holds
  \begin{equation*}
    \| \mathcal{K}F - \mathsf{K} (\mathcal{R}^{\gamma} F)
    \|_{\mathcal{G}^{\gamma + \beta, \tmmathbf{\alpha}'_{\gamma +
    \beta}}_{\bar{\tmmathbf{r}}, \bar{\tmmathbf{R}}} (M), K} \lesssim_n \| F
    \|_{\mathcal{G}^{\tmmathbf{\alpha }_{\gamma}}_{\tmmathbf{r },
    \tmmathbf{R}} (M), \overline{\Omega_{\tilde{n}}}} .   
  \end{equation*}
\end{theorem}

\begin{theorem}[Main Theorem 2]
Under the same assumptions of the previous Theorem, if $\mathsf{K} (\mathcal{R}^{\gamma} F)$ is a
  well-defined distribution in $\mathcal{D}' (M)$, then $\mathcal{K}^{\gamma,
  \beta} F$ is a well defined germ on $M$ and it holds that:
  \begin{itemize}
    \item[(i)] $\mathcal{K}^{\gamma, \beta} F \in \mathcal{G}^{\gamma + \beta,
    \tmmathbf{\alpha'_{\gamma + \beta}}}_{\tmmathbf{\bar{r}}, \tmmathbf{R'}}
    (M)$ and
    \begin{equation*}
      \mathcal{R}^{\gamma + \beta}
      (\mathcal{K}^{\gamma, \beta} F) = \mathsf{K} (\mathcal{R}^{\gamma} F),
    \end{equation*}
    on $M$, {\it i.e.}, $\mathsf{K} (\mathcal{R}^{\gamma} F)$ is the $(\gamma +
    \beta)-$reconstruction of the germ of distributions $\mathcal{K}^{\gamma,
    \beta} F$.
    
    \item[(ii)] If $F$ is $\bar{\tmmathbf{\alpha}}$-homogeneous, with
    \[ \bar{\alpha}_K \leqslant \gamma, \hspace{0.27em} \hspace{0.27em}
       \hspace{0.27em} \bar{\alpha}_K + \beta \neq 0, \]
    then $\mathcal{K}^{\gamma, \beta} F$ is $((\bar{\tmmathbf{\alpha}} +
    \beta) \wedge 0)-$homogeneous of a suitable order $\bar{\tmmathbf{r}}$,
    where $(\bar{\alpha}_K + \beta) \wedge 0 \assign \min \{ \bar{\alpha}_K +
    \beta, 0\}$.
    
    \item[(iii)] The map $F \mapsto \mathcal{K}^{\gamma, \beta} F$ is linear and
    continuous, both as
    \[ \mathcal{K}^{\gamma, \beta} :
       \mathcal{G}^{\tmmathbf{\alpha}_{\gamma}}_{\tmmathbf{r}, \tmmathbf{R}}
       (M) \to \mathcal{G}^{\tmmathbf{\alpha}'_{\gamma +
       \beta}}_{\bar{\tmmathbf{r}}, \overline{\tmmathbf{R}}} (M), \qquad
       \tmop{and} \qquad \mathcal{K}^{\gamma, \beta} :
       \mathcal{G}^{\tmmathbf{\bar{\alpha}},
       \tmmathbf{\alpha}_{\gamma}}_{\tmmathbf{r}} (M) \to
       \mathcal{G}^{(\tmmathbf{\bar{\alpha}} + \beta \wedge 0),
       \tmmathbf{\alpha}'_{\gamma + \beta}}_{\bar{\tmmathbf{r}},
       \overline{\tmmathbf{R}}} (M), \]
    abiding by the following continuity estimates: For a suitable compact $K'$
    such that $K \subset K'$ and for suitable parameters $\tmmathbf{\alpha}' 
    = \{ \alpha_K' \}_K$,
    \[ \lvert \lvert \mathcal{K}^{\gamma, \beta} F \rvert
       \rvert_{\mathcal{G}^{\tmmathbf{\alpha}'_{\gamma +
       \beta}}_{\bar{\tmmathbf{r}}, \overline{\tmmathbf{R}}} (M), K}
       \lesssim_n \hspace{0.17em} \lvert \lvert F \rvert
       \rvert_{\mathcal{G}^{\tmmathbf{\alpha}_{\gamma}}_{\tmmathbf{r},
       \tmmathbf{R}} (M), K'}, \qquad \lvert \lvert \mathcal{K}^{\gamma,
       \beta} F \rvert \rvert_{\mathcal{G}^{(\tmmathbf{\bar{\alpha}} + \beta
       \wedge 0), \tmmathbf{\alpha}'_{\gamma + \beta}}_{\bar{\tmmathbf{r}},
       \overline{\tmmathbf{R}}} (M), K} \lesssim_n \lvert \lvert F \rvert
       \rvert_{\mathcal{G}^{\tmmathbf{\bar{\alpha}},
       \tmmathbf{\alpha}_{\gamma}}_{\tmmathbf{r}} (M), K'}, \]
    where $n$ is the smallest integer such that $K \subset \Omega_n$. As a
    consequence the following diagram commutes
    \[ \begin{array}{ccc}
         \mathcal{G}^{\tmmathbf{\alpha}_{\gamma}}_{\tmmathbf{r}, \tmmathbf{R}}
         (M) & \xrightarrow{\mathcal{K}^{\gamma, \beta}} &
         \mathcal{G}^{\tmmathbf{\alpha}'_{\gamma +
         \beta}}_{\bar{\tmmathbf{r}}, \overline{\tmmathbf{R}}} (M)\\
         \longdownarrow \mathcal{R}^{\gamma} &  & \longdownarrow
         \mathcal{R}^{\gamma + \beta}\\
         \mathcal{D}' (M) & \xrightarrow{\Kappa} & \mathcal{D}' (M)
       \end{array} \]
  \end{itemize}
\end{theorem}

\noindent Both the above theorems can be proven as a consequence of Schauder estimates for germs on open subsets of $\mathbb{R}^d$. More precisely, exploiting the fact that the exponential map is a local diffeomorphism, one can work on open subsets of the tangent space and, then, pull back to the corresponding geodesically convex sets on the Riemannian manifold. To this avail, we develop novel tools to scale and re-center test functions on Riemannian manifolds, see Appendix \ref{AppendixC}. We mention that a result analogous to the two Main Theorems of our work has been obtained in \cite[Thm. 7.22]{HS23} in the context of {\it modeled distributions}. Yet, with respect to \cite{HS23}, we do not make any assumption on the structure of the kernel near the singular points of the manifold, {\it cf.}, \cite[Ass. 7.1, item 2.]{HS23}, and we do not impose any additional constraint on the underlying geometry. \\

\noindent \textbf{Outline of the paper}. The paper is organized as follows. In Section \ref{Sec: Germs of
distributions on open sets} we define germs of distributions on open subsets $U$ of $\mathbb{R}^d$, $d \ge 1$ and we introduce the notions of \emph{coherence} and \emph{homogeinity}, together with their weak counterparts. Section \ref{reconstruction theorem section} is devoted to the formulation of the \emph{reconstruction theorem} for $(\tmmathbf{\alpha}, \gamma)-$coherent germs on open sets, see Theorem \ref{Thm: Reconstruction Theorem}. \\
\noindent In Section \ref{Sec: Schauder Estimates Open Subsets}, we derive \emph{multi-level Schauder estimates} for germs on open sets, formulating one of the main results of this work, see Theorem \ref{Main Result 1}. \\
\noindent In Section \ref{Sec:Germs-of-distributions-manifold} we extend the theory of germs of distributions to Riemannian smooth
manifolds. Specifically, we show the equivalence of our definitions to the ones already established in the literature.\\
\noindent Goal of Section \ref{sectionReconstructionManifold} is to establish a \emph{reconstruction theorem} on Riemannian manifolds, see Theorem \ref{reconstructiononmanifolds}. In addition to existing literature, we prove the continuity of the
reconstruction map, assessing at the same time the regularity of the reconstructed
distribution in the sense of negative H{\"o}lder-Zygmund spaces on
manifolds, see Definition \ref{holderzyngmundManifold}. \\
\noindent In Section \ref{sectionSchaudermanifold}, we derive the two main results of our work, see Theorems \ref{maintheoremonmanifolds} and \ref{maintheoremonmanifolds2}, namely we establish Schauder estimates for germs on Riemannian manifolds, by merging the local
formulation of Section \ref{Sec: Schauder Estimates Open Subsets} and the
construction detailed in Sections~\ref{Sec:Germs-of-distributions-manifold} and \ref{sectionSchaudermanifold}.\\
\noindent In Appendix \ref{Sec: App A}, we discuss additional properties of homogeneous and coherent germs and we prove that these notions are actually connected under suitable assumptions.\\
\noindent Appendix \ref{appendixB} contains some technical results showing the topological equivalence in a local coordinate patch between the Riemannian and the Euclidean distance, see in particular Proposition \ref{topologicalequivalencenormalcoordinates}. \\
\noindent Finally, in Appendix \ref{AppendixC} we develop the tools needed for the scaling and re-centering of test functions on a Riemannian manifold.\\

\noindent \textbf{Notation and preliminaries}. In the first part of this work, we shall focus on connected, open subsets $U
\subseteq \mathbb{R}^d$ with $d \geqslant 1$ endowed with the standard
Euclidean coordinates $(x_1, ..., x_d)$. On top of $U$ we shall
consider notable standard functions spaces whose definitions are here
recollected to fix notations and conventions. \\
\noindent We denote by $\mathcal{D}
(U)$ the space of smooth and compactly supported test functions
endowed with the standard, locally convex Fr{\'e}chet topology and by
$\mathcal{D}' (U)$ its topological dual. Given any $\varphi \in \mathcal{D}
(U)$ and any $\lambda > 0$, we call \tmtextit{rescaled} test-function,
centered at $x \in U$,
\begin{equation}
  \label{Eq: Rescaled Test-Function} \varphi^{\lambda}_x (y) \assign
  \lambda^{- d} \varphi (\lambda^{- 1} (y - x)) . \quad \lambda \in (0, 1]
\end{equation}
In addition, we will consider $C^m (U)$, the space of $m$-times
differentiable functions supported in $U$, $m \in [0, \infty]$, regarded as a
Fr{\'e}chet space with respect to the topology of uniform convergence on
compact sets. If $m < \infty$ this is generated by the seminorms
\begin{equation}
  \label{Eq: C^m norm} \| \phi \|_{C^m (K)} \assign
  \max_{\tmscript{\begin{array}{c}
    \alpha \in \mathbb{N}^d_0\\
    | \alpha | \le m
  \end{array}}} \sup_{x \in K} | \partial^{\alpha} \phi (x) |,
\end{equation}
where $K \subset U$ is a compact subset, while $\mathbb{N}_0 \assign
\mathbb{N} \cup \{0\}$. In addition, denoting by $B (0, 1)$ the open ball in
$\mathbb{R}^d$ centered at the origin and of unit radius, we can identify the
following notable subspaces:

\begin{flalign}
  \mathfrak{B}^m & \assign \{\phi \in \mathcal{D}(B (0, 1)) \hspace{0.17em}
  \barsuchthat \hspace{0.17em} \| \phi \|_{C^m} \leqslant 1\} ;  \label{Eq:
  B^m}\\
  \mathfrak{B}_{\ell} & \assign \left\{ \phi \in \mathcal{D}(B (0, 1))
  \hspace{0.17em} \barsuchthat \hspace{0.17em} \int_{\mathbb{R}^d} x^k \phi
  (x) \mathd x = 0, \nospace \forall \hspace{0.17em} 0 \leqslant |k| \leqslant
  \ell \right\} ;  \label{Eq: B_l}\\
  \mathfrak{B}_{\ell}^m & \assign \left\{\begin{array}{l}
    \mathfrak{B}^m \cap \mathfrak{B}_{\ell}, \hspace{0.17em} \hspace{0.17em}
    \hspace{0.17em} \text{for} \hspace{0.17em} \hspace{0.17em} \ell \geqslant
    0\\
    \mathfrak{B}^m, \hspace{0.17em} \hspace{0.17em} \hspace{0.17em} \text{for}
    \hspace{0.17em} \hspace{0.17em} \ell < 0.
  \end{array}\right.  \label{Eq: B^m_l}
\end{flalign}

where $|| \cdot ||_{C^m}$ is the seminorm as per Equation \eqref{Eq: C^m norm}
setting $U = B (0, 1)$, while the role of $K$ is played by the support of
the test-function. At the same time $x^k$ is a shorthand notation for
$x_1^{k_1} \ldots x_d^{k_d}$ with $k_1 + \ldots + k_d = k$. Note that Equation
\eqref{Eq: B^m_l} enters in the definition of the H{\"o}lder-Zygmund
spaces, which we recall here for the reader's convenience -- see
{\cite[Def. 2.1]{BCZ}}, though we refer also to {\cite{Bahouri}} for additional
information. 

\begin{definition}[H{\"o}lder-Zygmund spaces]
  \label{Def: Holder-Zygmund}Given an open subset $U \subseteq \mathbb{R}^d$,
  we say that $f \in \mathcal{Z}^{\gamma} (U) \subset \mathcal{D}' (U)$ if,
  for all compact subsets $K \subset U$, there exists $\bar{\lambda} \in (0,
  D_K^U)$, $D^U_K \assign \mathrm{dist} (K, \partial U)$, such that $\| f
  \|_{\mathcal{Z}^{\gamma}_{K, \bar{\lambda}}} < \infty$, where
  \begin{equation}
    \label{holderboundeq12openset} \| f \|_{\mathcal{Z}^{\gamma}_{K,
    \bar{\lambda}}} = \left\{\begin{array}{ll}
      \sup_{\tmscript{\begin{array}{c}
        \lambda \in (0, \bar{\lambda}], x \in K\\
        \varphi \in \mathcal{B}^{\lfloor - \gamma + 1 \rfloor}
      \end{array}}} \frac{\lvert f (\varphi^{\lambda}_x)
      \rvert}{\lambda^{\gamma}} & \text{if } \gamma < 0\\
      \sup_{\tmscript{\begin{array}{c}
        x \in K\\
        \psi \in \mathcal{B}^0
      \end{array}}} \lvert f (\psi^{D_K^U}_x) \rvert +
      \sup_{\tmscript{\begin{array}{c}
        \lambda\in(0, \bar{\lambda}],x\in K\\
        \varphi\in\mathcal{B}^0_{\gamma}
      \end{array}}} \frac{\lvert f (\varphi^{\lambda}_x)
      \rvert}{\lambda^{\gamma}} & \text{if } \gamma \geq 0
    \end{array}\right. .
  \end{equation}
\end{definition}

\begin{remark}
Observe that, whenever $\gamma \nin \mathbb{Z}$, the H\"older-Zygmund space
$\mathcal{Z}^{\gamma} (U)$ coincides with that of H{\"o}lder-Besov
functions $C^{\alpha} (U) \equiv B^{\gamma}_{\infty, \infty} (U)$. 
\end{remark}

In the second part of this work, we shall confine our attention to a $d-$dimensional, connected  Riemannian manifold $M$, equipped with a smooth Riemannian metric $g$, which will be left implicitly understood. In the following, we shall assume several standard notions in differential geometry and refer the reader to {\cite{LeeS,LeeR}} for additional details. For the reader's convenience, we recall here only those key notions which will be used
extensively in our discussion.

\begin{definition}
\label{Rem: Good Covering} 
We say that $M$ admits a \textbf{finite good cover} if there exists a cover of $M$ made of open sets $\{U_i\}_{i
  \in I}$, where $I$ is a collection of indices, such that any non-empty
  finite intersection $U_{i_1} \cap \ldots \cap U_{i_n}$,
  $\{i_j \}_{j = 1, \ldots, n} \subset I$, is diffeomorphic to a
  contractible open subset of $\mathbb{R}^d$, $d = \dim M$. 
\end{definition}

\noindent We stress that every smooth
  manifold $M$ admits a \tmtextit{finite good cover} -- see {\cite[Prop. 13.3]{LeeS}}. Without loss of
  generality, we can choose each $U_{\alpha}$ in such a cover to be a geodesically convex open
  subset of $M$.

\begin{definition}
  \label{Def: Compact Exhaustion}Given a Riemannian manifold $M$, we call an
  \textbf{exhaustion by compact sets of $M$} a sequence of open sets $\{\Omega_n\}_{n \in
  \mathbb{N}} \subset M$, such that the following three conditions are met: 
  \begin{itemize}
      \item[(i)] the closure
  $\overline{\Omega_n}$ is a compact set for all $n \in \mathbb{N}$,
  \item[(ii)] $\bigcup_{n \in \mathbb{N}} \Omega_n
  = M$,
  \item[(iii)] $\Omega_n \subset \Omega_{n + 1}$ for all $ n \in \mathbb{N}$. 
  \end{itemize}
\end{definition}

\begin{definition}
  \label{Def: Convexity Radius}
Given a Riemannian manifold $M$, consider a subset $U \subseteq M$ and fix a point $p \in U$. Denoting by $B (0, R)$ the open ball of radius $R$ centered at the
  origin, define
  \[ \mathcal{R}_C (p) = \sup \{R > 0 \hspace{0.27em} | \hspace{0.27em} \exp_p
     (B (0, R)) \hspace{0.27em} \textrm{is a geodesically convex set} \} , \]
    where $\exp_p(\cdot)$ denotes the exponential map at $p$. 
   We call \tmtextbf{convexity radius} of $U$
  \[ \mathcal{R}_C (U) = \inf_{p \in U} \mathcal{R}_C (p). \]
\end{definition}

We stress that throughout the paper we shall use extensively
the notation $\lesssim$ to denote an inequality which holds true up to a
multiplicative constant depending on a set of underlying parameters which do
not play a role in the construction. When these constants might change
depending on the choice of a varying underlying compact subset $K$, we shall
employ the symbol $\lesssim_K$. More generally, whenever it is important to
highlight the dependence of an inequality on certain given parameters, we
shall attach to the symbol $\lesssim$ the corresponding subscripts.

\section{Germs of distributions on open sets}\label{Sec: Germs of
distributions on open sets}

In this section we review the construction first outlined in {\cite{RS21}},
based in turn on {\cite{CZ20}}, aimed at defining the notion of germs of
distributions on open subsets of $\mathbb{R}^d$, $d \geqslant 1$. First and foremost, we shall devise the counterpart in this setting of the notions of coherence and homogeinity of a germ, see Definitions \ref{Def: Coherence} and \ref{Def: Homogeneity}, then combined in Definition \ref{Def: Coherent and Homogeneous Germs}. In Section \ref{reconstruction theorem section}, we shall state one of the crucial results, namely the \emph{reconstruction theorem} for germs on open sets, see Theorem \ref{Thm: Reconstruction Theorem}. We shall conclude the section by investigating the regularity properties of the reconstruction, see Corollary \ref{propertiesreconstruction}. 

First of all, we will show that all relevant
structures and notions introduced in {\cite{CZ20}} are intrinsically local and
well-behaved under pull-back. 

\begin{definition}
  \label{Def: germ}A \tmtextbf{germ of distributions on $U \subset
  \mathbb{R}^d$} is a family $F \equiv (F_x)_{x \in U}$ of distributions $F_x
  \in \mathcal{D}' (U)$ for any $x \in U$, such that the map $x
  \mapsto F_x$ is measurable, \textit{i.e.}, 
  \[ \forall \varphi \in \mathcal{D} (U), \hspace{0.27em} \hspace{0.27em}
     \hspace{0.27em} x \to F_x (\varphi) \hspace{0.27em} \hspace{0.27em}
     \hspace{0.27em} \text{is measurable} . \]
\end{definition}

{\noindent}The rationale behind Definition \ref{Def: germ} is to read $F_x$ as
the local approximation at each $x \in U$ of a globally defined distribution
$\mathcal{R} (F)$ which can be recovered by means of the
\tmtextit{reconstruction theorem} -- see Section \ref{sectionReconstructionManifold}. Yet, this requires a further
specialization.

\begin{definition}[Coherence]
  \label{Def: Coherence}Under the hypotheses of Definition \ref{Def: germ},
  let $\gamma \in \mathbb{R}$ and consider two families
  $\tmmathbf{\alpha}_{\gamma} \assign (\alpha_K)_{K \subset U} \subset
  \mathbb{R}$, $\tmmathbf{r} \assign (r_K)_{K \subset U} \subset
  \mathbb{N}_0$, indexed by compact sets $K \subset U$, with $\alpha_K
  \leqslant \min \{0, \gamma\}, \forall K$. F is called
  \begin{itemize}
    \item {\underline{$(\tmmathbf{\alpha}, \gamma)$-coherent of order
    $\tmmathbf{r}$}} if, for any compact $K \subset U$, and for any
    $\bar{\lambda} \in (0, D^U_K)$, $D_K^U \assign \mathrm{dist} (K, \partial
    U)$,
    \begin{equation}
      \label{Eq: Coherent} \lvert (F_x - F_y) (\varphi^{\lambda}_x) \rvert
      \lesssim \lambda^{\alpha_K}  (\lvert x - y \rvert + \lambda)^{\gamma -
      \alpha_K},
    \end{equation}
    uniformly over $x, y \in K$, $\lambda \in (0, \bar{\lambda}]$, $\varphi \in
    \mathcal{B}^{r_K}$. Here $\varphi^{\lambda}_x$ and $\mathcal{B}^{r_K}$ are
    defined as per Equations \eqref{Eq: Rescaled Test-Function} and \eqref{Eq:
    B^m}, respectively. The space of $(\tmmathbf{\alpha}, \gamma)$-coherent germs of order
    $\tmmathbf{r}$, is denoted by
    $\mathcal{G}^{\tmmathbf{\alpha}_{\gamma}}_{\tmmathbf{r}} (U)$ and it is
    endowed with the family of seminorms
    \begin{equation}
      \label{Eq: Coherence Seminorm}
      \mathcal{G}^{\tmmathbf{\alpha}_{\gamma}}_{\tmmathbf{r}} (U) \ni F
      \mapsto \| F \|_{\mathcal{G}^{\tmmathbf{\alpha}_{\gamma}}_{\tmmathbf{r}}
      (U), K} \doteq \sup_{\tmscript{\begin{array}{c}
        \lambda \in (0, \bar{\lambda}]\\
        x, y \in K\\
        \varphi \in \mathcal{B}^{r_K}
      \end{array}}}  \frac{\lvert (F_x - F_y) (\varphi^{\lambda}_x)
      \rvert}{\lambda^{\alpha_K}  (\lambda + \lvert x - y \rvert)^{\gamma -
      \alpha_K}},
    \end{equation}
    \item {\underline{weakly $(\tmmathbf{\alpha}, \gamma)$-coherent of order
    $\tmmathbf{r}$}}, denoted by $F \in
    \check{\mathcal{G}}^{\tmmathbf{\alpha}_{\gamma}}_{\tmmathbf{r}} (U)$, if
    for any compact subset $K \subset U$ and for any $\bar{\lambda} \in (0,
    D^U_K)$, 
    \begin{equation}
      \label{Eq: Weakly Coherent} \lvert (F_x - F_y) (\psi_x^{D^U_K}) \rvert
      \lesssim 1 \quad \mathrm{and} \quad \lvert (F_x - F_y)
      (\varphi^{\lambda}_x) \rvert \lesssim \lambda^{\alpha_K} (\lvert x - y
      \rvert + \lambda)^{\gamma - \alpha_K},
    \end{equation}
    uniformly over $x, y \in K$, $\psi \in \mathcal{B}^{r_K}$, $\varphi \in
    \mathcal{B}^{r_K}_{\gamma}$ and $\lambda \in (0, \bar{\lambda}]$. Here
    $\psi_x^{D^U_K}$ and $\varphi^{\lambda}_x$ are as in Equation \eqref{Eq:
    Rescaled Test-Function}, whilst for $r \in \mathbb{N}_0$, $l \in
    \mathbb{R}$, the spaces $\mathcal{B}^r_l$ are as per Equation \eqref{Eq:
    B^m_l}. The set
    $\check{\mathcal{G}}^{\tmmathbf{\alpha}_{\gamma}}_{\tmmathbf{r}} (U)$ is
    equipped with the following family of seminorms
    \begin{equation}
      \label{Eq: Weak Coherence Seminorm}
      \check{\mathcal{G}}^{\tmmathbf{\alpha}_{\gamma}}_{\tmmathbf{r}} (U) \ni
      F \mapsto \| F
      \|_{\check{\mathcal{G}}^{\tmmathbf{\alpha}_{\gamma}}_{\tmmathbf{r}} (U),
      K} \doteq \sup_{\tmscript{\begin{array}{c}
        x, y \in K\\
        \psi \in \mathcal{B}^{r_K}_{\gamma}
      \end{array}}} \lvert (F_x - F_y) (\psi_x^{D^U_K}) \rvert +
      \sup_{\tmscript{\begin{array}{c}
        \lambda \in (0, \bar{\lambda}], \hspace{0.27em} x, y \in K\\
        \varphi \in \mathcal{B}^{r_K}
      \end{array}}}  \frac{\lvert (F_x - F_y) (\varphi^{\lambda}_x)
      \rvert}{\lambda^{\alpha_K}  (\lambda + \lvert x - y \rvert)^{\gamma -
      \alpha_K}} .
    \end{equation}
  \end{itemize}
\end{definition}

\begin{remark}
  As discussed in {\cite{RS21}}, the introduction of the parameter $D_K^U
  \assign \mathrm{dist} (K, \partial U)$ is necessary since, after the rescaling of the test-function,
  $\tmop{supp} (\varphi^{\lambda}_x) \subset B (x, \lambda)$. In order for $F_x (\varphi_x^{\lambda})$ to be meaningful we need that
  $\tmop{supp} (\varphi^{\lambda}_x) \subset B (x, \lambda) \subset U$, which is
  granted by the above construction.
\end{remark}

\begin{remark}
  Given a germ $F$ as per Definition \ref{Def: germ}, it is equivalent to
  state that it is $(\tmmathbf{\alpha}, \gamma)$-coherent of order
  $\tmmathbf{r}$ as per Definition \ref{Def: Coherence} or that, for every compact set, $K \subset U$, $\bar{\lambda}
  \in (0, D_K^U)$,
  \begin{equation}
    \label{coherencereformulation} \lvert (F_x - F_y) (\varphi^{\lambda}_x)
    \rvert \lesssim_K \left\{\begin{array}{ll}
      \lambda^{\alpha_K} \hspace{0.27em} \hspace{0.27em} \hspace{0.27em} &
      \text{if} \hspace{0.27em} \hspace{0.27em} \lambda \leqslant \lvert x - y
      \rvert\\
      \lambda^{\gamma} \hspace{0.27em} \hspace{0.27em} \hspace{0.27em} &
      \text{if} \hspace{0.27em} \hspace{0.27em} \lambda > \lvert x - y \rvert
    \end{array}\right.,
  \end{equation}
  uniformly for $x, y \in K$, $\varphi \in \mathcal{B}^{r_K}$, $\lambda \in (0, \bar{\lambda}]$.
\end{remark}

{\noindent}In addition to coherence it is convenient to introduce a second property of germs which goes under the name of \textit{homogeneity}.

\begin{definition}[Homogeneity]
  \label{Def: Homogeneity} In the same setting of Definition \ref{Def: germ},
  given two families $\bar{\tmmathbf{\alpha}} \doteq (\bar{\alpha}_K)_{K
  \subset U} \subset \mathbb{R}$, and $\tmmathbf{r} \doteq (r_K)_{K \subset U}
  \subset \mathbb{N}_0$, indexed by compact sets $K \subset U$, we say that $F$ is
  \begin{itemize}
    \item {\underline{$\tmmathbf{\bar{\alpha}}$-homogeneous of order
    $\tmmathbf{r}$}}, if for any $K \subset U$, and for any $\bar{\lambda} \in
    (0, D_K)$ with $D_K^U \assign \mathrm{dist} (K, \partial U)$,
    \begin{equation}
      \label{Eq: Homogeneity} \lvert F_x (\varphi^{\lambda}_x) \rvert \lesssim
      \lambda^{\bar{\alpha}_K},
    \end{equation}
    uniformly over $x \in K$, $\varphi \in \mathcal{B}^{r_K}$ and $\lambda \in
    (0, \bar{\lambda}]$. The space of $\tmmathbf{\bar{\alpha}}$-homogeneous
    germs of order $\tmmathbf{r}$ over $U$ is denoted by
    $\mathsf{G}^{\tmmathbf{\bar{\alpha}}}_{\tmmathbf{r}} (U)$ and it is
    endowed with the family of seminorms
    \begin{equation}
      \label{Eq: Homogeneity norm}
      \mathsf{G}^{\tmmathbf{\bar{\alpha}}}_{\tmmathbf{r}} (U) \ni F \mapsto \|
      F \|_{\mathsf{G}^{\tmmathbf{\bar{\alpha}}}_{\tmmathbf{r}} (U), K}
      \assign \sup_{\tmscript{\begin{array}{c}
        \lambda \in (0, \bar{\lambda}] \hspace{0.27em}\\
        x \in K\\
        \varphi \in \mathcal{B}^{r_K}
      \end{array}}}  \frac{\lvert F_x (\varphi^{\lambda}_x)
      \rvert}{\lambda^{\bar{\alpha}_K}},
    \end{equation}
    \item {\underline{weakly $\tmmathbf{\bar{\alpha}}$-homogeneous of order
    $\tmmathbf{r}$}}, denoted by $F \in
    \check{\mathsf{G}}^{\tmmathbf{\bar{\alpha}}}_{\tmmathbf{r}} (U)$, if, for
    any compact subset $K \subset U$ and for any $\bar{\lambda} \in (0,
    D^U_K)$,
    \begin{equation}
      \label{Eq: Weakly Homogeneous} \lvert F_x (\psi_x^{D^U_K}) \rvert
      \lesssim 1 \quad \mathrm{and} \quad \lvert F_x (\varphi^{\lambda}_x)
      \rvert \lesssim \lambda^{\bar{\alpha}_K},
    \end{equation}
    uniformly over $x \in K$, $\psi \in \mathcal{B}^{r_K}$, $\varphi \in
    \mathcal{B}^{r_K}_{\bar{\alpha}_K}$ and $\lambda \in (0, \bar{\lambda}]$.
    Here $\psi_x^{D^U_K}$ and $\varphi^{\lambda}_x$ are as per Equation
    \eqref{Eq: Rescaled Test-Function}, while the spaces $\mathcal{B}^{r_K}$  and $\mathcal{B}^{r_K}_{\bar{\alpha}_K}$ are as per Equations \eqref{Eq: B^m} and \eqref{Eq:
    B^m_l}, respectively. We equip
    $\check{\mathsf{G}}^{\tmmathbf{\bar{\alpha}}}_{\tmmathbf{r}} (U)$ with the
    family of seminorms
    \begin{equation}
      \label{Eq: Weak Homogeneity norm}
      \check{\mathsf{G}}^{\tmmathbf{\bar{\alpha}}}_{\tmmathbf{r}} (U) \ni F
      \mapsto \| F
      \|_{\check{\mathsf{G}}^{\tmmathbf{\bar{\alpha}}}_{\tmmathbf{r}} (U), K}
      \assign \sup_{\tmscript{\begin{array}{c}
        \hspace{0.27em} x \in K\\
        \psi \in \mathcal{B}^{r_K}
      \end{array}}} \lvert F_x (\psi_x^{D^U_K}) \rvert +
      \sup_{\tmscript{\begin{array}{c}
        \lambda \in (0, \bar{\lambda}], \hspace{0.27em} x \in K\\
        \varphi \in \mathcal{B}^{r_K}_{\bar{\alpha}_K}
      \end{array}}}  \frac{\lvert F_x (\varphi^{\lambda}_x)
      \rvert}{\lambda^{\bar{\alpha}_K}} .
    \end{equation}
  \end{itemize}
\end{definition}

\begin{remark}
  \label{negativehomogeneity}Note that, under the hypotheses of Definition
  \ref{Def: Homogeneity}, if $F$ has positive homogeneity, that is if $\bar{\alpha}_K > 0$ for all $K \subset U$, then
  \[ \lim_{\lambda \to 0^+} \lvert F_x (\varphi^{\lambda}_x) \rvert = 0. \]
  This scenario is not of interest for the forthcoming construction and,
  therefore, we shall focus mainly on the cases where $\bar{\alpha}_K
  \leqslant 0$.
\end{remark}

{\noindent}At this point, we can merge together Definitions \ref{Def: Coherence} and \ref{Def: Homogeneity} to devise a notion of coherent and homogeneous germs on $U$. 

\begin{definition}[Coherent and homogeneous germs]
  \label{Def: Coherent and Homogeneous Germs}Under the hypotheses of
  Definition \ref{Def: germ}, a germ $F$ is called
  \begin{itemize}
    \item \underline{$(\tmmathbf{\alpha}, \gamma)$-coherent with homogeneity
    $\tmmathbf{\bar{\alpha}}$ of order $\tmmathbf{r}$} if Equations \eqref{Eq:
    Coherent} and \eqref{Eq: Homogeneity} hold true. The space of such germs
    is denoted $\mathcal{G}^{\tmmathbf{\bar{\alpha}},
    \tmmathbf{\alpha}_{\gamma}}_{\tmmathbf{r}} (U) \assign
    \mathsf{G}^{\tmmathbf{\bar{\alpha}}}_{\tmmathbf{r}} (U) \cap
    \mathcal{G}^{\tmmathbf{\alpha}_{\gamma}}_{\tmmathbf{r}} (U)$. This space
    is endowed with the family of seminorms
    \begin{equation}
      \label{Eq: Full Seminorms} \mathcal{G}^{\tmmathbf{\bar{\alpha}},
      \tmmathbf{\alpha}_{\gamma}}_{\tmmathbf{r}} (U) \ni F \mapsto \| F
      \|_{\mathcal{G}^{\tmmathbf{\bar{\alpha}},
      \tmmathbf{\alpha}_{\gamma}}_{\tmmathbf{r}} (U), K} \assign \| F
      \|_{\mathcal{G}^{\tmmathbf{\alpha}_{\gamma}}_{\tmmathbf{r}} (U), K} + \|
      F \|_{\mathsf{G}^{\tmmathbf{\bar{\alpha}}}_{\tmmathbf{r}} (U), K} .
    \end{equation}
    \item \underline{weakly $(\tmmathbf{\alpha}, \gamma)$-coherent with homogeneity
    $\tmmathbf{\bar{\alpha}}$ of order $\tmmathbf{r}$} if Equations \eqref{Eq:
    Weakly Coherent} and \eqref{Eq: Weakly Homogeneous} hold true. The space of
    such germs is denoted $\check{\mathcal{G}}^{\tmmathbf{\bar{\alpha}},
    \tmmathbf{\alpha}_{\gamma}}_{\tmmathbf{r}} (U) =
    \check{\mathcal{G}}^{\tmmathbf{\alpha}_{\gamma}}_{\tmmathbf{r}} (U) \cap
    \check{\mathsf{G}}^{\tmmathbf{\bar{\alpha}}}_{\tmmathbf{r}} (U)$ and it is
    endowed with the family of seminorms
    \begin{equation}
      \label{Eq: Full Weakly Seminorms}
      \check{\mathcal{G}}^{\tmmathbf{\bar{\alpha}},
      \tmmathbf{\alpha}_{\gamma}}_{\tmmathbf{r}} (U) \ni F \mapsto \| F
      \|_{\check{\mathcal{G}}^{\tmmathbf{\bar{\alpha}},
      \tmmathbf{\alpha}_{\gamma}}_{\tmmathbf{r}} (U), K} \assign \| F
      \|_{\check{\mathsf{G}}^{\tmmathbf{\bar{\alpha}}}_{\tmmathbf{r}} (U), K}
      + \| F
      \|_{\check{\mathcal{G}}^{\tmmathbf{\alpha}_{\gamma}}_{\tmmathbf{r}} (U),
      K} .
    \end{equation}
  \end{itemize}
\end{definition}

\begin{remark}
  \label{weakarestrong} Observing that the difference between coherence (\textit{resp}. homogeneity) and weak
  coherence (\textit{resp}. weak homogeneity) only concerns the classes of test functions
  involved, if for all compact sets $K \subset U$, it holds that
  $\bar{\alpha}_K < 0$, then
  $\check{\mathsf{G}}^{\tmmathbf{\bar{\alpha}}}_{\tmmathbf{r}} (U) =
  \mathsf{G}^{\tmmathbf{\bar{\alpha}}}_{\tmmathbf{r}} (U)$ (\textit{resp}.,  whenever
  $\gamma < 0$,
  $\check{\mathcal{G}}^{\tmmathbf{\alpha}_{\gamma}}_{\tmmathbf{r}} (U)
  =\mathcal{G}^{\tmmathbf{\alpha}_{\gamma}}_{\tmmathbf{r}} (U)$). This entails
  that the notions of weak homogeneity and weak coherence become relevant only if $\bar{\alpha}_K
  \geqslant 0$ and $\gamma \geqslant 0$, respectively.
\end{remark}

{\noindent}It is worth discussing two notable examples of germs. Of particular
relevance is the second one since, heuristically speaking, germs could be
thought as the counterpart at the level of distributions of what Taylor series
are for functions of positive regularity.

\begin{example}[Constant germ]
  Let $U \subseteq \mathbb{R}^d$ be an open set and consider a distribution $f \in
  \mathcal{Z}^{\bar{\alpha}} (U)$, as per Definition \ref{Def:
  Holder-Zygmund}. Set $F_x \assign f$ for any $x \in U$. Then, for any $\alpha,
  \gamma \in \mathbb{R}$, the constant germ $F$ is $(\alpha, \gamma)$-coherent as per
  Definition \ref{Def: Coherence} and $\bar{\alpha}$-homogeneous as per
  Definition \ref{Def: Homogeneity}.
\end{example}

\begin{example}[Taylor Polynomials]
  \label{Taylor}Let $U \subseteq \mathbb{R}^d$, and let $\gamma > 0$. Given a
  function $f \in \mathcal{C}^{\gamma} (U)$, for all compact sets $K \subset
  U$,
  \[ \lvert f (y) -\mathcal{T}_x^{\gamma} f (y) \rvert \lesssim_K \lvert x - y
     \rvert^{\gamma}, \]
  where $\mathcal{T}_x^{\gamma} f$ is the Taylor polynomial of $f$ at $x$, \textit{i.e.}, 
  \begin{equation}
    \label{Eq: Taylor1} \mathcal{T}^{\gamma}_x f (y) \assign \sum_{0 \leqslant
    \lvert k \rvert < \lfloor \gamma \rfloor} (\partial^k f) (x)  \frac{(y -
    x)^k}{k!},
  \end{equation}
  where $k = (k_1 \ldots k_d)$ is a multi-index with $|k| = \sum_{i = 1}^d
  k_i$. We prove that $\mathcal{T}_x^{\gamma} f$ is a $(0, \gamma)$-coherent
  germ of order $0$ and that, for all compact set $K \subset U$,
  \[ \| \mathcal{T}^{\gamma}_x f
     \|_{\mathcal{G}^{\tmmathbf{0}_{\gamma}}_{\tmmathbf{0}} (U), K} \lesssim
     \| f \|_{\mathcal{C}^{\gamma} (U)}, \]
  where the seminorms above are defined in Equations \eqref{Eq: C^m norm} and \eqref{Eq: Coherence Seminorm}. Fixing an arbitrary compact set $K \subset U$, for $0
  \leqslant \lvert k \rvert \leqslant \gamma$, $\partial^k f \in
  \mathcal{C}^{\gamma - \lvert k \rvert} (U)$. If we consider the Taylor
  expansion of $\partial^k f (x)$ on the right hand side of Equation
  \eqref{Eq: Taylor1} around a point $z \in K$
  \begin{equation}
    \label{Taylor12} \mathcal{T}_x^{\gamma} f (y) = \sum_{\lvert k \rvert <
    \gamma} \left( \sum_{\lvert l \rvert \leqslant \gamma - \lvert k \rvert}
    \partial^{k + l} f (z) \frac{(x - z)^l}{l!} + R^k (x, z) \right) \frac{(y
    - x)^k}{k!},
  \end{equation}
  where $R^k (x, z)$ is the Lagrange remainder. Standard results in calculus
  entail that
  \begin{equation}
    \label{Eq: Taylor13} \partial^k f (x) -\mathcal{T}_z^{\gamma - \lvert k
    \rvert} \partial^k f (x) = R^k (x, z) \hspace{0.27em} \hspace{0.27em}
    \hspace{0.27em} \text{with} \hspace{0.27em} \hspace{0.27em}
    \hspace{0.27em} \lvert R^k (x, z) \rvert \leqslant \| f
    \|_{\mathcal{C}^{\gamma} (U)}  \lvert x - z \rvert^{\gamma - \lvert k
    \rvert},
  \end{equation}
  uniformly for $x,z\in K$. Switching from $l$ to $k' \assign k + l$ we can
  rewrite Equation \eqref{Taylor12} as
  \begin{equation}
       \mathcal{T}_x^{\gamma} f (y)  = \sum_{\lvert k' \rvert < \gamma}
       \partial^{k'} f (z) \left( \sum_{k \leqslant k'} \frac{(x - z)^{k' -
       k}}{(k' - k) !}  \frac{(y - x)^k}{k!} \right) +\\ \sum_{\lvert k \rvert <
       \gamma} R^k (x, z) \frac{(y - x)^k}{k!}\\
        =\mathcal{T}_z^{\gamma} f (y) + \sum_{\lvert k \rvert < \gamma} R^k
       (x, z) \frac{(y - x)^k}{k!} .
     \end{equation}
  As a consequence,
  \begin{equation}
    \label{Eq: Taylor14} \mathcal{T}_z^{\gamma} f (y) -\mathcal{T}_x^{\gamma}
    f (y) = - \sum_{\lvert k \rvert < \gamma} R^k (x, z) \frac{(y - x)^k}{k!},
  \end{equation}
  and, exploiting Equation \eqref{Eq: Taylor13}, it follows that
  \[ \lvert \mathcal{T}_z^{\gamma} f (y) -\mathcal{T}_x^{\gamma} f (y) \rvert
     \leqslant \| f \|_{\mathcal{C}^{\gamma} (U)}  \sum_{\lvert k \rvert <
     \gamma} \lvert y - x \rvert^{\lvert k \rvert}  \lvert x - z
     \rvert^{\gamma - \lvert k \rvert} . \]
  Therefore, for any test function $\varphi \in \mathcal{B}^0$, see Equation
  \eqref{Eq: B^m}, it descends that
  \begin{equation}
    \label{Eq: Taylor2} \left| \int_{\mathbb{R}^d} (\mathcal{T}_z^{\gamma} f
    (y) -\mathcal{T}_x^{\gamma} f (y)) \varphi^{\lambda}_x (y) \mathd y
    \right| \leqslant \| f \|_{\mathcal{C}^{\gamma} (U)} \| \varphi \|_{C^0
    (U)}  \sum_{n < \gamma} \lvert z - x \rvert^{\gamma - n} \lambda^n,
  \end{equation}
  uniformly for $x, z \in K$ with $\lambda \in (0, D_K^U)$, see Definition
  \ref{Def: Homogeneity}. This entails the sought conclusion.
\end{example}

\subsection{The reconstruction theorem}\label{reconstruction theorem section}

Goal of this section is to present one of the main results concerning the theory of germs
of distributions, the \emph{reconstruction theorem}, here specialized to open sets $U$ of $\mathbb{R}^d$, $d \ge 1$. We refer the reader to {\cite{CZ20}} for the statement of this theorem on $\mathbb{R}^d$, $d \geqslant 1$ and to {\cite{RS21}} for its generalization to Riemannian manifolds.
Henceforth, for the sake of conciseness, we shall omit all proofs which can be found in the aforementioned references. 

For later convenience, we give the following preliminary definitions. Given a compact subset $K \subset U$ and $\epsilon > 0$ we call
\tmtextit{$\epsilon$-enlargement} of $K$ the set
\begin{equation}
  \label{Eq: K-enlargement} K_{\epsilon} \assign K + \overline{B (0,
  \epsilon)},
\end{equation}
where $B (0, \epsilon)$ is the open ball centered at the origin of radius
$\epsilon$. In addition, we set
\begin{equation}
  \label{Eq: Diameter of K} \mathrm{Diam} (K) = \max_{x, y \in K} |x - y| .
\end{equation}

We can now state the crucial result of this section. 

\begin{theorem}[Reconstruction on open sets]
  \label{Thm: Reconstruction Theorem}
  Let $U \subseteq \mathbb{R}^d$, $d\ge1$ be an open set and let $F \in
  \mathcal{G}^{\tmmathbf{\alpha}_{\gamma}}_{\tmmathbf{r}} (U)$ be an $(\tmmathbf{\alpha}, \gamma)-$coherent germ of order $\tmmathbf{r}$ on $U$ as per
  Definition \ref{Def: Coherence}. Given a compact set $K \subset
  U$, we denote by
  \begin{equation}
    \label{mkkk} \tilde{K} \assign K_{3 D_K^U / 4}, \hspace{0.27em}
    \hspace{0.27em} \hspace{0.27em} m_K \assign \max \{\lfloor -
    \alpha_{\tilde{K}} + 1 \rfloor, r_{\tilde{K}} + d + 1\},
  \end{equation}
  where $K_{3 D_K^U / 4}$ is as per Equation \eqref{Eq: K-enlargement}. There
  exists $\mathcal{R}^{\gamma} F \in \mathcal{D}' (U)$, called the
  \textbf{reconstruction of $F$}, such that, for every compact set $K \subset U$ and for
  every $\bar{\lambda} \in (0, D_K^U)$,
  \begin{equation}
    \label{Eq: Reconstruction Theorem} \lvert (\mathcal{R}^{\gamma} F - F_z)
    (\varphi^{\lambda}_z) \rvert \lesssim_{\alpha_{\tilde{K}}, r_{\tilde{K}},
    \bar{\lambda}} \| F
    \|_{\mathcal{G}^{\tmmathbf{\alpha}_{\gamma}}_{\tmmathbf{r}} (U),
    \tilde{K}} \left\{\begin{array}{ll}
      \lambda^{\gamma} \hspace{0.27em} \hspace{0.27em} \hspace{0.27em} &
      \text{for} \hspace{0.27em} \hspace{0.27em} \hspace{0.27em} \gamma \neq
      0\\
      \log (1 + \frac{1}{\lambda}) \hspace{0.27em} \hspace{0.27em}
      \hspace{0.27em} & \text{for} \hspace{0.27em} \hspace{0.27em}
      \hspace{0.27em} \gamma = 0
    \end{array}\right.
  \end{equation}
  uniformly for $x \in K$, $\lambda \in (0, \bar{\lambda})$, $\varphi \in
  \mathcal{B}^{r_K}$, see Equations \eqref{Eq: Rescaled
  Test-Function} and \eqref{Eq: B^m}. In addition,
  \begin{enumerate}
    \item if $\gamma > 0$, the reconstruction is the unique distribution in
    $\mathcal{D}' (U)$, such that
    \[ \lim_{\lambda \to 0^+} \lvert (\mathcal{R}^{\gamma} F - F_z)
       (\varphi^{\lambda}_z) \rvert = 0. \]
    \item for any $\gamma \in \mathbb{R}$, the
    reconstruction can always be chosen in such a way that $\mathcal{R}^{\gamma} :
    \mathcal{G}^{\tmmathbf{\alpha}_{\gamma}}_{\tmmathbf{r}} (U) \to
    \mathcal{D}' (U)$ is linear.
  \end{enumerate}
\end{theorem}

\noindent Barring minor modifications,
the proof of Theorem \ref{Thm: Reconstruction Theorem} follows slavishly from that of {\cite[Thm. 5.1]{CZ20}}, hence we omit it.

\begin{remark}
  \label{reconstructionishomogeneous1}A direct consequence of Theorem \ref{Thm:
  Reconstruction Theorem} is that, if $\gamma \neq 0$, on a compact set
  $K \subset U$, the germ $F -\mathcal{R}^{\gamma} F$ is $\gamma$-homogeneous
  of order $m_K$, see Equation \eqref{mkkk}, as well as $(\tmmathbf{\alpha},
  \gamma)$-coherent of order $r_K$. On account of {\cite[Prop. B.1]{BCZ}}, it
  follows that $F -\mathcal{R}^{\gamma} F$ is coherent and homogeneous of
  order $\bar{r}_K \assign \lfloor \max \{- \alpha_K, - \gamma\}+ 1 \rfloor =
  \lfloor - \alpha_K \rfloor$. Comparing Equation \eqref{Eq: Reconstruction
  Theorem} with Equation \eqref{Eq: Homogeneity}, it descends that
  \begin{equation}
    \| F -\mathcal{R}^{\gamma} F
    \|_{\mathsf{G}^{\tmmathbf{\bar{\alpha}}}_{\tmmathbf{r}} (U), K}
    \lesssim_{\alpha_{\tilde{K}}, r_{\tilde{K}}, \bar{\lambda}} \| F
    \|_{\mathcal{G}^{\tmmathbf{\alpha}_{\gamma}}_{\tmmathbf{r}} (U),
    \tilde{K}}, \label{reconstructionishomogeneouseq}
  \end{equation}
  where $\tilde{K}$ is as per Equation \eqref{mkkk}. Suppose now that
  $F$ is $\bar{\tmmathbf{\alpha}}$-homogeneous, with $\bar{\alpha}_K \leqslant
  \gamma$ for all compact sets $K \subset U$. Putting together Equations
  \eqref{Eq: Homogeneity} and \eqref{reconstructionishomogeneouseq}, and using
  the triangle inequality, it follows that the constant germ
  $\mathcal{R}^{\gamma} F$ satisfies the following homogeneity estimate
  \begin{equation}
    \| \mathcal{R}^{\gamma} F
    \|_{\mathsf{G}^{\tmmathbf{\bar{\alpha}}}_{\bar{\tmmathbf{r}}} (U), K}
    \lesssim_{\alpha_{\tilde{K}}, r_{\tilde{K}}, \bar{\lambda}} \| F
    \|_{\mathcal{G}^{\tmmathbf{\alpha}_{\gamma}}_{\tmmathbf{r}} (U),
    \tilde{K}} + \| F \|_{\mathsf{G}^{\tmmathbf{\bar{\alpha}}}_{\tmmathbf{r}}
    (U), K} \leqslant \| F \|_{\mathcal{G}^{\tmmathbf{\bar{\alpha}},
    \tmmathbf{\alpha}_{\gamma}}_{\tmmathbf{r}} (U), \tilde{K}} .
    \label{reconstructionishomogeneouseq1}
  \end{equation}
\end{remark}

{\noindent}At last, we establish a result on the H{\"o}lder regularity
of the reconstruction defined as per Theorem \ref{Thm: Reconstruction Theorem}.

\begin{corollary}[H{\"o}lder regularity of the reconstruction]
  \label{propertiesreconstruction}Let $U \subseteq \mathbb{R}^d$ and let $F =
  (F_x)_{x \in U}$ be a germ as per Definition \ref{Def: germ}.
  \begin{enumerate}
    \item If $F \in \mathcal{G}^{\tmmathbf{\bar{\alpha}},
    \tmmathbf{\alpha}_{\gamma}}_{\tmmathbf{r}} (U)$, as per Definition
    \ref{Def: Coherent and Homogeneous Germs}, then, fixed a compact set $K
    \subset U$, for every open set $U_K \subset K$, there exists $\tilde{\alpha}_{\tilde{K}}$ such that
    \begin{equation}
      \label{propertiesreconstructioneq2} \mathcal{R}^{\gamma} F \in
      \mathcal{Z}^{\tilde{\alpha}_{\tilde{K}}} (U_K), \hspace{0.27em}
      \hspace{0.27em} \hspace{0.27em} \text{and} \hspace{0.27em}
      \hspace{0.27em} \hspace{0.27em} \| \mathcal{R}^{\gamma} F
      \|_{\mathcal{Z}^{\tilde{\alpha}_{\tilde{K}}}_{\tilde{K}, \bar{\lambda}}}
      \leqslant \| F
      \|_{\mathcal{G}^{\tmmathbf{\alpha}_{\gamma}}_{\tmmathbf{r}} (U),
      \tilde{K}} + \| F
      \|_{\mathcal{G}^{\tmmathbf{\alpha}_{\gamma}}_{\tmmathbf{r}} (U),
      \tilde{K}},
    \end{equation}
   where $\mathcal{R}^{\gamma} F$ is as per Theorem \ref{Thm:
    Reconstruction Theorem} while $\mathcal{Z}^{\tilde{\alpha}_K} (U_K)$ is as per
    Definition \ref{Def: Holder-Zygmund}.

    \item If $F \in \mathcal{G}^{\tmmathbf{\alpha}_{\gamma}}_{\tmmathbf{r}}
    (U)$, with $\gamma \leqslant 0$, and $f, f' \in \mathcal{D}' (U)$ are two
    reconstructions of $F$, namely Equation \eqref{Eq: Reconstruction Theorem} holds
    true for $f$ and $f'$ in place of $\mathcal{R}^{\gamma} F$, then
    \begin{equation}
      \label{propertiesreconstructioneq3} f - f' \in \mathcal{Z}^{\gamma} (U) 
      \hspace{0.27em} \hspace{0.27em} \hspace{0.27em} \text{and}
      \hspace{0.27em} \hspace{0.27em} \hspace{0.27em} \| f - f'
      \|_{\mathcal{Z}^{\gamma}_{K, \bar{\lambda}}} \leqslant \| F
      \|_{\mathcal{G}^{\tmmathbf{\alpha}_{\gamma}}_{\tmmathbf{r}} (U),
      \tilde{K}} .
    \end{equation}
  \end{enumerate}
\end{corollary}

{\noindent}The proof of the above corollary is a minor adaptation of that on $\mathbb{R}^d$ discussed in {\cite[Thm. 12.7 \& Rem. 12.9]{CZ20}}, thus we omit it.

\section{Schauder Estimates for Germs on Open Subsets of
$\mathbb{R}^d$}\label{Sec: Schauder Estimates Open Subsets}

In this section, we shall discuss multi-level Schauder estimates for germs of distributions on
open subsets $U$ of $\mathbb{R}^d$, $d \ge 1$. The crucial result of this part is Theorem \ref{Main Result 1}, whose counterpart on $\mathbb{R}^d$ has been discussed in {\cite[Thm. 3.17]{BCZ}}. 

As a preliminary step, we need to recall the notion of \emph{regularizing kernel}, which will prove essential for the ensuing investigation. 

\begin{definition}[$\beta-$ regularizing kernel]
  \label{Def: Regularizing Kernel}Let $U \subseteq \mathbb{R}^d$, $d \geqslant
  1$, $m, r \in \mathbb{N}$ and $\beta, \rho > 0$. We say that $\Kappa : U
  \times U \rightarrow \mathbb{R}$ is a \textbf{$\beta$-regularizing kernel of order
  $(m, r)$ with range $\rho$} if
  \[ \Kappa (x, y) = \sum_{n = 0}^{\infty} \Kappa_n (x, y), \qquad \tmop{for}
     a.e. \quad x, y \in U. \]
  For any $n \in \mathbb{N}$, $\Kappa_n \in C^{m, r}  (U \times U)$,
  \textit{i.e.}, it is $m$-times differentiable with respect to the first variable and
  $r$-times with respect to the second one. In addition each $\Kappa_n$ is such that, for any compact set
  $K \subset U$,
  \begin{enumerate}
    \item $\tmop{supp} (\Kappa_n) \subset \{(x, y) \in U \times U \barsuchthat
    \hspace{0.27em} |x - y| \leqslant \rho 2^{- n} \}$.
    
    \item for any $k, \ell \in \mathbb{N}^d$ such
    that $|k| \leqslant m$ and $| \ell | \leqslant r$ and for any $x, y \in
    K$, there exists $c_K > 0$ such that
    \[ | \partial_1^k \partial_2^{\ell} \Kappa_n (x, y) | \leqslant c_K 2^{(d
       - \beta + | \ell | + |k|) n} \]
    where $\partial_1, \partial_2$ denote the partial derivatives with respect to the variables $x$
    and $y$, respectively.    
    \item for all $k, \ell \in \mathbb{N}^d$ such that
    $|k|, | \ell | \leqslant r$ and for any $x, y \in K$, there exists $c_K >
    0$ such that
    \[ \left| \int_U (y - x)^{\ell} \partial_2^k \Kappa_n (x, y) \mathd x
       \right| \leqslant c_K 2^{- \beta n}, \]
    where $(y - x)^{\ell} \assign (y_1 - x_1)^{\ell_1} \ldots (y_d - x_d)^{\ell_d}$.
  \end{enumerate}
\end{definition}

Given a $\beta$-regularizing kernel $\mathrm{K}$ as per previous definition, a natural question concerns the
possibility of establishing an interplay between its action and the theory of
germs on open sets of $\mathbb{R}^d$ outlined in Section \ref{Sec: Germs of
distributions on open sets}. Heuristically, let $(F_x)_{x \in U}$ be an $(\tmmathbf{\alpha}, \gamma)$-coherent germ with $\gamma > 0$ as per
Definition \ref{Def: germ} abiding by the hypotheses of Theorem \ref{Thm:
Reconstruction Theorem}. Hence, item $1.$ in Theorem \ref{Thm: Reconstruction Theorem} entails that the reconstruction $\mathcal{R}^{\gamma} F \in \mathcal{D}' (U)$ is
unique. Under the assumption that the convolution $\Kappa \star
\mathcal{R}^{\gamma} F$ is well-defined, one might wonder whether there exists
a corresponding map $\mathcal{K}$ acting on the space of germs and yielding a
novel germ of distributions such that $\Kappa \star \mathcal{R}^{\gamma} F
=\mathcal{R}^{\gamma'} \mathcal{K} (F)$. Here we use the symbol $\gamma'$ to
highlight that, \tmtextit{a priori}, the coherence parameter of the new germ
might change. Yet, the formalization of this intuitive argument within our setting is not as straightforward as one might think. 

\begin{remark}
  \label{Rem: Worse Alpha}In {\cite{BCZ}} this issue is addressed and solved
  on $\mathbb{R}^d$ and, under the same assumptions taken by these authors, our
  result might appear to be weaker, \tmtextit{cf.} {\cite[Thm. 3.17]{BCZ}} but this is not entirely true. On the one hand, this is because their proof relies on a large scale
  decomposition of test functions, which does not admit a straightforward
  counterpart on bounded open sets, \tmtextit{cf.} {\cite[Thm. 5.6 \& Lem.
  5.11]{BCZ}}. On the other hand, we shall consider homogeneous and coherent
  germs as per Definitions \ref{Def: Coherence} and \ref{Def: Homogeneity},
  where we allow the parameters $\alpha_K$ and $\bar{\alpha}_K$ to depend on
  an underlying compact subset $K$, contrary to {\cite{BCZ}} where the
  counterparts are taken to be constant. This allows to bypass the issue
  mentioned above with reference to the large scale decomposition.
  Furthermore, when working on a Riemannian manifold $M$, the absence of a
  global chart leads naturally to considering parameters $\alpha_K$ and
  $\bar{\alpha}_K$ depending on an underlying compact subset and, unless the
  manifold $M$ is compact itself, we cannot work with $\inf_{K \subset M}
  \alpha_K$ or with $\inf_{K \subset M}  \bar{\alpha}_K$ as they might not be
  finite.
\end{remark}

{\noindent}In view of this remark, we feel that, although the proof
of the following results is similar in many aspects to that given in
{\cite{BCZ}}, the differences shall become so relevant in the following
sections that delving into the details is justified. To avoid a lengthy and not
so user friendly statement we shall divide it in several propositions. Furthermore, we set most of the notation and of the conventions before the main
statements. More precisely, for all $\rho > 0$ and given a bounded open subset
$U \subseteq \mathbb{R}^d$, we introduce the restriction
\begin{equation}
  \label{Eq: Notation} U^{\rho} \assign \{ x \in U \, \barsuchthat \, \text{dist} (x, \partial
  U) > \rho \} .
\end{equation}
With $K$ we shall denote instead a compact subset such that $K \subset U^{2
\rho}$, while $K'$ is the smallest convex compact set containing $K$. In addition, let
$\tilde{K}^{\prime} \assign K^{\prime}_{3 D_K^{U^{\rho}} / 4}$, see Equation
\eqref{Eq: K-enlargement} with $D_K^{U^{\rho}} = \mathrm{dist} (K, \partial
U^{\rho})$. At last, given $\alpha_K \in \mathbb{R}$ and $r_K \in \mathbb{N}$, we define
$\bar{\alpha}_K \assign \min \{\alpha_K, - r_K - d\}$. 

We start by
stating two preliminary results which will play a role in the proof of the
main results of this section.

\noindent The first of these states that under suitable assumptions the integration of a
coherent (\emph{resp.} homogeneous) germ against a $\beta$-regularising kernel produces a
weakly coherent (\emph{resp.} weakly homogeneous) germ with coherence (\emph{resp.} homeogeneity) increased by
$\beta$. As the proof follows the steps of {\cite[Thm. 5.4]{BCZ}} we omit it.

\begin{theorem}[Integration of germs]
  \label{Thm: Integration of Germs}Let $U \subseteq \mathbb{R}^d$ be an open
  set. Consider $\mathsf{K}$, a $\beta$-regularising kernel of order $(m, r)$
  over $U$ with range $\rho$ as per Definition \ref{Def: Regularizing Kernel} and let $\gamma \in \mathbb{R}$. For any families
  of real numbers $\bar{\alpha}_K, \alpha_K$ and $r_K$, indexed by the compact
  sets $K \subset U$ with $\bar{\alpha}_K, \alpha_K \leqslant \gamma
  \hspace{0.17em} \hspace{0.17em}, \forall K \subset U$, it holds that
  \begin{itemize}
    \item if $m > \bar{\alpha}_K + \beta, \hspace{0.17em} \hspace{0.17em}
    \forall K \subset U$, the map $\mathsf{K} :
    \mathsf{G}^{\tmmathbf{\bar{\alpha}}}_{\tmmathbf{r}} (U) \to
    \check{\mathsf{G}}^{\tmmathbf{\bar{\alpha}} + \beta}_{\tmmathbf{r}} (U^{2
    \rho})$ is continuous, \textit{i.e.},
    \begin{equation}
      \label{kernelhom} \| \Kappa F
      \|_{\check{\mathsf{G}}^{\tmmathbf{\bar{\alpha}} + \beta}_{\tmmathbf{r}}
      (U^{2 \rho}), K} \lesssim_{K, \rho, \bar{\lambda}} \| F
      \|_{\mathsf{G}^{\tmmathbf{\bar{\alpha}}}_{\tmmathbf{r}} (U), K} .
    \end{equation}
    \item if $m > \gamma + \beta$, the map $\mathsf{K} :
    \mathcal{G}^{\tmmathbf{\alpha}_{\gamma}}_{\tmmathbf{r}} (U) \to
    \check{\mathcal{G}}^{(\tmmathbf{\alpha}+ \beta)_{\gamma +
    \beta}}_{\tmmathbf{r}} (U^{2 \rho})$ is continuous, \textit{i.e.},
    \begin{equation}
      \label{kernelcoh} \| \Kappa F
      \|_{\check{\mathcal{G}}^{(\tmmathbf{\alpha}+ \beta)_{\gamma +
      \beta}}_{\tmmathbf{r}} (U^{2 \rho}), K} \lesssim_{K, \rho,
      \bar{\lambda}} \| F
      \|_{\mathcal{G}^{\tmmathbf{\alpha}_{\gamma}}_{\tmmathbf{r}} (U), K} .
    \end{equation}
  \end{itemize}
  In addition, assuming that $r > \max \{ \bar{\alpha}_K, \alpha_K \}, \forall
  K \subset U$ and that $m > \gamma + \beta$, the map $\mathsf{K} :
  \mathcal{G}^{\tmmathbf{\bar{\alpha}},
  \tmmathbf{\alpha}_{\gamma}}_{\tmmathbf{r}} (U) \to
  \check{\mathcal{G}}^{\tmmathbf{\bar{\alpha}} + \beta, (\tmmathbf{\alpha}+
  \beta)_{\gamma + \beta}}_{\tmmathbf{r}} (U^{2 \rho})$ is continuous, \textit{i.e.},
  \begin{equation}
    \label{kernelcohom} \| \Kappa F
    \|_{\check{\mathcal{G}}^{\tmmathbf{\bar{\alpha}} + \beta,
    (\tmmathbf{\alpha}+ \beta)_{\gamma + \beta}}_{\tmmathbf{r}} (U^{2 \rho}),
    K} \lesssim_{K, \rho, \bar{\lambda}} \| F
    \|_{\mathcal{G}^{\tmmathbf{\bar{\alpha}},
    \tmmathbf{\alpha}_{\gamma}}_{\tmmathbf{r}} (U), K} .
  \end{equation}
\end{theorem}

The second result concerns the possibility of turning a weakly homogeneous
germ in a homogeneous one by means of the subtraction of a suitable Taylor
polynomial. This procedure is also known in the literature as
\tmtextit{renormalization}, although we shall refrain from adopting this
nomenclature since it carries different meanings both in quantum field theory
and in stochastic partial differential equations. We prefer to avoid
possible sources of confusion. For the proof, see {\cite[Thm. 5.10]{BCZ}}.

\begin{theorem}
  \label{Thm: Positive Renormalization}Let $U \subseteq \mathbb{R}^d$ and let
  $F \in \check{\mathsf{G}}^{\tmmathbf{\bar{\alpha}}}_{\tmmathbf{r}} (U)$,
  see Definition \ref{Def: Coherence}.
  \begin{itemize}
    \item For any compact $K \subset U$, it holds that
    \begin{equation}
      \label{pointwisederivativesbound1} \sup_{x \in K} |D^k F_x (x) |
      \lesssim \| F
      \|_{\check{\mathsf{G}}^{\tmmathbf{\bar{\alpha}}}_{\tmmathbf{r}} (U), K},
      \hspace{0.27em} \hspace{0.27em} \hspace{0.27em} \hspace{0.27em}
      \text{for} \hspace{0.27em} \hspace{0.27em} 0 \leqslant |k| <
      \bar{\alpha}_K,
    \end{equation}
    where $D^k \equiv \partial_{x^1}^{k_1} \ldots \partial_{x^n}^{k_n}$ with
    $|k| = k_1 + \ldots + k_n$.
    
    \item Setting $\mathbb{X}^k_x (y) \assign (y - x)^k$ and assuming that,
    $\forall K \subset U, \bar{\alpha}_K \nin \mathbb{N}_0$, then
    \begin{equation}
      \label{renormalization} G_x \assign F_x -\mathcal{T}^{\gamma}_x (F_x) =
      F_x - \sum_{0 \leqslant \lvert k \rvert < \gamma} D^k F_x (x)
      \frac{\mathbb{X}^k_x}{k!},
    \end{equation}
    is an $\tmmathbf{\bar{\alpha}}$-homogeneous germ of distributions of order
    $\tmmathbf{r}$, \tmtextit{i.e.}, $G \in
    \mathsf{G}^{\tmmathbf{\bar{\alpha}}}_{\tmmathbf{r}} (U)$, and
    \begin{equation}
      \label{positiverenormbound2} \| G
      \|_{\mathsf{G}^{\tmmathbf{\bar{\alpha}}}_{\tmmathbf{r}} (U), K} \lesssim
      \| F \|_{\check{\mathsf{G}}^{\tmmathbf{\bar{\alpha}}}_{\tmmathbf{r}}
      (U), K} .
    \end{equation}
    \item If, for all compact $K \subset U$ such that $\bar{\alpha}_K \in
    \mathbb{N}$ and for all multi-indices $k$ with $\lvert k \rvert =
    \bar{\alpha}_K$, $D^k F \in \mathsf{G}^0_{r_K} (U)$, then $G \in
    \mathsf{G}^{\tmmathbf{\bar{\alpha}}}_{\tmmathbf{r}} (U)$ and,
    \begin{equation}
      \label{positiverenormbound3} \| G
      \|_{\mathsf{G}^{\tmmathbf{\bar{\alpha}}}_{\tmmathbf{r}} (U), K} \lesssim
      \| F \|_{\check{\mathsf{G}}^{\tmmathbf{\bar{\alpha}}}_{\tmmathbf{r}}
      (U), K} + \sum_{\lvert k \rvert = \bar{\alpha}_K} \| D^k F
      \|_{\check{\mathsf{G}}^0_{\tmmathbf{r}} (U), K} .
    \end{equation}
  \end{itemize}
\end{theorem}

{\noindent}As a Corollary of Theorem \ref{Thm: Integration of Germs} we
establish how an element lying in $\check{\mathcal{G}}^{\gamma,
\tmmathbf{\alpha}_{\gamma}}_{\tmmathbf{r}} (U)$, namely a weakly coherent and
weakly homogeneous germs, can be turned into a homogeneous and coherent germ
by subtracting a suitable Taylor polynomial, see {\cite[Thm. 5.10]{BCZ}}. In
the following we shall use extensively the notation and the results listed in 
Appendix \ref{Sec: App A}.

\begin{corollary}
  \label{Cor: Postive Renormalization}Let $U \subseteq \mathbb{R}^d$ be a
  convex open set and let $F = (F_x)_{x \in U} \in
  \check{\mathcal{G}}^{\gamma, \tmmathbf{\alpha}_{\gamma}}_{\tmmathbf{r}}
  (U)$, with $\alpha_K \neq 0 \hspace{0.17em}, \forall K \subset U$, $\gamma
  \nin \mathbb{N}_0$. Then for all multi-indices $k$ such that $0 \leqslant
  \lvert k \rvert \leqslant \gamma$, the pointwise derivatives $D^k F_x (x)$
  are well-defined. Setting
  \begin{equation}
    \label{Eq: renormalization} G_x \assign F_x -\mathcal{T}^{\gamma}_x (F_x)
    \equiv F_x - \sum_{0 \leqslant \lvert k \rvert < \gamma} D^k F_x (x)
    \frac{\mathbb{X}^k_x}{k!},
  \end{equation}
  it holds that
  \begin{itemize}
    \item[(i)] there exists a family $\{ \tilde{\alpha}_K \}_{\Kappa \subset U}
    \subset \mathbb{R}$ with $\tilde{\alpha}_K \leqslant \min \{0, \gamma\},
    \forall K$ such that $G \in \mathcal{G}^{\gamma,
    \tilde{\tmmathbf{\alpha}}_{\gamma}}_{\tmmathbf{r}} (U)$.
    
    \item[(ii)] the map $F_x \mapsto G_x$ is linear, and for all compact sets $K
    \subset U$, one has that
    \begin{align}
      & \| G \|_{\mathsf{G}^{\gamma}_{\tmmathbf{r}} (U), K} \lesssim
      \hspace{0.27em} \| F \|_{\check{\mathsf{G}}^{\gamma}_{\tmmathbf{r}} (U),
      K}  \label{Eq: Postive Renormalization Estimate 1}\\
      & \| G
      \|_{\mathcal{G}^{\tilde{\tmmathbf{\alpha}}_{\gamma}}_{\tmmathbf{r}} (U),
      K} \lesssim \hspace{0.27em} \| F
      \|_{\check{\mathsf{G}}^{\gamma}_{\tmmathbf{r}} (U), K'} .  \label{Eq:
      Postive Renormalization Estimate 2}
    \end{align}
  \end{itemize}
\end{corollary}

\begin{proof}
  We focus only on Equation \eqref{Eq: Postive Renormalization Estimate 2}
  since all other statements are a direct consequence of {\cite[Thm.
  5.10]{BCZ}} replacing $\mathbb{R}^d$ with $U \subset \mathbb{R}^d$. Theorem
  \ref{Thm: Positive Renormalization} entails that, given $F \in
  \check{\mathcal{G}}^{\gamma, \tmmathbf{\alpha}_{\gamma}}_{\tmmathbf{r}} (U)$
  as per Definition \ref{Def: Coherent and Homogeneous Germs}
  \[ G_x \assign F_x -\mathcal{T}^{\gamma}_x (F_x) = F_x - \sum_{0 \leqslant
     \lvert k \rvert < \gamma} D^k F_x (x) \mathbb{X}^k_x \in
     \mathsf{G}^{\gamma}_{\tmmathbf{r}} (U), \qquad \| G
     \|_{\mathsf{G}^{\gamma}_{\tmmathbf{r}} (U), K} \lesssim \| F
     \|_{\check{\mathsf{G}}^{\gamma}_{\tmmathbf{r}} (U), K} . \]
  To prove coherence of $G$, we use Proposition \ref{homcoherence}, item $(2)$
  yielding the existence of $\tilde{\alpha}_K$, \tmtextit{cf.} \ Equation
  \eqref{alphak}, such that $G$ is $(\tilde{\alpha}_k, \gamma)$-coherent.
  Additionally, again by Proposition \ref{homcoherence}, we have that
  \[ \| G \|_{\mathcal{G}^{\tilde{\tmmathbf{\alpha}}_{\gamma}}_{\tmmathbf{r}}
     (U), K} \lesssim \| G \|_{\mathsf{G}^{\gamma}_{\tmmathbf{r}} (U), K'}
     \lesssim \| F \|_{\check{\mathsf{G}}^{\gamma}_{\tmmathbf{r}} (U), K'}, \]
  where $K'$ is the smallest convex compact set containing $K$.
\end{proof}

{\noindent} At this stage, we have all the ingredients to establish the main result of this section, namely 
Schauder estimates for germs of distributions on an open convex set.

\begin{theorem}[Schauder estimates on open convex sets]
  \label{Main Result 1}Let $U \subseteq \mathbb{R}^d$ be open, convex and
  bounded. Let $F \in \mathcal{G}^{\tmmathbf{\alpha}_{\gamma}}_{\tmmathbf{r}}
  (U)$, as per Definition \ref{Def: Coherence}, with $\alpha_K \leqslant
  \gamma$, $\forall K \subset U$ and $\gamma \neq 0$. We denote by
  $\mathcal{R}^{\gamma} F$ the reconstruction of $F$, as per Theorem \ref{Thm:
  Reconstruction Theorem}. Furthermore, let $\beta > 0$ satisfy
  \[ \alpha_K + \beta \neq 0 \hspace{0.27em} \hspace{0.27em} \hspace{0.27em}
     \forall K \subset U, \hspace{0.27em} \hspace{0.27em} \hspace{0.27em}
     \gamma + \beta \nin \mathbb{N}, \]
  and consider a $\beta$-regularizing kernel $\mathsf{K}$ on $U$ with range
  $\rho$ of order $(m, r)$, as per Definition \ref{Def: Regularizing Kernel},
  to be such that
  \begin{equation}
    \label{regularitykernel} m > \gamma + \beta, \qquad \qquad r > - \alpha_K,
    \qquad \qquad \forall K \subset U.
  \end{equation}
  Defining the germ $(\mathcal{K}^{\gamma, \beta} F)_x$ as
  \begin{equation}
    \mathcal{K}^{\gamma, \beta} F \assign \mathsf{K} F -\mathcal{T}^{\gamma +
    \beta} (\mathsf{K} \{F -\mathcal{R}^{\gamma} F\}) = \mathsf{K} F - \sum_{0
    \leqslant \lvert k \rvert \leqslant \lfloor \gamma + \beta \rfloor} D^k
    (\mathsf{K} \{F -\mathcal{R}^{\gamma} F\}) (x) \frac{\mathbb{X}^k}{k!},
    \label{maintheoremequationeq}
  \end{equation}
  where $\mathbb{X}^k_x (y) \assign (y - x)^k$, it holds that, with $\bar{r}_K
  \assign - \alpha_K$,
  \begin{equation}
    \label{maintheoremequation1} \mathcal{K}^{\gamma, \beta} F - \mathsf{K}
    (\mathcal{R}^{\gamma} F) \in \mathcal{G}^{\gamma + \beta,
    \tilde{\tmmathbf{\alpha}}_{\gamma + \beta}}_{\bar{\tmmathbf{r}}} (U^{2
    \rho}),
  \end{equation}
  where $\tilde{\tmmathbf{\alpha}}$ is as per Corollary~\ref{Cor: Postive
  Renormalization}.
  In addition, the following continuity estimate holds true: There exists a
  compact $\tilde{K}' \subset U$ such that
  \begin{equation}
    \begin{array}{ll}
      \| \mathcal{K}^{\gamma, \beta} F - \mathsf{K} (\mathcal{R}^{\gamma} F)
      \|_{\mathcal{G}^{\gamma + \beta, \tilde{\tmmathbf{\alpha}}_{\gamma +
      \beta}}_{\bar{\tmmathbf{r}}} (U^{2 \rho}), K} & \lesssim \hspace{0.17em}
      \| F \|_{\mathcal{G}^{\tmmathbf{\alpha}_{\gamma}}_{\tmmathbf{r}} (U),
      \tilde{K}'},
    \end{array} \label{maintheoremequation2}
  \end{equation}
  where here $\lesssim$ encodes a proportionality constant depending on $K,
  \rho, \gamma + \beta, \alpha_{\tilde{K}'}, \bar{r}_{\tilde{K}'},
  \bar{\lambda}$, \tmtextit{cf.} Definition \ref{Def: Coherence}.
\end{theorem}

\begin{proof}
  Given $F \in \mathcal{G}^{\tmmathbf{\alpha}_{\gamma}}_{\tmmathbf{r}} (U)$ as
  per hypothesis, $F -\mathcal{R}^{\gamma} F$ is both $\gamma$-homogeneous and
  $(\tmmathbf{\alpha}, \gamma)$-coherent of order $\bar{\tmmathbf{r}} =
  \lfloor -\tmmathbf{\alpha} \rfloor$, as shown in Remark
  \ref{reconstructionishomogeneous1}. By Theorem \ref{Thm: Integration of
  Germs}, $\mathsf{K} (F -\mathcal{R}^{\gamma} F) \in
  \check{\mathcal{G}}^{\gamma + \beta, (\tmmathbf{\alpha}+ \beta)_{\gamma +
  \beta}}_{\tmmathbf{r}} (U^{2 \rho})$. On account of Equation
  \eqref{maintheoremequationeq}, we have
  \[ \mathcal{K}^{\gamma, \beta} F - \mathsf{K} (\mathcal{R}^{\gamma} F) =
     \mathsf{K} (F -\mathcal{R}^{\gamma} F) - \sum_{0 \leqslant \lvert k
     \rvert \leqslant \lfloor \gamma + \beta \rfloor} D^k (\mathsf{K} \{F
     -\mathcal{R}^{\gamma} F\}) (x) \frac{\mathbb{X}^k}{k!} . \]
  Applying Corollary \ref{Cor: Postive Renormalization}, it follows that there
  exists $\tilde{\tmmathbf{\alpha}}$ with $\tilde{\alpha}_K \leqslant
  \min \{ 0, \gamma + \beta \}$ for any $K$, such that it holds that
  $\mathcal{K}^{\gamma, \beta} F - \mathsf{K} (\mathcal{R}^{\gamma} F) \in
  \mathcal{G}^{\gamma + \beta, \tilde{\tmmathbf{\alpha}}_{\gamma +
  \beta}}_{\bar{\tmmathbf{r}}} (U)$, with
  \begin{equation}
    \label{proofofmainresult1} \| \mathcal{K}^{\gamma, \beta} F - \mathsf{K}
    (\mathcal{R}^{\gamma} F)
    \|_{\mathcal{G}^{\tilde{\tmmathbf{\alpha}}_{\gamma +
    \beta}}_{\bar{\tmmathbf{r}}} (U^{2 \rho}), K} \lesssim_{K, \gamma + \beta,
    \bar{r}_K, d} \| \mathsf{K} (F -\mathcal{R}^{\gamma} F)
    \|_{\check{\mathsf{G}}^{\gamma + \beta}_{\tmmathbf{r}} (U^{2 \rho}), K'},
  \end{equation}
  and $K' \supset K$ is a convex compact set contained in $U$. By Equation
  \eqref{kernelhom} it descends that, for a suitable enlargement $\tilde{K}'$
  of $K'$
  \begin{equation}
    \begin{array}{lll}
      \| \mathsf{K} (F -\mathcal{R}^{\gamma} F) \|_{\check{\mathsf{G}}^{\gamma
      + \beta}_{\tmmathbf{r}} (U^{2 \rho}), K'} & \lesssim_{\rho,
      \bar{\lambda}} & \| F -\mathcal{R}^{\gamma} F
      \|_{\check{\mathsf{G}}^{\gamma}_{\tmmathbf{r}} (U), K'}\\
      & \lesssim_{\alpha_{\tilde{K}'}, r_{\tilde{K}'}, d} & \| F
      \|_{\mathcal{G}^{\tilde{\tmmathbf{\alpha}}_{\gamma}}_{\bar{\tmmathbf{r}}}
      (U^{2 \rho}), \tilde{K}'}
    \end{array} \label{proofofmainresult3}
  \end{equation}
  where we used Equation \eqref{reconstructionishomogeneouseq}. Equations
  \eqref{proofofmainresult1} and \eqref{proofofmainresult3} together give
  Equation \eqref{maintheoremequation2}.
\end{proof}

\begin{proposition}
  \label{Main Result 2}{\noindent}Under the assumptions and notations of
  Theorem \ref{Main Result 1}, denoting by $\bar{\tmmathbf{r}}' \assign
  (r_{K'})_{K' \subset U}$, and assuming $\mathsf{K}
  (\mathcal{R}^{\gamma} F) \in \mathcal{D}' (U^{2 \rho})$, the map
  \[ \mathcal{K}^{\gamma, \beta} :
     \mathcal{G}^{\tmmathbf{\alpha}_{\gamma}}_{\tmmathbf{r}} (U) \to
     \mathcal{G}^{\tilde{\tmmathbf{\alpha}}_{\gamma +
     \beta}}_{\bar{\tmmathbf{r}}'} (U^{2 \rho}), \qquad \qquad F \mapsto
     \mathcal{K}^{\gamma, \beta} F, \]
  is linear and continuous, \textit{i.e.}, 
  \begin{equation}
    \label{Eq: Main Result coherence} \| \mathcal{K}^{\gamma, \beta} F
    \|_{\mathcal{G}^{\tilde{\tmmathbf{\alpha}}_{\gamma +
    \beta}}_{\bar{\tmmathbf{r}}'} (U^{2 \rho}), K} \lesssim \| F
    \|_{\mathcal{G}^{\tmmathbf{\alpha}_{\gamma}}_{\tmmathbf{r}} (U),
    \tilde{K}'} .
  \end{equation}
  Furthermore
  \begin{equation}
    \label{maintheoremequation3} \mathcal{R}^{\gamma + \beta}
    (\mathcal{K}^{\gamma, \beta} F) = \mathsf{K} (\mathcal{R}^{\gamma} F),
  \end{equation}
  on $U^{2 \rho}$, i.e., $\mathsf{K} (\mathcal{R}^{\gamma} F)$ is a  $(\gamma +
  \beta)-$reconstruction of the germ of distributions $\mathcal{K}^{\gamma,
  \beta} F$.
\end{proposition}

\begin{proof}
  If $\mathsf{K} (\mathcal{R}^{\gamma} F) \in \mathcal{D}' (U^{2 \rho})$, by
  triangular inequality
  \begin{equation}
    \| \mathcal{K}^{\gamma, \beta} F
    \|_{\mathcal{G}^{\tilde{\tmmathbf{\alpha}}_{\gamma +
    \beta}}_{\bar{\tmmathbf{r}}'} (U^{2 \rho}), K} \leqslant \|
    \mathcal{K}^{\gamma, \beta} F - \mathsf{K} (\mathcal{R}^{\gamma} F)
    \|_{\mathcal{G}^{\tilde{\tmmathbf{\alpha}}_{\gamma +
    \beta}}_{\bar{\tmmathbf{r}}'} (U^{2 \rho}), K} + \| \mathsf{K}
    (\mathcal{R}^{\gamma} F)
    \|_{\mathcal{G}^{\tilde{\tmmathbf{\alpha}}_{\gamma +
    \beta}}_{\bar{\tmmathbf{r}}'} (U^{2 \rho}), K} \lesssim \| F
    \|_{\mathcal{G}^{\tmmathbf{\alpha}_{\gamma}}_{\tmmathbf{r}} (U),
    \tilde{K}'}, \label{proofofmainresult5}
  \end{equation}
  where in the last inequality we used Equation \eqref{maintheoremequation2}
  together with
  \[ \| \mathsf{K} (\mathcal{R}^{\gamma} F)
     \|_{\mathcal{G}^{\tilde{\tmmathbf{\alpha}}_{\gamma +
     \beta}}_{\bar{\tmmathbf{r}}'} (U^{2 \rho}), K} = 0, \]
  as the coherence norm of a constant germ is vanishing. Note that in the last
  step of Equation \eqref{proofofmainresult5}, the constant of the estimate
  depends on $\bar{r}_{\tilde{K}'}$ and not on $r_{\tilde{K}'}$ since we can
  always set $\bar{r} = \max \{- \alpha_K, - \gamma\} = - \alpha_K \geqslant
  r$. Equation \eqref{maintheoremequation3} follows from Equation
  \eqref{proofofmainresult1}. As a matter of fact, it is sufficient to compare
  Equations \eqref{Eq: Homogeneity} and \eqref{Eq: Reconstruction Theorem}
  noting that, by assumptions, $\gamma + \beta \neq 0$.
\end{proof}

{\noindent}At last we establish also a bound reading all the germs above as
being both homogeneous and coherent as per Definition \ref{Def: Coherent and
Homogeneous Germs}.

\begin{proposition}
  \label{Main Result 3}Under the same assumptions and notations of Theorem
  \ref{Main Result 1} and of Proposition \ref{Main Result 2}, the map
  \[ \mathcal{K}^{\gamma, \beta} : \mathcal{G}^{\tmmathbf{\bar{\alpha}},
     \tmmathbf{\alpha}_{\gamma}}_{\tmmathbf{r}} (U) \to
     \mathcal{G}^{(\tmmathbf{\bar{\alpha}} + \beta \wedge 0),
     \tilde{\tmmathbf{\alpha}}_{\gamma + \beta}}_{\bar{\tmmathbf{r}}} (U^{2
     \rho}), \]
  is linear and continuous, namely there exists a compact set $\tilde{K} \subset
  U$ such that the following bound holds true
  \begin{equation}
    \label{Eq: Homogeneous and coherent norm reconstruction} \|
    \mathcal{K}^{\gamma, \beta} F \|_{\mathcal{G}^{(\tmmathbf{\bar{\alpha}} +
    \beta \wedge 0), \tilde{\tmmathbf{\alpha}}_{\gamma +
    \beta}}_{\bar{\tmmathbf{r}}} (U^{2 \rho}), K} \lesssim \| F
    \|_{\mathcal{G}^{\tmmathbf{\bar{\alpha}},
    \tmmathbf{\alpha}_{\gamma}}_{\tmmathbf{r}} (U), \tilde{K}} .
  \end{equation}
\end{proposition}

\begin{proof}
  Under the above assumptions, Equation \eqref{reconstructionishomogeneouseq1}
  entails that, given $F \in \mathcal{G}^{\tmmathbf{\bar{\alpha}},
  \tmmathbf{\alpha}_{\gamma}}_{\tmmathbf{r}} (U)$,
  \begin{equation}
    \label{reconstructionishomogeneouseq11} \| \mathcal{R}^{\gamma} F
    \|_{\mathsf{G}^{\tmmathbf{\bar{\alpha}}}_{\bar{\tmmathbf{r}}} (U), K}
    \lesssim \| F \|_{\mathcal{G}^{\tmmathbf{\alpha}_{\gamma}}_{\tmmathbf{r}}
    (U), \tilde{K}} + \| F
    \|_{\mathsf{G}^{\tmmathbf{\bar{\alpha}}}_{\tmmathbf{r}} (U), K} \leqslant
    \| F \|_{\mathcal{G}^{\tmmathbf{\bar{\alpha}},
    \tmmathbf{\alpha}_{\gamma}}_{\tmmathbf{r}} (U), \tilde{K}},
  \end{equation}
  while Equation~\eqref{kernelhom} gives
  \begin{equation}
    \label{ciaociao} \| \mathsf{K} (\mathcal{R}^{\gamma} F)
    \|_{\check{\mathsf{G}}^{\tmmathbf{\bar{\alpha}} +
    \beta}_{\bar{\tmmathbf{r}}} (U^{2 \rho}), K} \lesssim \|
    \mathcal{R}^{\gamma} F
    \|_{\mathsf{G}^{\tmmathbf{\bar{\alpha}}}_{\bar{\tmmathbf{r}}} (U), K} .
  \end{equation}
  Merging these two estimates, we obtain
  \begin{equation}
    \label{ciaociaociao} \| \mathsf{K} (\mathcal{R}^{\gamma} F)
    \|_{\check{\mathsf{G}}^{\tmmathbf{\bar{\alpha}} +
    \beta}_{\bar{\tmmathbf{r}}} (U^{2 \rho}), K} \lesssim \| F
    \|_{\mathcal{G}^{\tmmathbf{\bar{\alpha}},
    \tmmathbf{\alpha}_{\gamma}}_{\tmmathbf{r}} (U), \tilde{K}} .
  \end{equation}
  Notice once again that the germ $\mathcal{K}^{\gamma, \beta} F =
  (\mathcal{K}^{\gamma, \beta} F - \mathsf{K} (\mathcal{R}^{\gamma} F)) +
  \mathsf{K} (\mathcal{R}^{\gamma} F)$ is the sum of an
  $\bar{\tmmathbf{\alpha}}$-homogeneous germ of order $\bar{\tmmathbf{r}}$,
  and a weakly $\bar{\tmmathbf{\alpha}}$-homogeneous germ of order
  $\bar{\tmmathbf{r}}$. As a consequence of Remark \ref{weakarestrong},
  $\mathsf{K} (\mathcal{R}^{\gamma} F)$ is $(\bar{\alpha}_K +
  \beta)$-homogeneous on $K$ provided $\bar{\alpha}_K + \beta < 0$. If instead
  $\bar{\alpha}_K + \beta \geqslant 0$, since $\mathsf{K}
  (\mathcal{R}^{\gamma} F)$ is a constant germ, it is weakly $(\alpha,
  \gamma)$-coherent, for any $\alpha, \gamma \in \mathbb{R}$. Thus we can
  apply Corollary \ref{Cor: Postive Renormalization} to $\mathsf{K}
  (\mathcal{R}^{\gamma} F)$, obtaining that
  \begin{equation}
    \label{Gschauder} G \assign \mathsf{K} (\mathcal{R}^{\gamma} F)
    -\mathcal{T}^{\bar{\tmmathbf{\alpha}}_K + \beta} \mathsf{K}
    (\mathcal{R}^{\gamma} F),
  \end{equation}
  is $(\bar{\tmmathbf{\alpha}} + \beta)-$homogeneous, of order $\tmmathbf{r}$.
  The Taylor polynomials $\mathcal{T}^{\bar{\tmmathbf{\alpha}}_K + \beta}
  \mathsf{K} (\mathcal{R}^{\gamma} F)$ are homogeneous of degree $0$, for any
  order $r \in \mathbb{N}_0$. Additionally, Equation
  \eqref{pointwisederivativesbound1} entails that, for every compact set $K
  \subset U^{2 \rho}$,
  \begin{equation}
    \sup_{x \in K} \lvert D^k \mathsf{K} (\mathcal{R}^{\gamma} F) (x) \rvert
    \lesssim \| \mathsf{K} (\mathcal{R}^{\gamma} F)
    \|_{\check{\mathsf{G}}^{\tmmathbf{\bar{\alpha}}}_{\bar{\tmmathbf{r}}} (U),
    K}, \label{pointwisederivativesbound1ciao}
  \end{equation}
  for $0 \leqslant |k| < \bar{\alpha}_K$. Therefore, by adding the Taylor
  polynomial $\mathcal{T}^{\alpha_K + \beta} \mathsf{K} (\mathcal{R}^{\gamma}
  F)$ to $G$, we obtain that, on the compact set $K$, $\mathsf{K}
  (\mathcal{R}^{\gamma} F)$ is $0$-homogeneous of order $\bar{r}_K$. Bearing
  in mind Equation \eqref{Gschauder}, we end up with
  \begin{equation}
    \begin{array}{lll}
      \| \mathcal{K}^{\gamma, \beta} F \|_{\mathsf{G}^0_{\bar{\tmmathbf{r}}}
      (U^{2 \rho}), K} \leqslant \| G \|_{\mathsf{G}^0_{\bar{\tmmathbf{r}}}
      (U^{2 \rho}), K} + \| \mathcal{T}^{\alpha + \beta} \mathsf{K}
      (\mathcal{R}^{\gamma} F) \|_{\mathsf{G}^0_{\bar{\tmmathbf{r}}} (U^{2
      \rho}), K} & \lesssim & \| \mathsf{K} (\mathcal{R}^{\gamma} F)
      \|_{\check{\mathsf{G}}^{\tmmathbf{\bar{\alpha}}}_{\bar{\tmmathbf{r}}}
      (U^{2 \rho}), K}\\
      & \lesssim & \| \mathcal{R}^{\gamma} F
      \|_{\check{\mathsf{G}}^{\tmmathbf{\bar{\alpha}}}_{\bar{\tmmathbf{r}}}
      (U), K} \lesssim \| F \|_{\mathcal{G}^{\tmmathbf{\bar{\alpha}},
      \tmmathbf{\alpha}_{\gamma}}_{\tmmathbf{r}} (U), \tilde{K}'},
    \end{array} \label{proofofmainresult6}
  \end{equation}
  where $K'$ is the smallest convex compact set containing $K$, while
  $\tilde{K}' \assign K_{3 D_{K'}^{U^{2 \rho}} / 4}$. Here in the second
  inequality we exploited Equations \eqref{Eq: Postive Renormalization
  Estimate 1} while in the third one we used~\eqref{ciaociaociao}.
  Finally, in the last one we used Equation
  \eqref{reconstructionishomogeneouseq}. To conclude, putting together
  Equations \eqref{proofofmainresult5} and \eqref{proofofmainresult6}, and
  using the triangle inequality, it descends that
  \[ \mathcal{K}F \in \mathcal{G}^{\bar{\tmmathbf{\alpha}} + \beta \wedge 0,
     \tilde{\tmmathbf{\alpha}}_{\gamma}}_{U^{2 \rho}, \bar{\tmmathbf{r}}}, \]
  where $\bar{r}_K \assign - \alpha_K$, and $\tilde{\alpha}_K$ is as per
  Equation \eqref{alphak}. Once more the triangle inequality entails that
  \[ \begin{array}{lll}
       \| \mathcal{K}^{\gamma, \beta} F
       \|_{\mathcal{G}^{(\tmmathbf{\bar{\alpha}} + \beta \wedge 0),
       \tilde{\tmmathbf{\alpha}}_{\gamma + \beta}}_{\bar{\tmmathbf{r}}} (U^{2
       \rho}), K} & \leqslant & \| \mathcal{K}^{\gamma, \beta} F - \mathsf{K}
       (\mathcal{R}^{\gamma} F) \|_{\mathcal{G}^{(\tmmathbf{\bar{\alpha}} +
       \beta \wedge 0), \tilde{\tmmathbf{\alpha}}_{\gamma +
       \beta}}_{\bar{\tmmathbf{r}}} (U^{2 \rho}), K} + \| \mathsf{K}
       (\mathcal{R}^{\gamma} F) \|_{\mathcal{G}^{(\tmmathbf{\bar{\alpha}} +
       \beta \wedge 0), \tilde{\tmmathbf{\alpha}}_{\gamma +
       \beta}}_{\bar{\tmmathbf{r}}} (U^{2 \rho}), K}\\
       & \lesssim & \| F \|_{\mathcal{G}^{\tmmathbf{\bar{\alpha}},
       \tmmathbf{\alpha}_{\gamma}}_{\tmmathbf{r}} (U), \tilde{K}'},
     \end{array} \]
  where in the last line we exploited implicitly Equations
  \eqref{proofofmainresult5} and \eqref{proofofmainresult6}.
\end{proof}

To conclude the section, we highlight that all the results obtained can be
codified concisely in the following two commutative diagrams:
\[ \begin{array}{ccc}
     \mathcal{G}^{\tmmathbf{\alpha}_{\gamma}}_{\tmmathbf{r}} (U) &
     \xrightarrow{\mathcal{K}^{\gamma, \beta}} &
     \mathcal{G}^{\tilde{\tmmathbf{\alpha}}_{\gamma +
     \beta}}_{\bar{\tmmathbf{r}}} (U^{2 \rho})\\
     \longdownarrow \mathcal{R}^{\gamma} &  & \longdownarrow
     \mathcal{R}^{\gamma + \beta}\\
     \mathcal{D}' (U) & \xrightarrow{\Kappa} & \mathcal{D}' (U^{2 \rho})
   \end{array} \qquad \qquad \begin{array}{ccc}
     \mathcal{G}^{\tmmathbf{\bar{\alpha}},
     \tmmathbf{\alpha}_{\gamma}}_{\tmmathbf{r}} (U) &
     \xrightarrow{\mathcal{K}^{\gamma, \beta}} &
     \mathcal{G}^{(\tmmathbf{\bar{\alpha}} + \beta \wedge 0),
     \tilde{\tmmathbf{\alpha}}_{\gamma + \beta}}_{\bar{\tmmathbf{r}}} (U^{2
     \rho})\\
     \longdownarrow \mathcal{R}^{\gamma} &  & \longdownarrow
     \mathcal{R}^{\gamma + \beta}\\
     \mathcal{Z}^{\gamma} (U) & \xrightarrow{\Kappa} & \mathcal{Z}^{\gamma +
     \beta} (U^{2 \rho})
   \end{array} \]
   
\section{Germs of distributions on Riemannian
manifolds}\label{Sec:Germs-of-distributions-manifold}

Aim of this section is to extend the theory of germs of distributions to Riemannian smooth
manifolds, see Section \ref{Introduction} for a succinct overview of the key structures and notations. Since the definition of germs and the properties of coherence and
homogeneity on general smooth manifolds have already been discussed in
{\cite{RS21}}, henceforth we shall focus mainly on Schauder estimates on
Riemannian manifolds and on their interplay with the reconstruction theorem. We advise the reader to consult Appendix \ref{appendixB} and \ref{AppendixC}, as we will use extensively the results contained therein. 

As a preliminary step to achieve this goal, we introduce the notion of $\beta$-regularizing
kernel, generalizing the one in Definition~\ref{Def: Regularizing
Kernel}. In particular, in Proposition~\ref{Prop:kernel-and-exp} we prove that
the following definition is compatible with Definition~\ref{Def: Regularizing
Kernel} through the exponential map.

\begin{definition}[$\beta$-regularising kernel on Riemannian manifolds]
  \label{regularisingkernelonmanifold} Let $M$ be a Riemannian manifold, see Section \ref{Introduction}. Let $\beta > 0$, $m, r \in \mathbb{N}_0$
  and consider an exhaustion by compact sets of $M$, $\{\Omega_n\}_{n \in
  \mathbb{N}}$, as per Definition \ref{Def: Compact Exhaustion}. A function
  $\mathsf{K} : M \times M \to \mathbb{R}$ is called \textbf{$\beta$-regularising
  kernel subordinated to $\{\Omega_n\}_{n \in \mathbb{N}}$ of order $(m, r)$} if it abides by the following properties: 
  \begin{enumerate}
    \item The following decomposition holds true:
    \begin{equation}
      \label{kerneldecompositiononmanifold} \mathsf{K} (p, q) = \sum_{N =
      0}^{+ \infty} \mathsf{K}_n (p, q)  \hspace{0.27em} \hspace{0.27em}
      \hspace{0.27em} \hspace{0.27em} for \hspace{0.27em} a.e. \hspace{0.17em}
      \hspace{0.27em} p, q \in M,
    \end{equation}
    where $\mathsf{K}_n \in \mathcal{C}^{m, r}  (M \times M)$, \textit{i.e.}, it is
    $m$-times differentiable with respect to the first variable and $r$-times
    with respect to the second.
    
    \item For all $n \in \mathbb{N}_0$, it holds
    \begin{equation}
      \label{suppkernelmanifold} \tmop{supp} (\mathsf{K}_n) \subseteq \{(p, q)
      : p \in \Omega_n \setminus \Omega_{n - 1}, \hspace{0.27em} d_g (p, q)
      \leqslant 2^{- n} \rho (p)\},
    \end{equation}
    where $n$ is the smallest natural number such that $p \in \Omega_n$, while
    $d_g$ is the metric induced distance. In addition $\rho (p)$ is defined by
    \begin{equation}
      \label{rhon} \rho (p) = \rho_n, \hspace{0.27em} \hspace{0.27em}
      \hspace{0.27em} \text{if} \hspace{0.27em} \hspace{0.27em}
      \hspace{0.27em} p \in \Omega_n \setminus \Omega_{n - 1}, \qquad
      \tmop{such} \tmop{that} \qquad \rho_n \leqslant \frac{\mathcal{R}_C
      (\Omega_n)}{4}, \qquad \forall n,
    \end{equation}
    where $\mathcal{R}_C(\Omega_n)$ denotes the convexity radius as per
    Definition \ref{Def: Convexity Radius}.
    
    \item Consider a local chart $(U, \phi)$ of $M$. Denoting the push-forward
    by $\mathsf{K}^{\phi} \assign \phi_{\ast} \mathsf{K}$, for any compact set
    $K \subseteq \phi (U)$ there exists a constant $c^{\phi}_K > 0$ for which,
    given $k, l \in \mathbb{N}_0^d$ with $\lvert k \rvert \leqslant m$,
    $\lvert l \rvert \leqslant r$ we have, for $\phi (p), \phi (q)
    \hspace{0.17em} \in K$,
    \begin{equation}
      \label{kernelbound1onmanifold} | \partial_1^k \partial_2^l
      \mathsf{K}^{\phi}_n (\phi (p), \phi (q)) | \leqslant c^{\phi}_K 2^{(d -
      \beta + \lvert l \rvert + |k|) n},
    \end{equation}
    where $\partial_1^k$ denotes the partial derivative with respect to the
    first variable.
    
    \item Considering a coordinate chart $(U, \phi)$ as well a compact set $K
    \subset U$, given $k, l \in \mathbb{N}^d_0$ with $|k|, |l| \leqslant r$,
    and calling $x \assign \phi (p), y \assign \phi (q) \in K$, we have
    \begin{equation}
      \label{kernelbound2onmanifold} \left| \int_{\phi (U)} (x - y)^l
      \partial_2^k \mathsf{K}^{\phi}_n (x, y) \mathd x \right| \leqslant
      c^{\phi}_K 2^{- \beta n}.
    \end{equation}
  \end{enumerate}
\end{definition}

\begin{remark}
  We highlight that Definition \ref{regularisingkernelonmanifold} is slightly
  different from others which are used in the literature, \tmtextit{e.g.},
  {\cite[Sec. 7]{HS23}}. In particular, therein it is required that, for
  all $p \in M$ and, for all $n \in \mathbb{N}$ such that $\min_{d_g (p, q) \leqslant 1}
  \tmop{inj}_q \leqslant 2^{- n}$, it holds that
  \[ \mathsf{K}_n (0, \cdot) = \mathsf{K}_n (\cdummy, 0) = 0, \]
  where $\tmop{inj}_q$ is the injectivity radius at $q \in M$. In contrast, here
  we do not require any specific behavior of $\mathsf{K}_n$ near the singular
  points.
\end{remark}

\begin{proposition}
\label{Prop:kernel-and-exp}
In the aforementioned setting, the following two results hold true:
  \begin{enumerate}
    \item Let $\mathsf{K}$ be a regularising kernel as per Definition
    \ref{regularisingkernelonmanifold} and consider a compact exhaustion of
    $M$, $\{\Omega_n\}_{n \in \mathbb{N}}$ as per Definition \ref{Def: Compact Exhaustion}. For any geodesically convex set $U
    \subset \Omega_n$, with associated exponential map $\exp_p$ centered at $p \in U$, we have that $\exp^{\ast}_p \mathsf{K}$ is a regularising
    kernel of order $(m, r)$ and range $\rho_n /\mathfrak{C}_n$ on $\exp_p
    (U)$ as per Definition~\ref{Def: Regularizing Kernel}. Here
    $\mathfrak{C}_n$ is the constant in Equation
    \eqref{topologicalequivalencenormalcoordinateseq} relative to the compact
    set $K_n$, which is the $\frac{\mathcal{R}_C (\Omega_n)}{2}$ enlargement
    of $\bar{\Omega}_n$. Furthermore, $\rho_n$ is as per Equation \eqref{rhon}.
    
    \item Let $U \subset M$ be a geodesically convex set such that its closure
    $\bar{U}$ is still a geodesically convex set, and consider two local
    charts $(U, \exp_p^{- 1}), (U, \exp_q^{- 1})$ of $M$. If $\exp_{p
    \hspace{0.17em} \ast}^{- 1} \mathsf{K}$ abides by the properties $2, 3$ of Definition~\ref{regularisingkernelonmanifold}, then also $\exp_{q
    \hspace{0.17em} \ast}^{- 1} \mathsf{K}$ does, with $c_K^{\exp_{q \hspace{0.17em} \ast}^{- 1}} = \tmop{cst} c^{exp_{p
    \hspace{0.17em} \ast}^{- 1}}_K$ where $\tmop{cst} > 0$ is a constant
    depending only on $U$.
  \end{enumerate}
\end{proposition}

\begin{proof}
  In order to prove the first statement it suffices to show that Equation
  \eqref{kernelbound1onmanifold} implies item $\textit{1}$ of Definition
  \ref{Def: Regularizing Kernel}, since item $\textit{1}$, $\textit{3}$ and
  $\textit{4}$ of Definition \ref{regularisingkernelonmanifold} automatically
  imply items $\textit{2}$ and $\textit{3}$ of Definition \ref{Def:
  Regularizing Kernel}. By Proposition
  \ref{topologicalequivalencenormalcoordinates}, it follows that, if $d_g (q,
  q') \leqslant \rho_n$, then $\lvert \exp_p (q) - \exp_p (q') \rvert \leq
  \rho_n /\mathfrak{C}_n$. As a consequence
  \begin{equation}
    \tmop{supp} (\exp_p^{\ast} \mathsf{K}) \subseteq \left\{ (x, y) \in
    \exp_p^{- 1} (U) \times \exp_p^{- 1} (U) : \lvert x - y \rvert \leq
    \frac{\rho_n}{\mathfrak{C}_n} \right\} .
  \end{equation}

  To prove the second statement we denote by $\phi \doteq \exp^{- 1}_p, \psi
  \doteq \exp^{- 1}_q$. Given a compact set $K \subset \phi (U)$, we \ exploit
  Equation (\ref{kernelbound1onmanifold}) applied to $\psi_{\ast}
  \mathsf{K}_n$ to get
  \begin{equation}
    \begin{array}{lll}
      \lvert \partial_1^k \partial_2^l (\psi_{\ast} \mathsf{K}_n (\psi (p),
      \psi (q))) \rvert & = & | \partial_1^k \partial_2^l (\psi \circ \phi^{-
      1})_{\ast} \phi_{\ast} \mathsf{K}_n (\psi (p), \psi (q)) |\\
      & \leqslant & \| \psi \circ \phi^{- 1} \|_{C^{m + r} (\phi (\bar{U}))} 
      \lvert \partial_1^k \partial_2^l (\phi^{\ast} \mathsf{K}_n (\phi (p),
      \phi (q))) \rvert\\
      & \leqslant & \| \psi \circ \phi^{- 1} \|_{C^{m + r} (\phi (\bar{U}))}
      c^{\phi}_K 2^{(d - \beta + \lvert l \rvert + |k|) n} .
    \end{array} \label{variables}
  \end{equation}
  We stress that in the first line of Equation \eqref{variables},
  $\partial_1^k, \partial_2^l$ are taken with respect to the variables in $x,
  y \in \psi (U)$, while in the second line the partial derivatives are taken
  with respect to the variables $x', y' \in \phi(U)$. Proposition~\ref{Lem:smoothness of
  diffeomorphism} entails that the function $\psi \circ \phi^{- 1} = \exp_q
  \circ \exp_p^{- 1}$ is a smooth function with respect to the base points $p$
  and $q$. Therefore $\lvert \lvert \exp_q \circ \exp_p^{- 1} \rvert
  \rvert_{C^{m + r} (\phi (\bar{U}))}$ is a continuous function in $p$ and
  $q$. As a consequence, $\sup_{p, q \in \bar{U}} \| \exp_q \circ \exp_p^{- 1}
  \|_{C^{m + r} (\phi (\bar{U}))} < + \infty$ and it depends only on the
  compact set $\bar{U}$. Equation (\ref{kernelbound2onmanifold}) follows now
  by exploiting a similar line of reasoning. Let $n \in \mathbb{N}$ be the
  smallest number such that $U \subset \Omega_n$. Then,
  \[ \begin{array}{lll}
       \left| \int_{\psi (U)} (\psi (p) - \psi (q))^l \partial_2^k
       (\psi_{\ast} \mathsf{K}_n (\psi (p), \psi (q))) \mathd x \right| & = &
       \left| \int_{\phi (U)} (\phi (p) - \phi (q))^l \partial_2^k
       (\phi_{\ast} \mathsf{K}_n (\phi (p), \phi (q))) \lvert J_{(\psi \circ
       \phi^{- 1})} \rvert \mathd x \right|\\
       & \leqslant & \| J_{(\psi \circ \phi^{- 1})} \|_{C^0 (\phi (\bar{U}))}
       c^{\phi}_K 2^{- \beta n},
     \end{array} \]
  where $J_{(\psi \circ \phi^{- 1})}$ is the Jacobian associated to $\psi
  \circ \phi^{- 1}$. By Proposition \ref{Lem:smoothness of diffeomorphism} we
  have $\lvert \lvert J_{(\psi \circ \phi^{- 1})} \rvert \rvert_{C^0
  (\phi (\bar{U}))} \leqslant \tmop{cst}_U$, where $\tmop{cst}_U$ is a
  positive constant depending only on $U$. Therefore, we obtain $c_K^{\psi} =
  \tmop{cst} c^{\phi}_K$, where 
  $$\tmop{cst} \doteq \max \{\sup_{p, q \in
  \bar{U}} \| \exp_q \circ \exp_p^{- 1} \|_{C^{m + r} (\phi (\bar{U}))} \}.$$
\end{proof}

As we anticipated, in this section we outline the theory of germs of
distributions on Riemannian manifolds building on the construction of Section
\ref{Sec: Germs of distributions on open sets}. In particular, we need to adapt
the underlying building blocks to a Riemannian manifold, employing an
intrinsically geometric language, see Section \ref{Introduction} for an account of some key geometric definitions. To this end we consider a Riemannian
manifold $(M,g)$ and a finite good cover $\{U_i \}_{i \in I}$ as
per Definition \ref{Rem: Good Covering}. Fixing $U \in \{U_i\}_{i \in I}$, for every compact set $K \subset
U$, we define
\begin{equation}
  \label{Eq: Distance on a Manifold} D_{K} \assign \tmop{dist} (\exp^{- 1} (K),
  \partial \exp^{- 1} (U)),
\end{equation}
where $\exp$ denotes the exponential map, which is well defined since each
$U$ in the finite good cover can be chosen to be a geodesically convex subset. 
Inspired by
Equation \eqref{Eq: B^m}, we introduce the set of functions
\begin{equation}
  \label{Bronmanifold} \mathcal{B}^m_p (U) \assign \{\varphi \in
  \mathcal{D}(M) : \hspace{0.27em} \tmop{supp} (\varphi) \subset \{q \in M :
  d_g (p, q) < \rho_U \assign \min \{\mathcal{R}_C (U) / 2, 1\}\},
  \hspace{0.27em} \| \varphi \|_{C^m} \leqslant 1\},
\end{equation}
where $\mathcal{R}_C$ is the convexity radius as in Definition \ref{Def:
Convexity Radius}, while $\| \cdummy \|_{C^m}$ denotes the seminorm as in
Equation \eqref{Eq: C^m norm}. We observe that in Equation \eqref{Bronmanifold} $U$ can be replaced by a compact set $K\subset M$.  
As a consequence, Equation \eqref{Eq: Rescaled
Test-Function} can be rewritten as follows: Given $U$ as above and
$\varphi \in \mathcal{B}^m_p (U)$, we set for $q \in M$ and $\lambda \in (0,1)$,
\begin{equation}
  \label{evaluationonmanifolds} \varphi_p^{\lambda} (q) \assign \left\{
  \begin{array}{ll}
    \lambda^{- d} (\varphi \circ \exp_p) \left( \frac{\exp_p^{- 1}
    (q)}{\lambda} \right) \qquad & \tmop{if} \frac{d_g (p, q)}{\lambda}
    <\mathcal{R}_C (U) \\
    0 & \tmop{otherwise} 
  \end{array} \right. 
\end{equation}

\begin{remark}
  We observe that the above definition is well-posed and it yields a smooth
  function for any $\lambda \in (0, 1)$. Indeed, if $\frac{d_g
  (p, q)}{\lambda} <\mathcal{R}_C (U)$, then $d_g (p, q)
  <\mathcal{R}_C (U)$ and therefore $\tmop{exp}_p \left( \frac{\exp_p^{- 1}
  (q)}{\lambda} \right)$ is well-defined. Furthermore,
  for $\frac{d_g (p, q)}{\lambda} \geqslant \mathcal{R}_C (U)$, \tmtextit{a
  priori} $\exp_p \left( \frac{\exp_p^{- 1} (q)}{\lambda} \right)$ might
  not be well-defined, that is why we extend it to zero in this regime.
  Nevertheless, we stress that this extension matches in the right way
  with the smoothness of the rescaled function. Indeed, by
  Equation~\eqref{Bronmanifold}, $\tmop{supp} (\varphi)$ is a compact set
  strictly contained in $B (p, \mathcal{R}_C (U))$. Therefore we have that
  there exists $\varepsilon > 0$ such that the function $(\varphi \circ
  \exp_p) \left( \frac{\exp_p^{- 1} (q)}{\lambda} \right)$ is vanishing for
  \[ \mathcal{R}_C (U) - \varepsilon < \frac{d_g (p, q)}{\lambda}
     <\mathcal{R}_C (U) . \]
  This grants the possibility of extending in a smooth way the function to
  zero for $\frac{d_g (p, q)}{\lambda} >\mathcal{R}_C (U)$.
\end{remark}

With these novel structures at our disposal we can set our attention towards germs, looking
for a generalization to the case in hand of Definitions \ref{Def: Coherence}
and \ref{Def: Homogeneity}. We start from the latter observing that Definition
\ref{Def: germ} does not need to be modified except replacing the open set with $M$. In
the following definitions, we shall denote by $M$ a Riemannian manifold, leaving implicit the metric $g$,
and we let $F = (F_p)_{p \in M}$ be a germ of distributions thereon.

\begin{definition}[Coherence on Riemannian manifolds]
  \label{Def: Coherence on M}Let $\gamma \in \mathbb{R}$ and consider three
  families
  \[ \tmmathbf{\alpha}_{\gamma} \assign (\alpha_K)_{K \subset M} \subset
     \mathbb{R}, \qquad \tmmathbf{r} \assign (r_K)_{K \subset M} \subset
     \mathbb{N}_0, \qquad \tmmathbf{R} \assign (R_K)_{K \subset M} \subset
     \mathbb{R}^+, \]
  indexed by compact sets $K \subset M$, with $\alpha_K \leqslant \min \{0,
  \gamma\}, \forall K \subset M$. We say that $F$ is \textbf{$(\tmmathbf{\alpha},
  \gamma)$-coherent of order $\tmmathbf{r}$ and range $\tmmathbf{R}$} if for
  every compact set $K \subset M$, there exists $C_K > 0$ such that, for any
  fixed $\bar{\lambda} \in (0, \rho_K)$,
  \begin{equation}
    \label{coherenceonmanifold1} \lvert (F_p - F_q) (\varphi^{\lambda}_q)
    \rvert \leqslant C_K \lambda^{\alpha_K} (\lambda + d_g (p, q))^{\gamma -
    \alpha_K},
  \end{equation}
  uniformly for $\lambda \in (0, \bar{\lambda}]$, $\varphi \in
  \mathcal{B}_p^{r_K} (K)$ and for all $p, q \in K$ such that $d_g (p, q)
  \leqslant R_K$. The vector space of $(\tmmathbf{\alpha}, \gamma)$-coherent
  germs on $M$ is denoted by
  $\mathcal{G}^{\tmmathbf{\alpha}_{\gamma}}_{\tmmathbf{r}, \tmmathbf{R}} (M)$
  and it is endowed with the family of seminorms
  \begin{equation}
    \label{Eq: Coherence norm on M} \| F
    \|_{\mathcal{G}^{\tmmathbf{\alpha}_{\gamma}}_{\tmmathbf{r}, \tmmathbf{R}}
    (M), K} \assign \sup_{\tmscript{\begin{array}{c}
      p, q \in K\\
      d (p, q) \leqslant R_K\\
      \lambda \in (0, \bar{\lambda}], \hspace{0.17em} \varphi \in
      \mathcal{B}_p^{r_K} (K)
    \end{array}}}  \hspace{0.27em} \frac{\lvert (F_p - F_q)
    (\varphi^{\lambda}_p) \rvert}{\lambda^{\alpha_K} (\lambda + d (p,
    q))^{\gamma - \alpha_K}} .
  \end{equation}
\end{definition}

\begin{definition}[Homogeneity on Riemannian manifolds]
  \label{Def: Homogeneity on M}Let
  \[ \bar{\tmmathbf{\alpha}} \assign (\bar{\alpha}_K)_{K \subset M} \subset
     \mathbb{R}, \qquad \tmmathbf{r} \assign (r_K)_{K \subset M} \subset
     \mathbb{N}_0, \]
  be two families indexed by compact sets $K \subset M$, we say that $F$ is
  \textbf{$\tmmathbf{\bar{\alpha}}$-homogeneous of order $\tmmathbf{r}$}, if for any $K
  \subset M$, and for any $\bar{\lambda} \in (0, \rho_K)$,
  \begin{equation}
    \label{homogeneityonmanifolds} \lvert F_p (\varphi^{\lambda}_p) \rvert
    \leqslant C_K \lambda^{\bar{\alpha}},
  \end{equation}
  uniformly for $p \in K$, $\lambda \in (0, \bar{\lambda}]$ and $\varphi \in
  \mathcal{B}_p^{r_K} (K)$. The space of $\tmmathbf{\bar{\alpha}}$-homogeneous
  germs of distributions on $M$ of order $\tmmathbf{r}$ is
  denoted by $\mathsf{G}^{\tmmathbf{\bar{\alpha}}}_{\tmmathbf{r}} (M)$ and it
  is endowed with the family of seminorms
  \begin{equation}
    \label{Eq: Homogeneity norm on M}
    \mathsf{G}^{\tmmathbf{\bar{\alpha}}}_{\tmmathbf{r}} (M) \ni F \mapsto \| F
    \|_{\mathsf{G}^{\tmmathbf{\bar{\alpha}}}_{\tmmathbf{r}} (M), K} \assign
    \sup_{\tmscript{\begin{array}{c}
      \lambda \in (0, \bar{\lambda}], \hspace{0.27em}\\
      p \in K, \varphi \in \mathcal{B}_p^{r_K} (K)
    \end{array}}} \frac{\lvert F_p (\varphi^{\lambda}_p)
    \rvert}{\lambda^{\bar{\alpha}_K}} .
  \end{equation}
\end{definition}

\noindent At last, we can merge the two previous definitions into the following one.

\begin{definition}[Coherence and homogeinity on Riemannian manifolds]
  \label{Def: Coherence and Homogeneity on M}We say that $F = (F_p)_{p \in M}$
  is a \textbf{$(\tmmathbf{\alpha}, \gamma)$-coherent with homogeneity
  $\tmmathbf{\bar{\alpha}}$} if $F \in \mathcal{G}^{\tmmathbf{\bar{\alpha}},
  \tmmathbf{\alpha}_{\gamma}}_{\tmmathbf{r}, \tmmathbf{R}} (M) \assign
  \mathsf{G}^{\tmmathbf{\bar{\alpha}}}_{\tmmathbf{r}} (M) \cap
  \mathcal{G}^{\tmmathbf{\alpha}_{\gamma}}_{\tmmathbf{r}, \tmmathbf{R}} (M)$.
  This space can be endowed with the family of seminorms
  \begin{equation}
    \label{Eq: Full Seminorms on M} \mathcal{G}^{\tmmathbf{\bar{\alpha}},
    \tmmathbf{\alpha}_{\gamma}}_{\tmmathbf{r}, \tmmathbf{R}} (M) \ni F \mapsto
    \| F \|_{\mathcal{G}^{\tmmathbf{\bar{\alpha}},
    \tmmathbf{\alpha}_{\gamma}}_{\tmmathbf{r}, \tmmathbf{R}} (M), K} \assign
    \| F \|_{\mathcal{G}^{\tmmathbf{\alpha}_{\gamma}}_{\tmmathbf{r},
    \tmmathbf{R}} (M), K} + \| F
    \|_{\mathsf{G}^{\tmmathbf{\bar{\alpha}}}_{\tmmathbf{r}} (M), K} .
  \end{equation}
\end{definition}

{\noindent}The next propositions relate Definitions \ref{Def: Coherence on M}
and \ref{Def: Homogeneity on M}, with the ones already present in the
literature {\cite[Def. 4 \& Prop. 29]{RS21}}. In particular, we prove that the
notion of coherence introduced in the present paper is compatible with the one
of {\cite{RS21}}. 

\begin{proposition}
\label{mineimplyrs}
Let $M$ be a Riemannian manifold and $\{U_i \}_{i
  \in I}$ a finite good cover as per Definition \ref{Rem: Good Covering}.
  \begin{enumerate}
    \item Given $F \in \mathcal{G}^{\tmmathbf{\alpha}_{\gamma}}_{\tmmathbf{r},
    \tmmathbf{R}} (M)$ as per Definition \ref{Def: Coherence on M}, for every
    $i \in I$, for every compact set $K \subset U_i$ and for every
    $\bar{\lambda} \in (0, D_K)$ as in Equation \eqref{Eq: Distance on a
    Manifold}, there exist a constant $C \equiv C (K, \bar{\lambda}) > 0$ and
    $\alpha_K \in \mathbb{R}$, such that
    \begin{equation}
      \label{coherenceboundonmanifold} | (\exp^{- 1}_{i \ast} F_p - \exp^{-
      1}_{i \ast} F_q) (\varphi^{\lambda}_{\exp^{- 1}_i (q)}) | \leqslant C
      (\lambda + \lvert \exp^{- 1}_i (p) - \exp^{- 1}_i (q) \rvert)^{\gamma -
      \alpha_K} .
    \end{equation}
    This bound is uniform over $\lambda \in (0, \bar{\lambda}]$, $p, q \in K$,
    $\varphi \in \mathcal{B}^r$, where $\mathcal{B}^r$ is as per Equation
    \eqref{Eq: B^m}.
    
    \item Given $F \in \mathsf{G}^{\tmmathbf{\bar{\alpha}}}_{\tmmathbf{r}}
    (M)$ as per Definition \ref{Def: Homogeneity on M}, for every $i \in I$,
    for every compact set $K \subset U_i$ and for every $\bar{\lambda} \in (0,
    D_K)$ as in Equation \eqref{Eq: Distance on a Manifold}, there exist a
    constant $C' \equiv C' (K, \bar{\lambda}) > 0$ and $\bar{\alpha}_K \in
    \mathbb{R}$, such that
    \[ | \exp^{- 1}_{i \ast} (F_q) (\varphi^{\lambda}_{\exp^{- 1}_i (q)}) |
       \leqslant C' \lambda^{\bar{\alpha}_K} . \]
    This bound is uniform over $\lambda \in (0, \bar{\lambda}]$, $\varphi \in
    \mathcal{B}^r$ and $q \in K$.
  \end{enumerate}
\end{proposition}

\begin{remark}
In the proof of Proposition \ref{mineimplyrs}, we shall use extensively some of the results in Appendices \ref{appendixB} and \ref{AppendixC}. In particular, we will need Proposition~\ref{topologicalequivalencenormalcoordinates}. For sake of
readability, we recall here that, under the above assumptions, this result grants us that for any compact set $K \subset U_i$, with $U_i$ a
geodesically convex open set, there exist constants $\mathfrak{C}_K,
\mathfrak{C}_K' > 0$ such that
\[ \mathfrak{C}_K' d_g (q, r) \leqslant \lvert \exp^{- 1}_i (q) - \exp^{-
   1}_i (r) \rvert \leqslant \mathfrak{C}_K d_g (q, r), \qquad \forall q, r
   \in U_i . \]
\end{remark}

\begin{proof}
  We only prove the first statement, the second one following suit. In view of
  Definition \ref{Rem: Good Covering}, we can require each $U_i$ to be a
  geodesically convex, open subset of $M$. In addition, we fix a compact set $K
  \subset U_i$, $q \in K$ and $\bar{\lambda} \in (0, \rho_K)$, see Definition
  \ref{Def: Convexity Radius}. Setting $\bar{\lambda}_K \assign \min \{\rho_K,
  D_K \}$, we consider $p, q \in U_i$ so that $\lvert \exp_i^{- 1} (p) -
  \exp_i^{- 1} (q) \rvert \leqslant R_K \mathfrak{C}_K'$, where $| \cdummy |$
  is the Euclidean distance on $\exp_i^{- 1} (U_i) \subseteq \mathbb{R}^d$, $d
  = \dim M$, while $R_K$ is the range of the germ as per Definition \ref{Def:
  Coherence on M}. In view of Proposition
\ref{topologicalequivalencenormalcoordinates}, it descends that $d_g (p, q)
  \leqslant R_K$ where $d_g$ is the Riemannian distance.
  
  Given a test-function $\varphi \in \mathcal{B}^{r_K}$, as per Equation
  \eqref{Eq: B^m}, and $\lambda \in (0, 1]$, we exploit
  Lemma~\ref{fromflattomanifold} to infer that $\exp_{i \ast} \left(
  \varphi^{\lambda \hspace{0.17em} \bar{\lambda}}_{\exp_i^{- 1} (q)} \right) =
 \phi^{\lambda}_{[\lambda], q}$, with $[\lambda] \assign (\lambda, \bar{\lambda})$, for a suitable $ \phi_{[\lambda], q} \in
  \bar{\lambda}^{- r_K - d} \mathcal{B}^{r_K}_p (K)$, see
  Equation~\eqref{Bronmanifold}. 
  Here $r_K$ is the order of the germ. Hence, Lemma~\ref{fromflattomanifold} entails that the push-forward
  through $\exp_i$ sends $\varphi^{\lambda \hspace{0.17em}
  \bar{\lambda}}_{\exp_i^{- 1} (q)}$ to a good test-function on the manifold,
  up to a constant. Since $F \in
  \mathcal{G}^{\tmmathbf{\alpha}_{\gamma}}_{\tmmathbf{r}, \tmmathbf{R}} (M)$,
  we apply Equation \eqref{coherenceonmanifold1} with
  $\phi_{[\lambda], q}^{\lambda}$ as test function, obtaining
  \begin{equation}
    \begin{array}{lll}
      | (\exp^{- 1}_{i \ast} F_p - \exp^{- 1}_{i \ast} F_q) (\varphi^{\lambda
      \bar{\lambda}}_{\exp_i^{- 1} (p)}) | & = & \lvert (F_p - F_q)
      (\phi_{[\lambda], q}^{\lambda}) \rvert\\
      & \leqslant & \bar{\lambda}^{- r - d} \| F
      \|_{\mathcal{G}^{\tmmathbf{\alpha}_{\gamma}}_{\tmmathbf{r},
      \tmmathbf{R}} (M), K} \lambda^{\alpha_K} (\lambda + d_g (p, q))^{\gamma
      - \alpha_K}\\
      & \lesssim & \| F
      \|_{\mathcal{G}^{\tmmathbf{\alpha}_{\gamma}}_{\tmmathbf{r},
      \tmmathbf{R}} (M), K} (\lambda \bar{\lambda})^{\alpha_K} (\lambda
      \bar{\lambda} + d_g (p, q))^{\gamma - \alpha_K}\\
      & \lesssim & \| F
      \|_{\mathcal{G}^{\tmmathbf{\alpha}_{\gamma}}_{\tmmathbf{r},
      \tmmathbf{R}} (M), K} (\lambda \bar{\lambda})^{\alpha_K} (\lambda
      \bar{\lambda} + \lvert \exp^{- 1} (p) - \exp^{- 1} (q) \rvert)^{\gamma -
      \alpha_K},
    \end{array} \label{Equationabove}
  \end{equation}
  uniformly for $\lambda \in (0, 1]$, $\varphi \in \mathcal{B}^{r_K}$, $p, q
  \in K$, such that $\lvert \exp_i^{- 1} (p) - \exp_i^{- 1} (q) \rvert
  \leqslant R_K \mathfrak{C}_K'$. In the second line we applied the
  definition of coherence as per Equation \eqref{coherenceonmanifold1}, whilst in the
  third one we exploited the estimate $(A + B)^C \lesssim (AP + B)^C$, which holds true
  uniformly for $A, B > 0$ bounded from above and for $C, P > 0$ fixed. Finally, in the
  fourth line we used Equation \eqref{topologicalequivalencenormalcoordinateseq}. Equation
  \eqref{Equationabove} holds true for $\bar{\lambda} \in (0, \min \{\rho_K,
  D_K \})$, but we can extend its validity to all $\bar{\lambda} \in (0, D_K)$ using
  Proposition \ref{relaxingbound11}, and for all $p, q \in K$ using
  Proposition \ref{relaxingbound2}, see Equation \eqref{alphak}. Thus, Equation \eqref{coherenceboundonmanifold} descends.
\end{proof}

A natural question to ask at this stage is whether we can establish a converse of Proposition
\ref{mineimplyrs}, \tmtextit{i.e.}, if, given a coherent germ on a
Riemannian manifold $M$ as in {\cite[Def. 4]{RS21}}, one obtains a germ as
per Definition \ref{Def: Coherence on M} with a prescribed range $\tmmathbf{R}$. The answer is affirmative, as codified in the following proposition.

\begin{proposition}
  \label{rinaldiimpliesmine}Let $M$ be a Riemannian manifold and let
  $\{U_i\}_{i \in I}$ be a finite good cover of $M$. Let $\gamma \in \mathbb{R}$ and
  let $F$ be a germ of distributions on $M$.
  \begin{itemize}
    \item[1.] Assume that for every compact set $K \subset U_i \in \{U_j\}_{j \in
    I}$ and for every $\bar{\lambda} \in (0, D_K)$, where $D_K$ is as per
    Equation \eqref{Eq: Distance on a Manifold}, there exist constants
    $\tmop{cst}_K > 0$, $\alpha_K^{U_i} < \gamma$, for which the following
    bound holds true
    \begin{equation}
      \label{coherenceasrinaldisclavi} | (\exp^{- 1}_{i \ast} F_p - \exp^{-
      1}_{i \ast} F_q) (\varphi^{\lambda}_{\exp_i^{- 1} (p)}) | \leqslant
      \tmop{cst}_K \lambda^{\alpha_K} (\lambda + \lvert \exp^{- 1} (p) -
      \exp^{- 1} (q) \rvert)^{\gamma - \alpha_K},
    \end{equation}
    uniformly for $\lambda \in (0, \bar{\lambda}]$, $\varphi \in
    \mathcal{B}^{r_K}$, $p, q \in K$. Then, for any compact set $K \subset M$,
    there exist constants $R_K > 0$ and $\alpha'_K < \min \{ 0, \gamma \}$
    such that $F$ is an $(\alpha_K', \gamma)$-coherent germ of order $r_K$ and
    range $R_K$, as per Definition~\ref{Def: Coherence on M}. Additionally,
    $\| F \|_{\mathcal{G}^{\tmmathbf{\alpha}'_{\gamma}}_{\tmmathbf{r},
    \tmmathbf{R}} (M), K}$ depends linearly on $\tmop{cst}_K$.
    
    \item[2.] Assume that for every compact set $K \subset U_i \in \{U_j\}_{j \in
    I}$ and for every $\bar{\lambda} \in (0, D_K)$, where $D_K$ is as per
    Equation \eqref{Eq: Distance on a Manifold}, there exist constants
    $\tmop{cst}'_K > 0$ and $\bar{\alpha}_K \in \mathbb{R}$, for which the
    following bound holds true:
    \begin{equation}
      \label{coherenceasrinaldisclavi2} | (\exp^{- 1}_{i \ast} F_p)
      (\varphi^{\lambda}_{exp_i^{- 1} (p)}) | \leqslant \tmop{cst}_K'
      \lambda^{\bar{\alpha}_K},
    \end{equation}
    uniformly for $\lambda \in (0, \bar{\lambda}]$, $\varphi \in
    \mathcal{B}^r$, $p \in K$. Then for any compact set $K \subset M$ there
    exists a constant $r_K$ such that $F$ is an $\bar{\alpha}_K$-homogeneous
    germ of order $r_K$, as per Definition \ref{Def: Homogeneity on M}.
    Additionally, $\| F
    \|_{\mathsf{G}^{\tmmathbf{\bar{\alpha}}}_{\tmmathbf{r}} (M), K}$ depends
    linearly on $\tmop{cst}'_K$.
  \end{itemize}
\end{proposition}

\begin{proof}
  We only prove item $1.$, as the proof of item $2.$ follows by an analogous argument. First of all, we show that
  Equation \eqref{coherenceonmanifold1} holds true for every compact set contained
  in a single geodesically convex open set $U_i$. Then, we shall
  generalize the result to all compact sets by restricting the range. \\
  \noindent Let us
  fix a compact set $K \subset U_i \in \{U_j\}_{j \in I}$, with the associated
  exponential map $\exp_i$. Let $p, q \in K$ and $\varphi \in \mathcal{B}^r_p
  (K)$ as per Equation~\eqref{Bronmanifold}. Let $\bar{\lambda} \in (0, \min
  \{\rho_K, D_K \mathfrak{C}_K, \mathfrak{C}_K \})$. Reasoning as in the proof of Proposition \ref{mineimplyrs}, Lemma \ref{frommanifoldtoflat} entails that $\forall \lambda
  \in (0, 1)$ we can write $\exp^{- 1}_{i \ast} (\varphi_q^{\lambda
  \bar{\lambda}}) = \psi_{[\lambda], \exp_i^{- 1} (q)}^{\lambda}$, for a
  suitable $\psi_{[\lambda]} \in \mathcal{B}^r$ with $[\lambda] \assign (\lambda, \bar{\lambda})$, see Equation \eqref{Eq: B^m}. Thus, 
  \begin{equation}
    \begin{array}{lll}
      \lvert (F_p - F_q) (\varphi_p^{\lambda \bar{\lambda}}) \rvert & = & |
      (\exp^{- 1}_{i \ast} F_p - \exp^{- 1}_{i \ast} F_q) (\exp^{- 1}_{i \ast}
      (\varphi_p^{\lambda \bar{\lambda}}) \nobracket |\\
      & = & | (\exp^{- 1}_{i \ast} F_p - \exp^{- 1}_{i \ast} F_q)
      (\psi_{[\lambda], \exp_i^{- 1} (q)}^{\lambda}) |\\
      & \leqslant & C_K \bar{\lambda}^{- r - d} \lambda^{\alpha_K} (\lambda +
      \lvert \exp_i^{- 1} (p) - \exp_i^{- 1} (q) \rvert)^{\gamma - \alpha_K}\\
      & \lesssim & C_K (\lambda \bar{\lambda})^{\alpha_K} (\lambda
      \bar{\lambda} + \lvert \exp_i^{- 1} (p) - \exp_i^{- 1} (q)
      \rvert)^{\gamma - \alpha_K}\\
      & \lesssim & C_K (\lambda \bar{\lambda})^{\alpha_K} (\lambda
      \bar{\lambda} + d_g (p, q))^{\gamma - \alpha_K},
    \end{array} \label{rinaldiimpliesmineeq}
  \end{equation}
  uniformily for $\lambda \in (0, 1]$, $p, q \in K$, $\varphi \in
  \mathcal{B}^r_p (K)$. In particular, in the third line we exploited
  Equation~\eqref{coherenceasrinaldisclavi}, whereas in the fourth and fifth
  lines we used Equation (\ref{topologicalequivalencenormalcoordinateseq}) together with $(A + B)^C \lesssim (AP + B)^C$, uniformly for $A,B > 0$ bounded from above and for $C, P > 0$, with $P = \bar{\lambda}$. We can now extend
  Equation (\ref{rinaldiimpliesmineeq}) to all $\bar{\lambda} \in (0, \rho_K)$
  by means of Proposition \ref{extendingthebound1onmanifold}.
  
  Equation \eqref{rinaldiimpliesmineeq} has been established for a compact set $K
  \subset U_i \in \{U_j\}_{j \in I}$. The next step is thus to consider a compact
  set $K$ which is covered by two sets $U_i, U_{\ell}$ of the finite good
  covering $\{U_j\}_{j \in I}$. We denote the associated exponential maps by
  $\exp_i$ and $\exp_{\ell}$, respectively. In this scenario, we define $K \assign K_i
  \cup K_{\ell}$, with $K_i$ and $K_{\ell}$ compact sets such that $K_i
  \subset U_i$ and $K_{\ell} \subset U_{\ell}$ and we omit the dependence on the indices $k, \ell$ for notational ease. 
  
  On account of what we already proved, we know that the thesis holds true
  whenever $p, q$ lie both in $K_i$ or in $K_{\ell}$. The only case to discuss
  is when, \tmtextit{e.g.}, $p \in K_i$ and $q \in K_{\ell}$. Calling $D_{K_i}
  = \tmop{dist} (\exp_i^{- 1} (K_i), \partial U_i)$, we define
  $R_{K_i} \assign \frac{D_{K_i}}{2\mathfrak{C}_{K_i}}$, where
  $\mathfrak{C}_{K_i}$ is the constant in
Equation~\eqref{topologicalequivalencenormalcoordinateseq} relative to the compact
  set $K_i$. Then, if we require that $d_g (p, q) \leqslant R_{K_i}$, by
  Proposition \ref{topologicalequivalencenormalcoordinates}, it follows that
  $\lvert \exp_i^{- 1} (p) - \exp_i^{- 1} (q) \rvert \leqslant D_{K_i} / 2$.
  Hence, these points lie in the $D_{K_i} / 2$ enlargement of $K_i \subset
  U_i$. We call such a set $K'_i$. We highlight that $K'_i$ is a compact set
  still contained in $U_i$. As a consequence, since $p, q \in K'_i$, we can apply the coherence condition to the
  compact set $K'_i \subset U_i$. We can now use the same argument for
  $K_{\ell}$, getting $R_{K_{\ell}}$ and $K'_{\ell}$. Thus, for $K \assign K_i \cup
  K_{\ell}$, choosing the parameters as
  \begin{equation}
    \label{alphacappa'} R_K \assign \min \{R_{K_i}, R_{K_{\ell}} \}, \qquad
    \alpha_K' \assign \min \{\alpha_{K'_i}, \alpha_{K'_{\ell}} \},
  \end{equation}
  we conclude the validity of Equation \eqref{rinaldiimpliesmineeq} for all
  $p, q \in K$ such that $d_g (p, q) \leqslant R_{K }$. The general case, namely
  for $K$ covered by a finite number of geodesic sets, can be proven by iterating
  this procedure.
  
Since we have shown that the coherence bound in Equation \eqref{rinaldiimpliesmineeq} holds true for
  all
  \[ \bar{\lambda} \in (0, \min \{\rho_K, D_K \mathfrak{C}_K, \mathfrak{C}_K
     \}) , \]
we can extend such bound using Proposition
  \ref{extendingthebound1onmanifold}, hence obtaining the sought result. 
\end{proof}

\noindent We conclude this section by stating a result equivalent to Proposition \ref{homcoherence} on
Riemannian manifolds.

\begin{proposition}[Necessity of coherence on manifolds]
  \label{homocoherenceonmanifolds} Let $M$ be a Riemannian manifold and
  let $F$ be a $\gamma$-homogeneous germ of order $\tmmathbf{r}$ on $M$, as
  per Definition~\ref{Def: Homogeneity on M}. Then, there exists a family $\tmmathbf{R} = (R_K)_{K \subset M} \subset
  \mathbb{R}_+$ indexed by compact subsets of $M$ so that $F$ is an
  $(\tmmathbf{\alpha}, \gamma)$-coherent germ of order $\tmmathbf{r}$ and
  range $\tmmathbf{R}$, where $\alpha_K = \min \{- r_K - d, \gamma\}$.
  Additionally, the coherence seminorm relative to the compact set $K$ is
  proportional to the homogeneity seminorm.
\end{proposition}

\begin{proof}
  We prove that, given a compact set $K \subset M$ and $\bar{\lambda} \in (0,
  \rho_K)$, there exist $\alpha_K \leqslant \min \{0, \gamma\}$ and $R_K > 0$
  so that
  \begin{equation}
    \label{necessityonmanifoldeq1} \lvert (F_p - F_q) (\varphi^{\lambda}_q)
    \rvert \lesssim_K \lambda^{\alpha_K} (\lambda + d_g (p, q))^{\gamma -
    \alpha_K},
  \end{equation}
  uniformly for $\lambda \in (0, \bar{\lambda})$, $p, q \in K$ with $d_g (p, q)
  \leqslant R_K$, while $\varphi \in \mathcal{B}^{r_K}_q (K)$, as per
  Equation~\eqref{Bronmanifold}.
  
  We already know that, for any $\alpha \leqslant \gamma$,
  \begin{equation}
    \label{necessityonmanifoldeq1.5} \lvert F_q (\varphi^{\lambda}_q) \rvert
    \lesssim_K \lambda^{\gamma} \lesssim \lambda^{\alpha} (\lambda + d_g (p,
    q))^{\gamma - \alpha} .
  \end{equation}
  In order to find an estimate for $\lvert F_p (\varphi^{\lambda}_q) \rvert$,
  we can use Lemma~\ref{recenteringonmanifolds}. On account of this result, we
  can infer that there exists a $K$-dependent parameter $\lambda_K$ such that, if $\lambda \in (0, \lambda_K)$, we can recenter the function
  $\varphi^{\lambda}_q$ as $\varphi^{\lambda}_q = \xi^{\lambda''}_{[\lambda],
  \hspace{0.17em} q}$, with $\lvert \lvert \xi_{[\lambda], \hspace{0.17em} q}
  \rvert \rvert_{C^{r_K}} \lesssim_K \left( \frac{\lambda}{\lambda + d (p, q)}
  \right)^{- r_K - d}$. Here $\lambda''$ is a constant established in
  Lemma~\ref{recenteringonmanifolds}. We get
  \begin{equation}
    \lvert F_p (\varphi^{\lambda}_q) \rvert = \lvert F_p (\xi^{\lambda''}_p)
    \rvert \lesssim_K \left( \frac{\lambda}{\lambda + d (p, q)} \right)^{- r_K
    - d} \lambda^{'' \gamma}  \lesssim_K \lambda^{- r_K - d} (\lambda
    + d (p, q))^{\gamma + r_K + d}, \label{necessityonmanifoldeq2}
  \end{equation}
  uniformly for $\lambda \in (0, \lambda_K)$, $\varphi \in \mathcal{B}^r_q
  (K)$. Here $p, q \in K$ are such that $d_g (p, q) \leqslant
  \widetilde{\tmop{cst}}_K$, where, once more $\widetilde{\tmop{cst}}_K > 0$
  is a constant introduced in Lemma~\ref{recenteringonmanifolds}. Putting
  together Equation \eqref{necessityonmanifoldeq1.5} and Equation
  \eqref{necessityonmanifoldeq2}, we obtain
  \[ \lvert (F_p - F_q) (\varphi^{\lambda}_q) \rvert \lesssim_K
     \lambda^{\alpha_K} (\lambda + d (p, q))^{\gamma - \alpha_K}, \]
  with $\alpha_K = - r_K - d$. We can now extend the result to all
  $\bar{\lambda} \in (0, \rho_K)$ using
  Proposition~\ref{extendingthebound1onmanifold}.
\end{proof}

\section{Reconstruction Theorem on Riemannian
manifolds}\label{sectionReconstructionManifold}
Having introduced the notion of coherent and homogeneous germs of distributions on Riemannian manifolds, see Section \ref{Sec:Germs-of-distributions-manifold}, we shall now state a reconstruction result, tailored to this setting. 

In addition to the
existing literature {\cite[Thm. 18]{RS21}}, we prove the continuity of the
reconstruction operator and we assess the regularity of the reconstructed
distribution in the sense of negative H{\"o}lder-Zygmund spaces on Riemannian
manifolds, whose definition we recall here for the reader's convenience.

\begin{definition}[Negative H{\"o}lder-Zygmund spaces on Riemannian
manifolds]
  \label{holderzyngmundManifold}
  Let $M$ be a Riemannian manifold, and let
  $\gamma \in \mathbb{R}$. We define $\mathcal{Z}^{\gamma} (M)$ as the set of
  distributions $f \in \mathcal{D}' (M)$ such that
  \begin{equation}
    \label{zygmundholderspaceseq} \| f \|_{\mathcal{Z}^{\gamma}_{K,
    \bar{\lambda}}} = \sup_{\tmscript{\begin{array}{c}
      \lambda \in (0, \bar{\lambda}], \hspace{0.27em} p \in K\\
      \varphi \in \mathcal{B}^r_p (K)  \text{with} \hspace{0.27em} r = \lfloor
      - \gamma + 1 \rfloor
    \end{array}}} \frac{\lvert f (\varphi^{\lambda}_p)
    \rvert}{\lambda^{\gamma}} < + \infty,
  \end{equation}
  for all compact sets $K \subset M$ and for all $\bar{\lambda} \in (0,
  \rho_K)$, where $\rho_K$ and the spaces $\mathcal{B}^r_p (K)$ are as per
  Equation~\eqref{Bronmanifold}.
\end{definition}

\begin{remark}
  We define the H{\"o}lder-Zygmund spaces only in the case $\gamma \leqslant
  0$ as the definition for positive $\gamma$ would require additional
  structures such as the notion of jets. Nevertheless, since the homogeneity
  of the germ corresponds to the H{\"o}lder-Zygmund regularity of the
  reconstruction, see Theorem \ref{reconstructiononmanifolds}, and since germs
  with positive homogeneity have the identically zero distribution as
  reconstruction, see Remark \ref{negativehomogeneity}, the case $\gamma > 0$
  is not of interest for our goals.
\end{remark}

\begin{theorem}[Reconstruction theorem on Riemannian manifolds]
  \label{reconstructiononmanifolds}Let $M$ be a Riemannian manifold, and
  let $\{U_i\}_{i \in I}$ be a finite good cover of $M$ as per Definition \ref{Rem: Good Covering}. Consider an
  $(\tmmathbf{\alpha}, \gamma)$-coherent germ $F$ of order $\tmmathbf{r}$ and
  range $\tmmathbf{R}$ on $M$. If $\gamma > 0$, there exists a unique
  $\mathcal{R}^{\gamma} F \in \mathcal{D}' (M)$ such that, for every compact
  set $K \subset U_j$, with $j \in I$, and for every $\bar{\lambda} \in (0,
  \rho_K)$ there exists a compact set $\tilde{K}$ such that $K \subset
  \tilde{K} \subset U_j$ and
  \begin{equation}
    \label{reconstructionboundonmanifolds} \lvert (F_p -\mathcal{R}^{\gamma}
    F) (\varphi^{\lambda}_p) \rvert \lesssim \| F
    \|_{\mathcal{G}^{\tmmathbf{\alpha}_{\gamma}}_{\tmmathbf{r}, \tmmathbf{R}}
    (M), \tilde{K}} \lambda^{\gamma}
  \end{equation}
  uniformly over $\lambda \in (0, \bar{\lambda})$, $\varphi \in
  \mathcal{B}_p^{m_K} (K)$ and $p \in K$. Furthermore, it holds true that
  \begin{itemize}
    \item[1.] The map $F \mapsto \mathcal{R}^{\gamma} F$ is linear,
    
    \item[2.] If $F$ is also $\bar{\tmmathbf{\alpha}}$-homogeneous, as per
    Definition~\ref{Def: Homogeneity on M}, for each $U_j \in \{U_i\}_{i \in
    I}$, denoting by $K_j$ the closure of $U_j$, $\mathcal{R}^{\gamma} F \in
    \mathcal{Z}^{\alpha_{K_j}} (U_j)$, as per Definition
    \ref{holderzyngmundManifold}, and we have
    \begin{equation}
      \label{ReconstructionisholderManifold} \| \mathcal{R}^{\gamma} F
      \|_{\mathcal{Z}^{\gamma}_{K, \bar{\lambda}}} \lesssim_{\alpha_K,
      \bar{\lambda}, d} \| F
      \|_{\mathcal{G}^{\tmmathbf{\alpha}_{\gamma}}_{\tmmathbf{r},
      \tmmathbf{R}} (M), \tilde{K}} .
    \end{equation}
  \end{itemize}
\end{theorem}

\begin{proof}
  The strategy of the proof consists in exploiting the analogous results on
  open sets of $\mathbb{R}^d$, Theorem \ref{Thm: Reconstruction Theorem}, and
  then show that the exponential map yields a reconstruction on the whole
  manifold.
  
  Given a finite good cover $\{U_i\}_{i \in I}$ and taking $U_j \in \{U_i\}_{i \in I}$,
  by Proposition \ref{mineimplyrs} $\exp_{j \ast}^{- 1} F$ satisfies a
  coherence estimate for all points $\exp_j^{- 1} (p), \exp_j^{- 1} (q)$ such
  that $d_g (p, q) \leqslant R_K$. We can always assume that $\tmop{diam}
  (U_j) \leqslant R_K$, therefore we have that $\exp_{j \ast}^{- 1} F$ is an
  $(\tilde{\tmmathbf{\alpha}}, \gamma)$-coherent germ on $U_j' \assign
  \exp_j^{- 1} U_j$, as per Definition \ref{Def: Coherence}.
  
  By Theorem \ref{Thm: Reconstruction Theorem}, there exists a unique
  distribution which we denote by $f_j \assign \mathcal{R}^{\gamma} \exp_{j
  \ast}^{- 1} F \in \mathcal{D}' (U_j')$, such that, fixed a compact set $K
  \subset U_j$ and $\bar{\lambda} \in (0, \tmop{dist} (K, \partial U_j))$, it
  holds that
  \[ \lvert (\exp_{j \ast}^{- 1} F_p - f_j) (\varphi^{\lambda}_{\exp_j^{- 1}
     (p)}) \rvert \lesssim \| \exp_{j \ast}^{- 1} F
     \|_{\mathcal{G}^{\tmmathbf{\alpha}_{\gamma}}_{\tmmathbf{r}} (U'_j), K'_{3
     D_{K'} / 4}} \lambda^{\gamma}, \]
  uniformly for $p \in K$, $\varphi \in \mathcal{B}^{m_{K'}}$ and $\lambda \in
  (0, \bar{\lambda})$. Here we denote by $K' \assign \exp_j^{- 1} K$ a compact set contained in $U_j'$, while $K'_{3 D_{K'} / 4}$ is its
  $(3 D_{K'} / 4)-$enlargement with $D_{K'} \assign \tmop{dist} (K', \partial
  U_j')$. Recall that on account of Proposition \ref{mineimplyrs}, we have
  \[ \lvert \lvert \exp_{j \ast}^{- 1} F \rvert
     \rvert_{\mathcal{G}^{\tmmathbf{\alpha}_{\gamma}}_{\tmmathbf{r}} (U'_j),
     K'_{3 D_{K'} / 4} \tmcolor{red}{}} \leqslant \| F
     \|_{\mathcal{G}^{\tmmathbf{\alpha}_{\gamma}}_{\tmmathbf{r}, \tmmathbf{R}}
     (M), \tilde{K}}, \]
  where $\tilde{K}  \assign \exp_j
  (K'_{3 D_{K'} / 4}) \subset U_j$ is a compact set. Overall, this implies that
  \begin{equation}
    \lvert (\exp_{j \ast}^{- 1} F_p - f_j) (\varphi^{\lambda}_{\exp_j^{- 1}
    (p)}) \rvert \lesssim \| F
    \|_{\mathcal{G}^{\tmmathbf{\alpha}_{\gamma}}_{\tmmathbf{r}, \tmmathbf{R}}
    (M), \tilde{K}} \lambda^{\gamma} .
    \label{reconstructionboundonmanifoldsproof}
  \end{equation}
  We define the reconstruction of $F$ restricted to $U_j \subset M$ as
  \[ \mathcal{R}^{\gamma}_j F \assign \exp_{j \ast} f_j . \]
  So far, we have built a family of distributions $\{ \mathcal{R}^{\gamma}_j F \}_{j
  \in I}$ such that $\mathcal{R}^{\gamma}_j F \in \mathcal{D}' (U_j)$ for any
  $j \in I$. Our next goal is to prove that this family identifies a well
-defined global distribution $\mathcal{R}^{\gamma} F \in \mathcal{D}' (M)$.To this avail, given $U_{\ell} \in \{U_i\}_{i \in I}$, we have to establish the overlapping condition, \textit{i.e.},  $\mathcal{R}^{\gamma}_j F
  =\mathcal{R}^{\gamma}_{\ell} F$ on $U_j \cap U_{\ell}$ or, equivalently,
  $f_j = \exp_j^{\ast} \circ \exp_{\ell \ast} f_{\ell}$ on $U_j \cap
  U_{\ell}$, where we used that $\exp_j^{\ast} = \exp_{j \ast}^{- 1}$. By Theorem {\cite[Thm. 23]{RS21}}, this will entail the existence of a global
  distribution $\mathcal{R}^{\gamma} F \in \mathcal{D}' (M)$.
  
  The overlapping condition follows by uniqueness of the reconstruction on an
  open set when $\gamma > 0$, as per Theorem~\ref{Thm: Reconstruction
  Theorem}. Given a compact set $K \in U_j \cap U_{\ell}$, $\varphi \in
  \mathcal{B}^{m_K}$, $p \in K$ and $\bar{\lambda} \in (0, \tmop{dist} (K,
  \partial (U_j \cap U_{\ell})))$, it holds that
  \[ \begin{array}{lll}
       | (f_j - \exp_j^{\ast} \circ \exp_{\ell \ast} f_{\ell})
       (\varphi^{\lambda}_{\exp_{\ell}^{- 1} (p)}) | & \leqslant & | (f_j -
       \exp_j^{\ast} F_p) (\varphi^{\lambda}_{\exp_{\ell}^{- 1} (p)}) | +
       \lvert (\exp_j^{\ast} F_p - \exp_j^{\ast} \circ \exp_{\ell \ast}
       f_{\ell}) (\varphi^{\lambda}_{\exp_{\ell}^{- 1} (p)}) \rvert\\
       & \backassign & A + B.
     \end{array} \]
  On account of Equation~\eqref{reconstructionboundonmanifoldsproof}, it holds that $A
  \lesssim \lambda^{\gamma}$. In addition, we observe that
  \[ B = \lvert
     (F_p - \exp_{\ell \ast} f_{\ell})  (\exp_{j \ast}
     \varphi^{\lambda}_{\exp_{\ell}^{- 1} (p)}) \\= \lvert (\exp_{\ell}^{\ast}
     F_p - f_{\ell})  (\exp_{\ell}^{\ast} \circ \exp_{j \ast}
     \varphi^{\lambda}_{\exp_{\ell}^{- 1} (p)}) \lesssim \lambda^{\gamma}. \]
Being the map $\exp_{\ell}^{\ast} \circ \exp_{j \ast}$ bounded
  on $C^{m_K} (K)$, $\exp_{\ell}^{\ast} \circ \exp_{j \ast}
  \varphi_{\exp_{\ell}^{- 1} (p)} \in \tmop{cst} \mathcal{B}^{m_K}$, for a
  suitable constant $\tmop{cst} > 0$. Therefore, by uniqueness of the reconstruction, $f_j = \exp_j^{\ast} \circ
  \exp_{\ell \ast} f_{\ell}$ on $U_j \cap U_{\ell}$ and thus the existence of
  a global distribution $\mathcal{R}^{\gamma} F \in \mathcal{D}' (M)$ is
  proven.
  
  It remains to prove that this global distribution actually abides by Equation~\eqref{reconstructionboundonmanifolds}. Note that Equation \eqref{reconstructionboundonmanifoldsproof} entails that the germ of
  distributions $F -\mathcal{R}^{\gamma} F$ satisfies an homogeneity condition
  as per Proposition \ref{rinaldiimpliesmine}, item $2$, with
  proportionality constant given by $\| F
  \|_{\mathcal{G}^{\tmmathbf{\alpha}_{\gamma}}_{\tmmathbf{r}, \tmmathbf{R}}
  (M), \tilde{K}}$. Therefore, Proposition \ref{rinaldiimpliesmine}
  entails that $F -\mathcal{R}^{\gamma} F$ is a $\gamma$-homogeneous germ
  as per Definition \ref{Def: Homogeneity on M}, which is tantamount to saying
  that Equation \eqref{reconstructionboundonmanifolds} holds true.
  
  Linearity of the map $F \mapsto \mathcal{R}^{\gamma} F$ descends from that of
  $\exp_{j \ast}^{- 1} F \mapsto f_j, \hspace{0.27em} \hspace{0.27em} \forall
  j \in I$. The part of the statement concerning the regularity of the reconstruction follows by invoking locally Corollary~\ref{propertiesreconstruction}.

\end{proof}

\section{Schauder Estimates for germs on Riemannian
Manifolds}\label{sectionSchaudermanifold}

Goal of this section is to prove the analogous of Theorem~\ref{Main Result 1}
in the Riemannian manifold setting. This is achieved by merging the local
formulation of Section~\ref{Sec: Schauder Estimates Open Subsets} and the
construction detailed in Section~\ref{Sec:Germs-of-distributions-manifold}.

Similarly to what we did in the local setting, for the sake of readability, we
shall split our main result in two parts, dubbed \emph{Main Theorem 1} and \emph{2},
respectively. As usual, in the following $M$ denotes a Riemannian manifold, the metric $g$ left implicit for notational ease.

\begin{theorem}[Main Theorem 1]
  \label{maintheoremonmanifolds}Let $\{\Omega_n\}_{n \in \mathbb{N}}$ be an
  exhaustion by compact sets of $M$ as per Definition \ref{Def: Compact Exhaustion} and let $\gamma, \beta > 0$. Let $F$ be an
  $(\tmmathbf{\alpha}, \gamma)$-coherent germ of distributions of order
  $\tmmathbf{r}$ and range $\tmmathbf{R}$, as per Definition \ref{Def:
  Coherence on M} and $\Kappa$ a $\beta$-regularising kernel of order $(m, r)$
  and range $\rho_n$ subordinated to $\{\Omega_n\}_{n \in \mathbb{N}}$ as per Definition
  \ref{regularisingkernelonmanifold}. Assume, moreover, that, $\rho_n$ is small enough, see Equation~\eqref{Eq:rhoN}, and that, for
  every compact set $K \subset M$,
  \begin{equation}
    \label{MaintheoremconditionseqManifold} \alpha_K + \beta > 0,
    \hspace{0.27em} \hspace{0.27em} \hspace{0.27em} \gamma + \beta \nin
    \mathbb{N}_0, \hspace{0.27em} \hspace{0.27em} \hspace{0.27em} m > \gamma +
    \beta, \hspace{0.27em} \hspace{0.27em} \hspace{0.27em} r > - \alpha_K .
  \end{equation}
  Then, there exists a germ of distributions, denoted $\mathcal{K}^{\gamma,
  \beta} F$, such that
  \begin{equation}
    \label{maintheoremequation1Manifold} \mathcal{K}^{\gamma, \beta} F -
    \mathsf{K} (\mathcal{R}^{\gamma} F) \in \mathcal{G}^{\gamma + \beta,
    \tmmathbf{\alpha}'_{\gamma + \beta}}_{\bar{\tmmathbf{r}},
    \bar{\tmmathbf{R}}} (M),
  \end{equation}
  \textit{i.e.}, $\mathcal{K}^{\gamma, \beta} F - \mathsf{K} (\mathcal{R}^{\gamma} F)$
  is $(\gamma + \beta)$-homogeneous and $(\tmmathbf{\alpha}', \gamma +
  \beta)$-coherent on $M$, of order $\tmmathbf{\bar{r}}$, with $\bar{r}_K
  \assign \lfloor - \alpha_K \rfloor$, $\alpha_K' \assign \alpha_K - d$. Here $\bar{R}_K$ is a suitable parameter codifying the range of the new
  germ. Furthermore, the following estimate holds
  \begin{equation}
    \| \mathcal{K}F - \mathsf{K} (\mathcal{R}^{\gamma} F)
    \|_{\mathcal{G}^{\gamma + \beta, \tmmathbf{\alpha}'_{\gamma +
    \beta}}_{\bar{\tmmathbf{r}}, \bar{\tmmathbf{R}}} (M), K} \lesssim_n \| F
    \|_{\mathcal{G}^{\tmmathbf{\alpha }_{\gamma}}_{\tmmathbf{r },
    \tmmathbf{R}} (M), \overline{\Omega_{\tilde{n}}}} ..
    \label{maintheoremequation2Manifold}
  \end{equation}
\end{theorem}

\begin{proof}
  For the sake of readability, we divide the proof in several parts, the first being denoted by $0$ since it serves the purpose of fixing the setting.
  
  \tmtextbf{Part 0:} We consider
  an atlas $\mathcal{A} \assign (U_i, \phi_i)$, $\phi_i \assign \exp_i^{- 1}$
  with $U_i \subset M$ a geodesically convex set. For the sake of clarity, when writing
  $\exp_i$ we mean $\exp_{p_i}$ for a suitable point $p_i \in U_i$ playing the
  role of center of the convex set $U_i$. Since $\Omega_n$ is a relatively
  compact set, there exists $\bar{n} \in \mathbb{N}$ for which $\Omega_n
  \subset \overline{\Omega_n} \subset \bigcup_{i = 1}^{\bar{n}} U_i$.
  
  In addition, without loss of generality, we require that $\forall i \in \{ 1,
  \ldots, \bar{n} \}$, $\tmop{diam} (U_i) = \max_{p, q \in U_i} d_g (p, q)
  <\mathcal{R}_C (\Omega_n) / 4$. We now introduce $\tilde{n}$ as the smallest
  integer such that $(\Omega_n)_{\mathcal{R}_C (\Omega_n)} \subset
  \Omega_{\tilde{n}}$, where $(\Omega_n)_{\mathcal{R}_C (\Omega_n)}$
  stands for the $\mathcal{R}_C (\Omega_n)-$enlargement of $\Omega_n$. By
  construction, it holds that $\bigcup_{i = 1}^{\bar{n}} U_i \subset
  \Omega_{\tilde{n}}$. We define
  \begin{equation}
    \label{RN} \mathfrak{R}_n \assign \min \left\{ \frac{\mathcal{R}_C
    (\Omega_n)}{4}, \frac{\pi}{8 \sqrt{\lvert \mathcal{S} \rvert}},
    \frac{R_{\overline{\Omega_{\tilde{n}}}}}{2} \right\},
  \end{equation}
  where $\lvert \mathcal{S} \rvert$ is the supremum norm of the sectional
  curvature tensor evaluated on $\Omega_n$. The norm
  $|{\mathcal{S}}|$ is finite since the curvature
  tensor is continuous with respect to the base point and $\Omega_n$ is a
  relatively compact set, {\it cf.} Proposition \ref{topologicalequivalencenormalcoordinates}. Finally,
  $R_{\overline{\Omega_{\tilde{n}}}}$ denotes the range of germ $F$ associated
  with the compact set $\overline{\Omega_{\tilde{n}}}$. Concerning the range
  $\rho_n$, we assume that
  \begin{equation}
    \rho_n \leqslant \min \{ \frac{\mathfrak{R}_n}{4},
    \frac{\mathfrak{C}_{\overline{\Omega_{\tilde{n}}}} \mathfrak{R}_n}{4} \},
    \label{Eq:rhoN}
  \end{equation}
  where $\mathfrak{C}_{\overline{\Omega_{\tilde{n}}}}$ is the constant
  associated with the compact set $\overline{\Omega_{\tilde{n}}}$ as per
  Proposition \ref{topologicalequivalencenormalcoordinates}.
  
  \
  
  \tmtextbf{Part 1:} Let now $\Omega_n$, with $n \in \mathbb{N}$, be fixed
  and let $\{ U_i \}_{i = 1}^{\bar{n}}$ be its covering as above. Let now $p
  \in \Omega_n$ and let $K \subset B (p, \mathfrak{R}_n)$ be a compact set.
  Then, on account of Definition \ref{regularisingkernelonmanifold}, in
  particular Equation~\eqref{kernelbound1onmanifold}, associated with $p$ and
  $K$ we have the constant $c_K^{\exp_p^{- 1}}$ for the kernel $\Kappa$.
  
  Our goal now is to prove that it is possible to choose such constants
  $c_K^{\exp_p^{- 1}}$in a way which is uniform with respect to $p \in
  \Omega_n$ and $K \subset B (p, \mathfrak{R}_n)$, thus obtaining a constant
  depending solely on $n$.
  
  For any $p \in \Omega_n$ there exists $j \in \{ 1 \ldots, \bar{n} \}$ such
  that $p \in U_j$. Proposition \ref{Prop:kernel-and-exp}, entails that, for
  any fixed $p_j \in U_j$,
  \begin{equation}
    c^{\exp^{- 1}_p}_K = c^{\exp^{- 1}_{p_j}}_K  \| \exp_p^{- 1} \circ
    \exp_{p_j} \|_{C^{m + r} (\bar{B} (0, \mathcal{R}_C (\Omega_n) / 4))} .
    \label{Eq:factorization-constant}
  \end{equation}
  It now suffices to prove that both factors appearing on the right hand side of
  Equation~\eqref{Eq:factorization-constant} can be bounded uniformly with respect to
  $p \in \Omega_n$ and $K \subset B (p, \mathfrak{R}_n)$ yielding a constant
  depending only on $n$.
  
  To handle the contribution coming from $\| \exp_p^{- 1} \circ \exp_{p_j}
  \|_{C^{m + r} (\bar{B} (0, \mathcal{R}_C (\Omega_n) / 4))}$, we use Lemma
  \ref{Lem:smoothness of diffeomorphism} to obtain smoothness with respect to $p$. Setting
  \[ \tmop{cst}_n \assign \max_{i = 1 \ldots, \bar{n}} \sup_{p \in
     \overline{U_i}} \| \exp_p^{- 1} \circ \exp_{p_i} \|_{C^{m + r} (\bar{B}
     (0, \mathcal{R}_C (\Omega_n) / 4))} < \infty, \]
  we get, for any $j \in \{ 1, \ldots, \bar{n} \}$ and for any $p \in
  \Omega_n$
  \[ \| \exp_p^{- 1} \circ \exp_{p_j} \|_{C^{m + r} (\bar{B} (0, \mathcal{R}_C
     (\Omega_n) / 4))} \leqslant \tmop{cst}_n . \]
  Focusing on the contribution coming from $c^{\exp^{- 1}_{p_j}}_K$, we
  observe that by construction $B (p, \mathfrak{R}_n) \subset B (p_j,
  \mathcal{R}_C (\Omega_n) / 2)$, and thus for any compact $K \subset B (p,
  \mathfrak{R}_n)$, we have $K \subset \overline{B (p_j, \mathcal{R}_C
  (\Omega_n) / 2)}$. We set now
  \[ \tmop{cst}_n' \assign \max_{i = 1, \ldots, \bar{n}} c_{\overline{B (p_i,
     \mathcal{R}_C (\Omega_n) / 2)}}^{\exp^{- 1}_{p_i}} . \]

  By monotonicity with respect to inclusions of compact sets of the constants
  $c^{\exp^{- 1}_{p_j}}_K$, \tmtextit{cf.} Equation~\eqref{kernelbound1onmanifold}, we have, for any $j \in \{ 1,
  \ldots, \bar{n} \}$ and for any $p \in \Omega_n$
  \[ c^{\exp^{- 1}_{p_j}}_K \leqslant \tmop{cst}_n' . \]
  As a consequence, for every point $p \in \Omega_n$ and for every compact set
  $K \subset B (p, \mathcal{R}_C (\Omega_n) / 4)$, we have
  \[ c^{\exp^{- 1}_p}_K \leqslant \tmop{cst}_n' \tmop{cst}_n \backassign c_n .
  \]
  Thus we proved that it is possible to write the bound as per
  Equation~\eqref{kernelbound1onmanifold} with constants only depending on
  $n$.
  
  \
  
  \tmtextbf{Part 2:} Let $p \in M$ and let $n$ be the smallest integer such
  that $p \in \Omega_n$. We want to localize $F_p \in \mathcal{D}' (M)$, so
  that it is a distribution supported in $B (p, 2\mathfrak{R}_n)$. To this
  end, let $\psi_p \in C^{\infty}_0 (M)$ be such that $\tmop{supp} (\psi_p)
  \subset B (p, 2\mathfrak{R}_n)$ and $\psi_p = 1$ on $B (p,
  \mathfrak{R}_n)$. Hence $\psi_p F$ has the desired localization property.
  
  On account of Proposition \ref{Prop:kernel-and-exp} the kernel
  $\exp_p^{\ast} \Kappa$ is a $\beta$-regularising kernel on the open set $B
  (0, \mathfrak{R}_n) \subset \mathbb{R}^d$ with range
  $\frac{\rho_n}{\mathfrak{C}_{\overline{\Omega_n}}}$. We now denote by
  $\bar{\mathcal{K}}_p$ the map between germs on open sets introduced in
  Equation (\ref{maintheoremequationeq}) associated with the kernel
  $\exp_p^{\ast} \Kappa$.
  
  By Theorem~\ref{mineimplyrs}, item $1.$, the germ $\exp_{p \ast}^{- 1}
  (\psi_p F)$ is $(\tmmathbf{\alpha}, \gamma)$-coherent of order $\tmmathbf{r}$ on $B (0, \mathfrak{R}_n) \subset
  \mathbb{R}^d$. In particular, we observe that as $\frac{\mathfrak{R}_n}{2} +
  \frac{2 \rho_n}{\mathfrak{C}_{\overline{\Omega_{\tilde{n}}}}} \leqslant
  \mathfrak{R}_n$, the germ $\exp_{p \ast}^{- 1} (\psi_p F)$ is
  $(\tmmathbf{\alpha}, \gamma)$-coherent of order
  $\tmmathbf{r}$ on $B (0, \frac{\mathfrak{R}_n}{2} + \frac{2
  \rho_n}{\mathfrak{C}_{\overline{\Omega_{\tilde{n}}}}}) \subset
  \mathbb{R}^d$.
  
  We are now in a position to define the map $\tilde{\mathcal{K}}_p$ acting on
  germs $(F_q)_{q \in M}$ on the manifold $M$ as
  \[ \tilde{\mathcal{K}}_p : F \rightarrow \tilde{\mathcal{K}}_p F, \qquad
     \tilde{\mathcal{K}}_p : F_q \mapsto (\tilde{\mathcal{K}}_p F)_q \assign
     \exp_{p \ast} \circ \bar{\mathcal{K}}_p \circ \exp_{p \ast}^{- 1} 
     (\psi_p (q) F_q) . \]
  On account of the above discussion, we have that $(\tilde{\mathcal{K}}_p
  F)_{q \in B (p, \frac{\mathfrak{R}_n}{2})}$ is a germ of distributions.
  
  \
  
  \tmtextbf{Part 3:} Since $\psi_p (q) = 1$ for any $q \in B (p,
  \mathfrak{R}_n)$, it follows that $\mathcal{R}^{\gamma} \exp_p^{\ast} 
  (\psi_p F) =\mathcal{R}^{\gamma} \exp_p^{\ast} (F)$ on the open set $B (0,
  \mathfrak{R}_n)$, and by Theorem \ref{Main Result 1} it follows that
  $\exp_p^{\ast} (\tilde{\mathcal{K}}_p F_q - \mathsf{K} \mathcal{R}^{\gamma}
  F)$ is a $(\gamma + \beta)-$homogeneous germ on $B (0, \mathfrak{R}_n / 2)$.
  Equation \eqref{maintheoremequation2} entails that for any $K \subset B (0,
  \frac{\mathfrak{R}_n}{2})$ compact
  \begin{gather}
    \| \exp_p^{\ast} (\tilde{\mathcal{K}}_p F_q - \mathsf{K}
    \mathcal{R}^{\gamma} F) \|_{\mathsf{G}^{\gamma +
    \beta}_{\bar{\tmmathbf{r}}} (B (0, \frac{\mathfrak{R}_n}{2})), K}
    \leqslant \| \mathcal{K}^{\gamma, \beta} F - \mathsf{K}
    (\mathcal{R}^{\gamma} F) \|_{\mathcal{G}^{\gamma + \beta,
    \tilde{\tmmathbf{\alpha}}_{\gamma + \beta}}_{\bar{\tmmathbf{r}}} (B (0,
    \frac{\mathfrak{R}_n}{2})), K} \notag\\
    \leqslant \tmop{cst} \| \exp_p^{\ast} F
    \|_{\mathcal{G}^{\tmmathbf{\alpha}_{\gamma}}_{\tmmathbf{r}} (B (0,
    \frac{\mathfrak{R}_n}{2} + \frac{2
    \rho_n}{\mathfrak{C}_{\overline{\Omega_{\tilde{n}}}}})), \tilde{K}'},
    \label{maintheoremequation2Manifold2}
  \end{gather}
  where $\tmop{cst}$ is a positive constant possibly dependent on $c^{exp_p^{-
  1}}_K, \rho_n, \gamma + \beta, \alpha_{\tilde{K}}, \bar{r}_{\tilde{K}},
  \bar{\lambda}, d$ and where $\tilde{K}' \subset B (0,
  \frac{\mathfrak{R}_n}{2} + \frac{2
  \rho_n}{\mathfrak{C}_{\overline{\Omega_{\tilde{n}}}}})$ is a compact set as
  per Theorem \ref{Main Result 1}.
  
  In Part $1$ we showed that, for every $p \in \Omega_n$ and $K$
  compact, the constant $c^{exp_p^{- 1}}_K$ relative to the kernel
  $\exp_p^{\ast} \mathsf{K}$ depends only on $n$. Additionally, Proposition
  \ref{rinaldiimpliesmine} -- in particular, its proof -- entails that the germ $\exp_p^{\ast} F$ is $(\gamma
  + \beta)-$homogeneous of order $\tmmathbf{r}$, on $B (0, \mathfrak{R}_n /
  2)$ for a suitable range $\overline{\tmop{cst}}_K$. Thus, invoking Proposition
  \ref{mineimplyrs}, we have that
  \begin{equation}
    \label{maintheoremequation3Manifold} \| \exp_p^{\ast} F
    \|_{\mathcal{G}^{\tmmathbf{\alpha}_{\gamma}}_{\tmmathbf{r}} (B (0,
    \frac{\mathfrak{R}_n}{2} + \frac{2
    \rho_n}{\mathfrak{C}_{\overline{\Omega_{\tilde{n}}}}})), \tilde{K}'}
    {\lesssim_{n, K}}  \| F \|_{\mathcal{G}^{\tmmathbf{\alpha }}_{\tmmathbf{r
    }, \tmmathbf{R}} (M), \exp_p (\tilde{K}')} \leqslant \| F
    \|_{\mathcal{G}^{\tmmathbf{{\alpha_{\gamma}} }}_{\tmmathbf{r },
    \tmmathbf{R}} (M), \overline{\Omega_{\tilde{n}}}}.
  \end{equation}
  Putting together Equations \eqref{maintheoremequation2Manifold2} and
  \eqref{maintheoremequation3Manifold}, we obtain
  \begin{equation}
    \label{maintheoremequation4Manifold} \| \exp_p^{\ast}
    (\tilde{\mathcal{K}}_p F_q - \mathsf{K} \mathcal{R}^{\gamma} F)
    \|_{\mathsf{G}^{\gamma + \beta}_{\bar{\tmmathbf{r}}} (B (0,
    \frac{\mathfrak{R}_n}{2})), K} {\lesssim_{n, K}}  \| F
    \|_{\mathcal{G}^{\tmmathbf{\alpha }}_{\tmmathbf{r }, \tmmathbf{R}} (M),
    \exp_p (\tilde{K}')} \leqslant \| F \|_{\mathcal{G}^{\tmmathbf{\alpha
    }_{\gamma}}_{\tmmathbf{r }, \tmmathbf{R}} (M),
    \overline{\Omega_{\tilde{n}}}},
  \end{equation}
  where the constant in the bound depends only on $n$, as we can exploit
  the monotonicity of the constants $\alpha_K$ with respect to compact sets to bound the parameters
  $\alpha_K$ and $r_K$ by those related to the compact set $\bar{\Omega}_n$.
  
  \
  
  \tmtextbf{Part 4:} Define the map
  \begin{equation}
    \label{shaudermaponmanifolds} \mathcal{K}: F_p \mapsto \mathcal{K}F_p
    \assign \tilde{\mathcal{K}}_p F_p .
  \end{equation}
  By Part 3 of this proof together with Proposition \ref{rinaldiimpliesmine},
  for any $p$, $\mathcal{K}F_p - \mathsf{K} (\mathcal{R}^{\gamma} F)$
  satisfies a $(\gamma + \beta)-$homogeneity bound of order
  $\bar{\tmmathbf{r}} = \lfloor -\tmmathbf{\alpha} \rfloor$ on $M$. Thanks to
  Part 1, the constant of the homogeneity estimate can be made dependent only
  on $n$, the smallest integer such that $p \in
  \Omega_n$. Thus, it does not depend pointwisely on $p$. Therefore, for a
  fixed compact set $K$, the bound is uniform over all points therein.
  Hence $\mathcal{K}F_p - \mathsf{K} (\mathcal{R}^{\gamma} F)$ is a $(\gamma +
  \beta)-$homogeneous germ as per Definition \ref{Def: Homogeneity on M}. By
  Proposition \ref{rinaldiimpliesmine} and Equation
  \eqref{maintheoremequation4Manifold}, we have that, for every compact set $K
  \subset B (p, \mathfrak{R}_N / 2) \subset \Omega_{\tilde{n}}$,
  \begin{equation}
    \label{maintheorem5manifold} \| \mathcal{K}F - \mathsf{K}
    (\mathcal{R}^{\gamma} F) \|_{\mathsf{G}^{\gamma +
    \beta}_{\bar{\tmmathbf{r}}} (M), K} \lesssim_n \| F
    \|_{\mathcal{G}^{\tmmathbf{\alpha }_{\gamma}}_{\tmmathbf{r },
    \tmmathbf{R}} (M), \overline{\Omega_{\tilde{n}}}} .
  \end{equation}
  By applying Proposition \ref{homocoherenceonmanifolds}, it follows that
  $\mathcal{K}F$ is a $(\tmmathbf{\alpha}', \gamma + \beta)-$coherent germ of
  order $\bar{\tmmathbf{r}} = \lfloor - \alpha_K \rfloor$, range
  $\overline{\tmmathbf{R}}$ where for each compact set $K \subset M$,
  $\bar{R}_K$ is the constant $\widetilde{\tmop{cst}}_K$ found in the proof of
  Lemma \ref{recenteringonmanifolds}, while $\alpha_K' \assign \alpha_K - d$.
  In addition, Proposition \ref{homocoherenceonmanifolds} entails that
  \begin{equation}
    \label{mattomatteo} \lvert \lvert \mathcal{K}F - \mathsf{K}
    \mathcal{R}^{\gamma} F \rvert
    \rvert_{\mathcal{G}^{\tmmathbf{\alpha'_{\gamma +
    \beta}}}_{\tmmathbf{\overline{r }}, \tmmathbf{\bar{R}}} (M), K} \lesssim_n
    \| F \|_{\mathcal{G}^{\tmmathbf{\alpha }_{\gamma}}_{\tmmathbf{r },
    \tmmathbf{R}} (M), K'} \leqslant \| F \|_{\mathcal{G}^{\tmmathbf{\alpha
    }_{\gamma}}_{\tmmathbf{r }, \tmmathbf{R}} (M),
    \overline{\Omega_{\tilde{n}}}} .
  \end{equation}
  As a consequence, overall we get
  \[ \| \mathcal{K}F - \mathsf{K} (\mathcal{R}^{\gamma} F)
     \|_{\mathcal{G}^{\gamma + \beta, \tmmathbf{\alpha}'_{\gamma +
     \beta}}_{\bar{\tmmathbf{r}}, \bar{\tmmathbf{R}}} (M), K} \assign \|
     \mathcal{K}F - \mathsf{K} (\mathcal{R}^{\gamma} F) \|_{\mathsf{G}^{\gamma
     + \beta}_{\bar{\tmmathbf{r}}, \overline{\tmmathbf{R}}} (M), K} + \lvert
     \lvert \mathcal{K}F - \mathsf{K} \mathcal{R}^{\gamma} F \rvert
     \rvert_{\mathcal{G}^{\tmmathbf{\alpha'_{\gamma +
     \beta}}}_{\overline{\tmmathbf{r }}, \overline{\tmmathbf{R}}} (M), K}
     \lesssim \| F \|_{\mathcal{G}^{\tmmathbf{\alpha }_{\gamma}}_{\tmmathbf{r
     }, \tmmathbf{R}} (M), \overline{\Omega_{\tilde{n}}}} \]
  
\end{proof}

\begin{theorem}[Main Theorem 2]
  \label{maintheoremonmanifolds2}Under the same assumptions of Theorem
  \ref{maintheoremonmanifolds}, if $\mathsf{K} (\mathcal{R}^{\gamma} F)$ is a
  well-defined distribution in $\mathcal{D}' (M)$, then $\mathcal{K}^{\gamma,
  \beta} F$ is a well defined germ on $M$ and it holds that:
  \begin{itemize}
    \item[(i)] $\mathcal{K}^{\gamma, \beta} F \in \mathcal{G}^{\gamma + \beta,
    \tmmathbf{\alpha'_{\gamma + \beta}}}_{\tmmathbf{\bar{r}}, \tmmathbf{R'}}
    (M)$ and
    \begin{equation}
      \label{maintheoremequation3Manifold4} \mathcal{R}^{\gamma + \beta}
      (\mathcal{K}^{\gamma, \beta} F) = \mathsf{K} (\mathcal{R}^{\gamma} F),
    \end{equation}
    on $M$, {\it i.e.}, $\mathsf{K} (\mathcal{R}^{\gamma} F)$ is the $(\gamma +
    \beta)-$reconstruction of the germ of distributions $\mathcal{K}^{\gamma,
    \beta} F$.
    
    \item[(ii)] If $F$ is $\bar{\tmmathbf{\alpha}}$-homogeneous, with
    \[ \bar{\alpha}_K \leqslant \gamma, \hspace{0.27em} \hspace{0.27em}
       \hspace{0.27em} \bar{\alpha}_K + \beta \neq 0, \]
    then $\mathcal{K}^{\gamma, \beta} F$ is $((\bar{\tmmathbf{\alpha}} +
    \beta) \wedge 0)-$homogeneous of a suitable order $\bar{\tmmathbf{r}}$,
    where $(\bar{\alpha}_K + \beta) \wedge 0 \assign \min \{ \bar{\alpha}_K +
    \beta, 0\}$.
    
    \item[(iii)] The map $F \mapsto \mathcal{K}^{\gamma, \beta} F$ is linear and
    continuous, both as
    \[ \mathcal{K}^{\gamma, \beta} :
       \mathcal{G}^{\tmmathbf{\alpha}_{\gamma}}_{\tmmathbf{r}, \tmmathbf{R}}
       (M) \to \mathcal{G}^{\tmmathbf{\alpha}'_{\gamma +
       \beta}}_{\bar{\tmmathbf{r}}, \overline{\tmmathbf{R}}} (M), \qquad
       \tmop{and} \qquad \mathcal{K}^{\gamma, \beta} :
       \mathcal{G}^{\tmmathbf{\bar{\alpha}},
       \tmmathbf{\alpha}_{\gamma}}_{\tmmathbf{r}} (M) \to
       \mathcal{G}^{(\tmmathbf{\bar{\alpha}} + \beta \wedge 0),
       \tmmathbf{\alpha}'_{\gamma + \beta}}_{\bar{\tmmathbf{r}},
       \overline{\tmmathbf{R}}} (M), \]
    abiding by the following continuity estimates: For a suitable compact $K'$
    such that $K \subset K'$ and for suitable parameters $\tmmathbf{\alpha}' 
    = \{ \alpha_K' \}_K$,
    \[ \lvert \lvert \mathcal{K}^{\gamma, \beta} F \rvert
       \rvert_{\mathcal{G}^{\tmmathbf{\alpha}'_{\gamma +
       \beta}}_{\bar{\tmmathbf{r}}, \overline{\tmmathbf{R}}} (M), K}
       \lesssim_n \hspace{0.17em} \lvert \lvert F \rvert
       \rvert_{\mathcal{G}^{\tmmathbf{\alpha}_{\gamma}}_{\tmmathbf{r},
       \tmmathbf{R}} (M), K'}, \qquad \lvert \lvert \mathcal{K}^{\gamma,
       \beta} F \rvert \rvert_{\mathcal{G}^{(\tmmathbf{\bar{\alpha}} + \beta
       \wedge 0), \tmmathbf{\alpha}'_{\gamma + \beta}}_{\bar{\tmmathbf{r}},
       \overline{\tmmathbf{R}}} (M), K} \lesssim_n \lvert \lvert F \rvert
       \rvert_{\mathcal{G}^{\tmmathbf{\bar{\alpha}},
       \tmmathbf{\alpha}_{\gamma}}_{\tmmathbf{r}} (M), K'}, \]
    where $n$ is the smallest integer such that $K \subset \Omega_n$. As a
    consequence the following diagram commutes
    \[ \begin{array}{ccc}
         \mathcal{G}^{\tmmathbf{\alpha}_{\gamma}}_{\tmmathbf{r}, \tmmathbf{R}}
         (M) & \xrightarrow{\mathcal{K}^{\gamma, \beta}} &
         \mathcal{G}^{\tmmathbf{\alpha}'_{\gamma +
         \beta}}_{\bar{\tmmathbf{r}}, \overline{\tmmathbf{R}}} (M)\\
         \longdownarrow \mathcal{R}^{\gamma} &  & \longdownarrow
         \mathcal{R}^{\gamma + \beta}\\
         \mathcal{D}' (M) & \xrightarrow{\Kappa} & \mathcal{D}' (M)
       \end{array} \]
  \end{itemize}
\end{theorem}

\begin{proof}
  The proof of this theorem is similar to the analogous results in the local
  framework of open sets, namely Propositions~\ref{Main Result 2}
  and~\ref{Main Result 3}.
  
  In order to prove the item $(i)$, we recall that the coherence bound in
  Equation~\eqref{coherenceonmanifold1} does not separate constant germs,
  \tmtextit{i.e.}, coherence is not affected by the addition of a
  distribution/constant germ. As a consequence, since by
  Theorem~\ref{maintheoremonmanifolds} we have that $\mathcal{K}F - \mathsf{K}
  (\mathcal{R}^{\gamma} F)$ is $(\tmmathbf{\alpha}', \gamma + \beta)-$coherent
  on $M$ and we are assuming that $\mathsf{K} (\mathcal{R}^{\gamma} F) \in
  \mathcal{D}' (M)$, we conclude that $\mathcal{K}F$ is a
  $(\tmmathbf{\alpha}', \gamma + \beta)-$coherent germ of order
  $\bar{\tmmathbf{r}} = \lfloor - \tmmathbf{\alpha} \rfloor$ and range
  $\overline{\tmmathbf{R}}$. In addition, it follows that
  \begin{equation}
    \lvert \lvert \mathcal{K}F \rvert
    \rvert_{\mathcal{G}^{\tmmathbf{\alpha'_{\gamma +
    \beta}}}_{\tmmathbf{\bar{r} }, \overline{\tmmathbf{R}}} (M), K} \lesssim_n
    \| F \|_{\mathcal{G}^{\tmmathbf{\alpha }_{\gamma}}_{\tmmathbf{r },
    \tmmathbf{R}} (M), \overline{\Omega_{\tilde{n}}}} .
    \label{Eq:coherence-continuity}
  \end{equation}
  Eventually, to conclude the proof of the item $(i)$, we need to show that
  $\mathcal{R}^{\gamma + \beta} \mathcal{K}F = \mathsf{K} \mathcal{R}^{\gamma}
  F$. We observe that, for $\gamma + \beta > 0$, existence and
  uniqueness of a reconstruction of the germ $\mathcal{K}F$ are granted by Theorem~\ref{reconstructiononmanifolds}. Additionally, we have already
  proven that the germ $\mathcal{K}F - \mathsf{K} (\mathcal{R}^{\gamma} F)$ is
  homogeneous of degree $\gamma + \beta$. Recalling now Equation~\eqref{reconstructionboundonmanifolds}, uniqueness of the
  reconstruction entails precisely that
  \begin{equation}
    \mathcal{R}^{\gamma + \beta} \mathcal{K}F = \mathsf{K}
    \mathcal{R}^{\gamma} F. \label{Eq:ref}
  \end{equation}
  We focus on the continuity bounds. We observe that
  the one involving coherence has already been proven in
  Equation~\eqref{Eq:coherence-continuity}. Recalling Definition~\ref{Def:
  Coherence and Homogeneity on M}, to conclude we need to prove a bound for
  the homogeneity norm.
  
  To this end, we show first that $\mathsf{K} \mathcal{R}^{\gamma} F$ is a
  $(\bar{\alpha} + \beta) \wedge 0-$homogeneous germ. Given $p \in \Omega_n$,
  using Equation~\eqref{maintheoremequation4Manifold} together with an
  argument analogous to the one we used to establish Equation~\eqref{Eq:ref}, we get
  \[ \exp_p^{\ast} (\mathsf{K} \mathcal{R}^{\gamma} F) =\mathcal{R}^{\gamma +
     \beta} \exp_p^{\ast} (\mathcal{K}F). \]  
By Equations~$\eqref{reconstructionishomogeneouseq1} $and~\eqref{Eq:
  Homogeneous and coherent norm reconstruction} it follows that, for every
  compact set $K \subset B (0, \mathfrak{R}_n / 2)$, $\exp_p^{\ast} \mathsf{K}
  \mathcal{R}^{\gamma} F$ is a $(\bar{\alpha}_K + \beta) \wedge 0-$homogeneous
  germ on $U \assign B (0, \mathfrak{R}_n / 2 + 2 \rho_n) \subset
  \mathbb{R}^d$, such that
  \begin{equation}
    \begin{array}{lll}
      \| \exp_p^{\ast} (\mathsf{K} \mathcal{R}^{\gamma} F)
      \|_{G^{(\tmmathbf{\bar{\alpha}} + \beta) \wedge 0}_{\tmmathbf{\bar{r}}}
      (U^{2 \rho_n}), K} & = & \| \mathcal{R}^{\gamma + \beta} \exp_p^{\ast}
      \mathcal{K}F \|_{G^{(\tmmathbf{\bar{\alpha}} + \beta) \wedge
      0}_{\tmmathbf{\bar{r}}} (U^{2 \rho_n}), K}\\
      & \lesssim^{\eqref{reconstructionishomogeneouseq1}} & \| \exp_p^{\ast}
      \mathcal{K}F \|_{\mathcal{G}^{(\tmmathbf{\bar{\alpha}} + \beta) \wedge
      0, \tmmathbf{\alpha'_{\gamma + \beta}}}_{\tmmathbf{\bar{r}},
      \tmmathbf{\bar{R}}} (U^{2 \rho_n}), \tilde{K}}\\
     & = & \| \bar{\mathcal{K}}^{\gamma, \beta} \exp_p^{\ast} F
      \|_{\mathcal{G}^{(\tmmathbf{\bar{\alpha}} + \beta) \wedge 0,
      \tmmathbf{\alpha'_{\gamma + \beta}}}_{\tmmathbf{\bar{r}},
      \tmmathbf{\bar{R}}} (U^{2 \rho_n}), \tilde{K}}\\
      & \lesssim^{\eqref{Eq: Homogeneous and coherent norm
      reconstruction}} & \| \exp_p^{\ast} F
      \|_{\mathcal{G}^{\tmmathbf{\bar{\alpha}},
      \tmmathbf{\alpha_{\gamma}}}_{\tmmathbf{\bar{r}}, \tmmathbf{\bar{R}}}
      (U), \tilde{K}}\\
      & \lesssim^{\ref{mineimplyrs}} & \| F
      \|_{\mathcal{G}^{\tmmathbf{\bar{\alpha}},
      \tmmathbf{\alpha_{\gamma}}}_{\tmmathbf{\bar{r}}, \tmmathbf{\bar{R}}}
      (M), \overline{\Omega_{\tilde{n}}}},
    \end{array} \label{Eq:last step}
  \end{equation}
  where in the last line we applied Proposition~\ref{mineimplyrs}. In this way
  we proved that $\exp_p^{\ast} (\mathsf{K} \mathcal{R}^{\gamma} F)$ is an
  $(\tmmathbf{\bar{\alpha}} + \beta) \wedge 0-$homogeneous germ in the sense
  of open sets and this, on account of Proposition~\ref{rinaldiimpliesmine},
  entails that $\mathsf{K} \mathcal{R}^{\gamma} F$ is an
  $(\tmmathbf{\bar{\alpha}} + \beta) \wedge 0-$homogeneous germ on the
  manifold $M$.
  
  This entails that $\mathcal{K}F = (\mathcal{K}F - \mathsf{K}
  \mathcal{R}^{\gamma} F) + \mathsf{K} \mathcal{R}^{\gamma} F$ is obtained by
  summing a $(\tmmathbf{\bar{\alpha}} + \beta)-$homogeneous germ and a
  $(\tmmathbf{\bar{\alpha}} + \beta) \wedge 0-$homogeneous germ. Equations
  \eqref{maintheorem5manifold} and \eqref{Eq:last step} together yield the
  sought result, \tmtextit{i.e.},
  \[ \lvert \lvert \mathcal{K}F \rvert \rvert_{G^{(\tmmathbf{\bar{\alpha}}  +
     \beta) \wedge 0}_{\tmmathbf{\bar{r} }} (M), K} \lesssim_n \lvert \lvert F
     \rvert \rvert_{\mathcal{G}^{\tmmathbf{\bar{\alpha}},
     \tmmathbf{\alpha}_{\gamma}}_{\tmmathbf{r}} (M),
     \overline{\Omega_{\tilde{n}}}} . \]
  
\end{proof}

\begin{acknowledgments*}
The work of BC has been supported by a fellowship of the University of Pavia and in part by a fellowship of the ``Progetto Giovani GNFM 2025" under the project {\em Hadamard states for linearized Yang-Mills theories}, fostered by Gruppo Nazionale di Fisica Matematica -- INdAM in collaboration with Simone Murro (University of Genova). BC, CD and PR acknowledge the support of the INFN Sezione di Pavia and of Gruppo Nazionale di Fisica Matematica, part of INdAM. Part of this work has appeared in the Master thesis of MS submitted on the 21/02/25 as a partial fulfillment of the requirement to obtain a Master's degree in Physics at the University of Pavia.
\end{acknowledgments*}

\appendix\section{Additional Properties of Homogeneous and Coherent
Germs}\label{Sec: App A}

Goal of this appendix is to prove a technical result which connects
homogeneity and coherence of a germ under suitable assumptions. This plays an
important role in the proof of Proposition \ref{Main Result 3}. As a starting point, we
state a transliteration to our setting of {\cite[Prop 6.2]{CZ20}}. The proof is identical to {\cite[Prop 6.2]{CZ20}} barring minor modifications, thus we omit it.  

\begin{theorem}
  \label{Prop: Necessity of Coherence}Let $U \subseteq \mathbb{R}^d$ and $f
  \in \mathcal{D}' (U)$. Consider a germ $(F_x)_{x \in U}$ as per Definition
  \ref{Def: germ}, and $\gamma \in \mathbb{R}$, such that, for every compact
  set $K \subset U$, there exist $r_K \in \mathbb{N}$ and $C_K > 0$, such that
  \begin{equation}
    \label{necessityiofcoherenceq} \lvert (f - F_x) (\varphi^{\lambda}_x)
    \rvert \leqslant C_K \lambda^{\gamma},
  \end{equation}
  uniformly for $x \in K$, $\lambda \in (0, D^U_K)$ and $\varphi \in
  \mathcal{B}^{r_K}$ see Equations \eqref{Eq: Rescaled Test-Function} and
  \eqref{Eq: B^m}. Then, given $\alpha_K \assign \min \{- r_K - d, \gamma\}$,
  it holds
  \begin{equation}
    \label{necessityiofcoherenceqoo} \lvert (F_x - F_y) (\varphi^{\lambda}_y)
    \rvert \leqslant 2 C_K \lambda^{\alpha_K} (\lambda + \lvert x - y
    \rvert)^{\gamma - \alpha_K},
  \end{equation}
  uniformly for $x, y \in K$, with $\lvert x - y \rvert \leqslant
  \frac{D^U_K}{2}$, $\lambda \in (0, \frac{D^U_K}{2})$ and $\varphi \in
  \mathcal{B}^r$.
\end{theorem}

As we have seen in Definitions~\ref{Def: Coherence} and~\ref{Def:
  Homogeneity}, both the notions of coherence and homogeneity of a germ depend on a number
  of relevant parameters, including the scaling parameter $\lambda$. Hence, it seems natural to wonder what should be the regimes in which these parameters may vary. In particular, in Definitions~\ref{Def: Coherence} and~\ref{Def:
  Homogeneity}, we have that $\lambda \in (0, D^U_K)$, where $D^K_U \assign \text{dist} (K, \partial U)$. Proposition \ref{relaxingbound11} entails that, if
  one can establish a coherence bound uniformly over $\lambda \in (0, \eta_K)$ for
  some $\eta_K \in (0, D_K^U)$, then one can extend the bound to the whole
  interval $\lambda \in (0, D_K^U)$, thus justifying the choices made in these Definitions.

\begin{proposition}
  \label{relaxingbound11}Let $U \subseteq \mathbb{R}^d$ and let $F$ be a germ
  of distributions thereon as per Definition \ref{Def: germ}.
  \begin{enumerate}
    \item Assume that $U$ is convex and that, given $\gamma \in \mathbb{R}$,
    for any compact set $K \subset U$ there exist $\alpha_K < \gamma$, $\eta_K
    \in (0, D_K^U)$, $r_K \in \mathbb{N}$ as well as $C_K > 0$ such that
    \begin{equation}
      \label{relaxingbound1eq} \lvert (F_x - F_y) (\varphi^{\lambda}_x) \rvert
      \leqslant C_K \lambda^{\alpha_K} (\lambda + \lvert x - y \rvert)^{\gamma
      - \alpha_K},
    \end{equation}
    uniformly for $x, y \in K$, $\lambda \in (0, \eta_K]$ and $\varphi \in \mathcal{B}^{r_K}$ as per Equation \eqref{Eq:
    B^m}, while $\varphi^{\lambda}_x$ is as per Equation \eqref{Eq: Rescaled
    Test-Function}. Then $F \in \mathcal{G}^{\tmmathbf{\alpha}_{\gamma}}_{\tmmathbf{r}} (U)$, namely it is
    a coherent germ of distributions as per Definition \ref{Def: Coherence}
    and
    \begin{equation}
      \label{relaxingbound1eqnorm} \| F
      \|_{\mathcal{G}^{\tmmathbf{\alpha}_{\gamma}}_{\tmmathbf{r}} (U), K}
      \lesssim C_{\tilde{K}, \bar{\lambda}},
    \end{equation}
    where $C_{\tilde{K}, \bar{\lambda}}$ depends on $\bar{\lambda} \in (0,
    D^U_K)$ and on a suitable compact set $\tilde{K} \supseteq K$.
    
    \item Assume that for any compact set $K \subset U$ there exist $\eta_K
    \in (0, D_K^U)$, $r_K \in \mathbb{N}$, $C_K > 0$ and $\bar{\alpha}_K
    \leqslant 0$, such that
    \begin{equation}
      \label{relaxingbound1homogeneityeq1} \lvert (F_x) (\varphi^{\lambda}_x)
      \rvert \leqslant C_K \lambda^{\bar{\alpha}_K},
    \end{equation}
    uniformly for $x, y \in K$, $\lambda \in (0, \eta_K]$ and $\varphi \in \mathcal{B}^{r_K}$. Then $F \in \mathsf{G}^{\tmmathbf{\bar{\alpha}}}_{\tmmathbf{r}}
    (U)$, namely it is an homogeneous germ of distributions and
    \begin{equation}
      \label{relaxingbound1homogeneityeq1norm} \| F
      \|_{\mathsf{G}^{\tmmathbf{\bar{\alpha}}}_{\tmmathbf{r}} (U), K} \lesssim
      C_{\tilde{K}, \bar{\lambda}},
    \end{equation}
    where $C_{\tilde{K}, \bar{\lambda}}$ depends on $\bar{\lambda} \in (0,
    D^U_K)$ and on a suitable compact set $\tilde{K} \supseteq K$.
  \end{enumerate}
\end{proposition}

\begin{proof}
  $1$. Consider a compact set $K \subset U$ and $\eta_K \in (0, D^U_K)$ as per hypothesis. The goal is to show that Equation \eqref{relaxingbound1eq} holds true uniformly for any $\lambda \in (0, \bar{\lambda}]$ with $\bar{\lambda} \in (0, D^U_K)$, see Definition \ref{Def: Coherence}. Observe that, if $\lambda \in (0, \eta_K]$, which corresponds to fixing $\bar{\lambda} \in (0, \eta_K]$, and if we let $ \lambda \in (0, \bar{\lambda}]$, there is nothing to prove. Equation \eqref{relaxingbound1eq} holds true uniformly in this regime by assumption. Hence, the only relevant case is
  $$\lambda \in (\eta_K, \bar{\lambda}), \, \, \, \, \, \bar{\lambda}
  \in (\eta_K, D^U_K) \, \, \text{such that} \, \,  \bar{\lambda} \leqslant 2\eta_{K}.$$
If we manage to prove the validity of Equation \eqref{relaxingbound1eq} in this scenario, we can conclude that it holds true $\forall \lambda \in (0, \bar{\lambda})$.

To this avail, let $K' \assign K_{\eta_K}$ be the
  $\eta_K-$enlargement of the compact $K$ as per Equation \eqref{Eq: K-enlargement}. Note that, since $K \subset K'$,
  $\eta_{K'} \leqslant \eta_K$. Consider a finite cover of $K'$ made of open balls $\{B (x_i,
  \eta_{K'})\}_{i = 1, \ldots N_K}$, where $x_i \in K_{\bar{\lambda} -
  \eta_{K'}} \subset K'$ and $N_K \assign \lfloor \mathrm{Diam} (K' / \eta_{K'}) + 1
  \rfloor$, see Equation \eqref{Eq: Diameter of K}. Let $\{\psi_i\}_{i = 1,
  \ldots N_K}$ be a partition of unity subordinated to such covering. Then,
  for all $\lambda \in (\eta_{K'}, \bar{\lambda})$, $\forall x \in K'$ and $\forall \varphi \in \mathcal{B}^{r_{K'}}$, we can write
  \begin{equation}
\label{extendingholderbound1coherence} \varphi^{\lambda}_x = \sum_{i =
    1}^{N_K} \varphi_{x_i}^{\eta_{K'}},
  \end{equation}
  where $\varphi_{x_i} \assign (\varphi^{\lambda}_x \psi_i)^{1 / \eta_{K'}}$.
  Given $r_{K'} \in \mathbb{N}$, relative to the compact set $K' \subset U$, we set $$C_{r_{K'} } \assign N_K \max_{i = 1, \ldots, N_K}
  \| \psi_i \|_{C^{r_{K'}} (U)}.$$ It descends that
  \begin{equation}
    \begin{array}{lll}
      \| \varphi_{x_i} \|_{C^{r_{K'}} (U)} = \| (\varphi^{\lambda}_x \psi_i)^{1 /
      \eta_{K'}} \|_{C^{r_{K'}} (U)} \leqslant (\eta_{K'})^{- r_{K'} - d} \|
      \varphi^{\lambda}_x \psi_i \|_{C^{r_{K'} } (U)} & \leqslant & (\eta_{K'})^{- r_{K'}  -
      d} 2^{r_{K'}} \| \psi_i \|_{C^{r_{K'}} (U)} \hspace{0.17em} \| \varphi^{\lambda}_x
      \|_{C^{r_{K'} } (U)}\\
      & \leqslant & 2^{r_{K'} } C_{r_{K'} } \left( \frac{\lambda}{\eta_{K'}} \right)^{- r_{K'}  -
      d} \| \varphi_x \|_{C^{r_{K'} } (U)}\\
      & \leqslant & 2^{r_{K'} } C_{r_{K'} } \| \varphi_x \|_{C^{r_{K'}} (U)},
    \end{array} \label{extendingholderbound1coherenceCr}
  \end{equation}
  where, in the last inequality, we used that $\left( \frac{\lambda}{\eta_{K'}} \right)^{- r_{K'} - d}
  \leqslant 1$ since $\lambda \geqslant \eta_{K'}$. Being $\varphi_x \in \mathcal{B}^{r_{K'}}$, it descends that $\varphi_{x_i} \in 2^{r_{K'}} C_{r_{K'}} \mathcal{B}^{r_{K'}}$. Letting $\alpha_{K'} < \gamma$ with $\gamma \in \mathbb{R}$, we have 
  \begin{equation*}
      \lambda^{\gamma - \alpha_{K'}} \leqslant (\lambda + |x-y|)^{\gamma - \alpha_{K'}}
  \end{equation*}
  and, multiplying both sides by $\lambda^{\alpha_{K'}} > 0$, we obtain 
  \begin{equation*}
      \lambda^{\alpha_{K'}} \lambda^{\gamma - \alpha_{K'}} \leqslant \lambda^{\alpha_{K'}}
    (\lambda + \lvert x - y \rvert)^{\gamma - \alpha_{K'}}. 
  \end{equation*}
  In addition, since $\lambda \in (\eta_{K'}, \bar{\lambda})$,
  \begin{equation*}
      \lambda^\gamma \geqslant \begin{cases}
          \eta_{K'}^{\gamma} \, \, \text{if $\gamma \ge 0$} \\
          \bar{\lambda}^{\gamma} \, \, \text{if $\gamma < 0$}
      \end{cases},
  \end{equation*}
  which entails $1 \lesssim_{\eta_{K'}, \bar{\lambda}} \lambda^{\gamma}$. Combining together these two results we get
  \begin{equation}
  \label{extendingholderbound2coherence} 1 \lesssim_{\eta_{K'}, \bar{\lambda}}
    \lambda^{\alpha_{K'}} \lambda^{\gamma - \alpha_{K'}} \leqslant \lambda^{\alpha_{K'}}
    (\lambda + \lvert x - y \rvert)^{\gamma - \alpha_{K'}},
  \end{equation}
  uniformly for $\lambda \in [\eta_{K'}, \bar{\lambda}]$ and $x, y \in K'$. By Equation \eqref{extendingholderbound1coherence}, we have that, for any $x,y \in K'$, $\lambda \in (\eta_{K'}, \bar{\lambda})$ and $\varphi \in \mathcal{B}^{r_{K'}}$
  \[ \lvert (F_x - F_y) (\varphi^{\lambda}_y) \rvert \leqslant \sum_{i = 1}^{N_K}
     \lvert (F_x - F_y) (\varphi_{x_i}^{\eta_{K'}}) \rvert . \]
     As a consequence, since  $\varphi_{x_i} \in 2^{r_{K'}} C_{r_{K'}} \mathcal{B}^{r_{K'}}$, we can apply Equation \eqref{relaxingbound1eq} to get
  that for any $i \in \{1, ..., N_K \}$
  \begin{equation}
    \begin{array}{lll}
      \lvert (F_x - F_y) (\varphi_{x_i}^{\eta_{K'}}) \rvert \leqslant \lvert
      (F_x - F_{x_i}) (\varphi_{x_i}^{\eta_{K'}}) \rvert + \lvert (F_{x_i} -
      F_y) (\varphi_{x_i}^{\eta_{K'}}) \rvert & \lesssim_{\eta_{K'},
      \bar{\lambda}} & C_{K'} \eta_{K'}^{\alpha} (2 \eta_{K'} + |x - x_i | +
      |x_i - y|)^{\gamma - \alpha_{K'}}\\
      & \lesssim_{\eta_{K'}, \bar{\lambda}} & C_{K'} \lambda^{\alpha} (\lambda +
      \lvert x - y \rvert)^{\gamma - \alpha_{K'}},
    \end{array} \label{extendingholderbound3coherence}
  \end{equation}
uniformly for $\lambda \in (\eta_{K'}, \bar{\lambda})$ and $x, y \in K'$, where
in the first inequality we used Equation \eqref{relaxingbound1eq}, whereas in the last one we exploited Equation \eqref{extendingholderbound2coherence}.

Therefore, we proved the sought result for
  $\bar{\lambda} \leqslant 2 \eta_K$. Note that there are two admissible scenarios: either $2\eta_K \geqslant D^U_K$ or $2\eta_K \leqslant D^U_K$. In the former, there is nothing left to prove and we can conclude by choosing $\tilde{K}$ to coincide with $K'$. The latter case instead requires further comments. Indeed, if $2\eta_K \leqslant D^U_K$ we can extend the
  range of validity of Equation \eqref{relaxingbound1eq} to all values $\bar{\lambda} \in (2\eta_K, D_K^U)$ by
  iterating the above procedure. In this way, we obtain a uniform bound over
  $\lambda \in (0, \bar{\lambda})$ with $\bar{\lambda} \in (0, 2 \eta_K)$. 
The maximum number of steps required in the iteration is given by $\lfloor
  D_K^U / \eta_K + 1 \rfloor$, and the constant in the estimate in Equation
  \eqref{relaxingbound1eqnorm} is still proportional to
  $C_{\tilde{K}}$, where $\tilde{K}$ is a suitable enlargement of $K$, which
  depends on the number of steps and on $\bar{\lambda}$. 
  
  $2$. As in the proof of item $1$, fixed a compact set $K \subset U$, we want to show that the estimate in Equation \eqref{relaxingbound1homogeneityeq1} holds true uniformly for any $\lambda \in (0, \bar{\lambda}]$ with $\bar{\lambda} \in (0, D^U_K)$. To this end, we shall follow a similar strategy. First of all, let us consider $K' \assign
  K_{\eta_K}$ with $\eta_K \in (0, D^U_K)$ given. We stress that $\eta_{K'} \leqslant \eta_K$,
  since $K \subset K'$. As before, fixing $\bar{\lambda} \in (\eta_K, D_K^U)$, such
  that $\bar{\lambda} \leqslant 2 \eta_K$, the only interesting case is the one in which $\lambda \in
  (\eta_K, \bar{\lambda})$. Using the same partition of unity $\{\psi_i\}_{i=1, \ldots, N_K}$ as in the proof of item $1$, it
  follows that, for an assigned $r_{K'} \in \mathbb{N}$, for any$\lambda \in (\eta_{K'},
  \bar{\lambda})$ and for $\varphi \in \mathcal{B}^{r_{K'}}$, we can write 
  \[ \varphi^{\lambda}_x = \sum_{i = 1}^N \varphi_{x_i}^{\eta_{K'}}, \]
  with $\varphi_{x_i} \assign (\varphi_x^{\lambda} \psi_i)^{1/\eta_{K'}} \in 2^{r_{K'}} C_{r_{K'}} \mathcal{B}^{r_{K'}}$. Re-centering the test function $\varphi_{x_i}^{\eta_{K'}}$
  at the point $x \in K$, we can write
  \begin{equation}
    \label{extendingholderboundhomogeneityeq1} \varphi^{\lambda}_x = \sum_{i =
    1}^{N_K} (\xi_i)_x^{\eta_{K'}},
  \end{equation}
  where $\xi_i \assign \varphi^{\lambda_2^i}_{w_i}$, $\lambda_1^i \assign
  \lvert x_i - x \rvert + \eta_{K'}$, while $\lambda_2^i \assign
  \frac{\eta_{K'}}{\lambda_1^i}$ and $w \assign \frac{x_i - x}{\lambda_1^i}$.
  Since for any test function $\varphi \in \mathcal{B}^{r_{K'}}$, any multi-index $k \in
  \mathbb{N}^d_0$ and any $\lambda \in \mathbb{R}$, $\partial^k
  \varphi^{\lambda} = (\lambda)^{- r_{K'} - d} (\partial^k \varphi)^{\lambda}$, we
  have that
  \[ \| \xi_i \|_{C^{r_{K'}}} = (\lambda_2^i)^{- r_{K'} - d} \| \varphi \|_{C^{r_{K'}}} \leqslant
     \eta_{K'}^{- r_{K'} - d}  (\mathrm{Diam} (K) + \eta_{K'})^{r_{K'} + d} =: c_{K, K'}, \]
  where $c_{K, K'} > 0$ and where we bounded $(\lambda_2^i)^{- r_{K'} - d} =
  \eta_{K'}^{- r_{K'} - d} (\lambda_1^i)^{r_{K'} + d} \leqslant \eta_{K'}^{- r_{K'} - d} 
  (\tmop{Diam} (K) + \eta_{K'})^{r_{K'} + d}$. Hence $\xi_i \in \tmop{cst} \mathcal{B}^{r_{K'}}$ and we can apply
  Equation \eqref{relaxingbound1homogeneityeq1} to Equation
  \eqref{extendingholderboundhomogeneityeq1}, thus obtaining
  \begin{equation}
\label{extendingholderbound1homoogeneity} \lvert F_x (\varphi^{\lambda}_x)
    \rvert \leqslant \sum_{i = 1}^{N_K} \lvert F_x ((\xi_i)_x^{\eta_{K'}})
    \rvert \leqslant C_{K', \bar{\lambda}} N_K \tmop{cst}_{K, \eta_{K'}}
    \eta_{K'}^{\bar{\alpha}_{K'}},
  \end{equation}
  uniformly for $\lambda \in (\eta_{K'}, \bar{\lambda}]$, $\varphi \in
  \mathcal{B}^{r_{K'}}$ and $x \in K'$. Hence, working now on the compact $K$, we have proven Equation \eqref{relaxingbound1homogeneityeq1} We can now extend the
  estimate in Equation \eqref{extendingholderbound1homoogeneity} via an iterative argument,
  by means of the same procedure outlined in the proof of item $1$.
\end{proof}

  In a similar fashion, the following proposition shows that to establish the
  coherence condition it suffices to check the validity of the bound close to the diagonal,
  namely for $| x - y |$ small. Indeed, as we shall see, in order to prove the
  coherence bound for any $x, y \in K$ with $K \subset U$ a compact set, it
  suffices to have the same bound for $| x - y | < R_K$ for some $R_K > 0$.
  The only price to pay is a reduction of the value of the parameters
  $\alpha_K$ entering in the coherence condition. Yet, this does not constitute
  a problem from the point of view of both the reconstruction theorem and Schauder
  estimates, see Sections \ref{reconstruction theorem section} and \ref{Sec: Schauder Estimates Open Subsets}.

\begin{proposition}
  \label{relaxingbound2}Let $U \subseteq \mathbb{R}^d$ be a convex, open set
  and let $(F_x)_{x \in U}$ be a germ of distributions thereon as per Definition \ref{Def: germ}. Given
  $\gamma \in \mathbb{R}$, assume that, for any compact set $K \subset U$,
  there exist $\alpha_K < \gamma$, $R_K, C_K > 0$ and $r_K \in \mathbb{N}$,
  such that, $\forall \bar{\lambda} \in (0, D_K^U)$,
  \begin{equation}
    \label{general scales} \lvert (F_x - F_y) (\varphi^{\lambda}_x) \rvert
    \leqslant C_{K, \bar{\lambda}} \lambda^{\alpha_K} (\lvert x - y \rvert +
    \lambda)^{\gamma - \alpha_K},
  \end{equation}
  uniformly for $\lambda \in (0, \bar{\lambda}]$, $\varphi \in
  \mathcal{B}^{r_K}$, see Equations \eqref{Eq: Rescaled Test-Function} and
  \eqref{Eq: B^m}, and for $x, y \in K$ such that $\lvert x - y \rvert
  \leqslant R_K$. Then $F \in
  \mathcal{G}^{\tilde{\tmmathbf{\alpha}}_{\gamma}}_{\tmmathbf{r}} (U)$ where,
 denoting by $K'$ the smallest convex compact set containing $K$,
  \begin{equation}
    \label{alphak} \tilde{\alpha}_K \equiv \tilde{\alpha}_K (R_K) \assign
    \alpha_K - N_K  (- r_K - d) \qquad \mathrm{with} \qquad N_K \assign
    \left\{\begin{array}{ll}
      \left\lfloor \frac{\mathrm{Diam} (K')}{\min \{R_K, D^U_{K'} / 4\}} + 1
      \right\rfloor, & \text{if} \hspace{0.27em} \hspace{0.27em}
      \hspace{0.27em} R_K < \mathrm{Diam} (K)\\
      0, & \text{if} \hspace{0.27em} \hspace{0.27em} \hspace{0.27em} R_K
      \geqslant \mathrm{Diam} (K)
    \end{array}\right.,
  \end{equation}
  In addition,
  \[ \| F \|_{\mathcal{G}^{\tilde{\tmmathbf{\alpha}}_{\gamma}}_{\tmmathbf{r}}
     (U), K} \lesssim_{D_{K'}^U, \bar{\lambda}} C_{\tilde{K}, \bar{\lambda}},
  \]
  where $\tilde{K} = K$ if $\bar{\lambda} \in (0, D_{K'}^U / 2]$, otherwise
  $\tilde{K} \supset K$ is as per Proposition \ref{relaxingbound11}.
\end{proposition}

\begin{proof}
Consider a compact set $K \subset U$ and take the smallest compact convex set $K'$ containing $K$. Fix $\bar{\lambda} \in (0, \frac{D_{K'}^U}{2})$ and let $\eta_{K'} \in \left( 0,
  \frac{D^U_{K'}}{4} \right)$ be such that $\eta_{K'} \leqslant R_{K'}$, where $R_{K'} > 0$ is given by assumption. Recall that $K'$ is  a compact set such that $K' \subset U$. Given
  two arbitrary points $x, z \in K$, if $\lvert x - z \rvert < R_{K'}$, the sought 
  result is a direct consequence of Equation \eqref{general scales} applied to the compact set $K'$. Thus, the only interesting case is $R_{K'}
  \leqslant \lvert x - z \rvert \leqslant R_{K'} + \eta_{K'}$. Let $y \in K'$
  lie on the line segment connecting $x$ to $z$ so that $\lvert x - y
  \rvert = \eta_{K'} \leqslant R_{K'}$. Then, $\lvert y - z \rvert \leqslant
  R_{K'}$ and, by the triangle inequality, we have that, for any $\varphi \in \mathcal{B}^{r_{K'}}$, 
  \begin{equation}
    \label{extendingbound2eq} \lvert (F_x - F_z) (\varphi^{\lambda}_z)) \rvert
    \leqslant \lvert (F_x - F_y) (\varphi^{\lambda}_z)) \rvert + \lvert (F_y -
    F_z) (\varphi^{\lambda}_z)) \rvert = \lvert (F_x - F_y)
    (\xi^{\lambda_1}_y)) \rvert + \lvert (F_y - F_z) (\varphi^{\lambda}_z)
    \rvert,
  \end{equation}
  where we have re-centered the test function $\varphi^{\lambda}_z$ at $y$.
  More precisely, in Equation \eqref{extendingbound2eq}, $\xi \assign
  \varphi^{\lambda_2}_w$, $\lambda_1 \assign \lvert z - y \rvert + \lambda =
  \eta_{K'} + \lambda \in \left( 0, 3 \frac{D^U_{K'}}{4} \right]$, while
  $\lambda_2 \assign \frac{\lambda}{\lambda_1}$ and $w \assign \frac{y -
  z}{\lambda_1}$. As a consequence $\| \xi \|_{C^{r_{K'}} (U)} = \lambda_2^{- r_{K'} -
  d} \leqslant (3 D_{K'}^U / 4)^{r_{K'} + d} \lambda^{- r_{K'} - d} \lesssim_{D^U_{K'}} \lambda^{- r_{K'} - d}$, $ \xi \in \mathcal{B}^{r_{K'}}$
  and, calling $\bar{\alpha}_{K'} \assign \alpha_{K'} - r_{K'} - d$, we can apply the
  coherence bound in Equation \eqref{general scales} to each term appearing in Equation \eqref{extendingbound2eq}, thus obtaining
  \begin{equation}
    \begin{array}{lll}
      \lvert (F_x - F_z) (\varphi^{\lambda}_x)) \rvert & \leqslant & C_{K',
      \bar{\lambda}} [(3 D_{K'}^U / 4)^{r_{K'} + d} \lambda^{- r_{K'} - d}
      \lambda^{\alpha_{K'}} (\lambda + \lvert x - y \rvert)^{\gamma - \alpha_{K'}} +
      \lambda^{\alpha_{K'}} (\lambda + \lvert y - z \rvert)^{\gamma -
      \alpha_{K'}}]\\
      & \leqslant & C_{K', \bar{\lambda}} [(3D_{K'}^U / 4)^{r_{K'} + d}
      \lambda^{\alpha_{K'}} (\lambda + \lvert x - y \rvert)^{\gamma - \alpha_{K'}} +
      \lambda^{\alpha_{K'}} (\lambda + \lvert y - z \rvert)^{\gamma -
      \alpha_{K'}}]\\
      & \lesssim_{D_{K'}^U} & C_{K', \bar{\lambda}} \lambda^{\bar{\alpha}_{K'}}
      (\lambda + \lvert x - y \rvert + \lvert y - z \rvert)^{\gamma -
      \bar{\alpha}_{K'}}\\
      & = & C_{K', \bar{\lambda}} \lambda^{\bar{\alpha}_{K'}} (\lambda + \lvert
      x - z \rvert)^{\gamma - \bar{\alpha}_{K'}} .
    \end{array} \label{extendingbound2eq1}
  \end{equation}
Equation \eqref{extendingbound2eq1} holds true uniformly for $x, z \in K'$, $\lambda \in \left( 0,
  \frac{D_{K'}^U}{2} \right)$ and $\varphi \in \mathcal{B}^{r_{K'}}$. In particular, in the third
  line, we used that $A^C + B^C \lesssim (A + B)^C$, for $A, B, C$ real
  positive numbers. Hence, we proved the result for all $x, z \in K'$ such that
  $R_{K'} \leqslant \lvert x - z \rvert \leqslant R_{K'} + \eta_{K'}$ and for $\bar{\lambda}
  \in (0, \frac{D_{K'}}{2})$. The most general case, \tmtextit{i.e.}, $\lvert x -
  z \rvert$ not bounded from above, is obtained by iterating this procedure. Finiteness of the number of steps of the procedure is granted by the fact
  that $\tmop{Diam} (K') < \infty$ as $K'$ is itself a compact set. At each
  iteration, the parameter $\alpha_{K'}$, see Definition \ref{Def: Coherence},
  decreases by $- r_{K'} - d$. It descends that, on the compact set $K$,
  \[ \lvert (F_x - F_z) (\varphi^{\lambda}_x) \rvert \lesssim C_{K',
     \bar{\lambda}} \lambda^{\tilde{\alpha}_K} (\lambda + \lvert x - z
     \rvert)^{\gamma - \tilde{\alpha}_K}, \]
  uniformly for $\varphi \in \mathcal{B}^{r_{K}}$, $x, z \in K$ and $\lambda \in (0,
  \bar{\lambda}]$, where $\bar{\lambda} \in \left( 0, \frac{D^U_{K}}{2}
  \right]$. Here $\tilde{\alpha}_K$ is as per Equation
  \eqref{alphak}, in which $N_K$ corresponds to the maximum number of iterations. In
  order to obtain the bound for all $\bar{\lambda} \in (0, D^U_K)$, we can invoke
  Proposition \ref{relaxingbound11}.
\end{proof}

{\noindent} Starting from the results of Propositions \ref{relaxingbound11} and \ref{relaxingbound2}, we can prove the following.

\begin{proposition}
  \label{homcoherence}Let $U \subseteq \mathbb{R}^d$ be an open set. Then
  \begin{enumerate}
    \item If $F \in \mathcal{G}^{\tmmathbf{\alpha}_{\gamma}}_{\tmmathbf{r}}
    (U)$ as per Definition \ref{Def: Coherence}, then $F \in
    \mathsf{G}^{\tmmathbf{\bar{\alpha}}}_{\bar{\tmmathbf{r}}} (U)$, as per
    Definition \ref{Def: Homogeneity} where \ $\bar{\tmmathbf{r}} \assign
    (\bar{r}_K)_{K \subset U}$ is such that $\bar{r}_K = \max \{r_K, r'_K \}$,
    with $r_K$ being the order of the germ, while $r'_K$ is the order of the
    distribution $F_y$ on the compact $K$ at a fixed but arbitrary point $y \in K$. In addition $\tmmathbf{\bar{\alpha}} =
    (\bar{\alpha}_K)_{K \subset U}$ is such that $\bar{\alpha}_K \assign \min
    \{\alpha_K, - \bar{r}_K - d\}$.
    
    \item Let $U \subseteq \mathbb{R}^d$ be a convex open set and let $F \in
    \mathsf{G}^{\gamma}_{\bar{\tmmathbf{r}}} (U)$. Then $F \in
    \mathcal{G}^{\tilde{\tmmathbf{\alpha}}_{\gamma}}_{\tmmathbf{r}} (U)$,
    where $\tilde{\alpha}_K$ is defined in Equation \eqref{alphak} and it
    holds that, denoting by $K'$ the smallest convex compact set containing $K$,
    \[ \| F
       \|_{\mathcal{G}^{\tilde{\tmmathbf{\alpha}}_{\gamma}}_{\tmmathbf{r}}
       (U), K} \lesssim \| F \|_{\mathsf{G}^{\gamma}_{\bar{\tmmathbf{r}}} (U),
       K'} . \]
  \end{enumerate}
\end{proposition}

\begin{proof}
  $1$. Fixed a compact set $K \subset U$ and $y \in K$, the distribution $F_y$
  restricted to $K$ has a finite order, say $r'_K \in \mathbb{N}_0$.
  Therefore, denoting by $\bar{r}_K = \max \{r_K, r'_K \}$, for any
  $\bar{\lambda} \in (0, D_K^U)$, with $D_K^U = \mathrm{dist} (K, \partial U)$,
  \[ \begin{array}{lll}
       \lvert F_x (\varphi^{\lambda}_x) \rvert \leqslant \lvert (F_x - F_y)
       (\varphi^{\lambda}_x) \rvert + \lvert F_y (\varphi^{\lambda}_x) \rvert
       & \lesssim_K & \lambda^{\alpha}_K (\lambda + \lvert x - y
       \rvert)^{\gamma - \alpha_K} + \lambda^{- r_K - d}\\
       & \lesssim_K & \lambda^{\bar{\alpha}_K} ((\bar{\lambda} +
       \mathrm{Diam} (K))^{\gamma - \alpha_K} + 1) .
     \end{array} \]
  The estimate is uniform in $x \in K$, $\lambda \in (0, \bar{\lambda}]$,
  $\varphi \in \mathcal{B}^{\bar{r}_K}$, see Equation \eqref{Eq: B^m}, In
  addition we have set $\bar{\alpha}_K \assign \min \{\alpha_K, - \bar{r}_K -
  d\}$, and we exploited that, for any $n \in \mathbb{N}$, $\|
  \varphi^{\lambda} \|_{C^n (U)} = \lambda^{- n - d} \| \varphi \|_{C^n (U)}$,
  see Equation \eqref{Eq: C^m norm}.
  
  $2$. This is a direct application of Proposition \ref{Prop: Necessity of
  Coherence} setting $f = 0$, since the germ is homogeneous per hypothesis. We
  can thus extend the estimate in Equation (\ref{necessityiofcoherenceqoo})
  for all $\bar{\lambda} \in (0, D_K^U)$, and for all $x, y \in K$ by using
  first Proposition \ref{relaxingbound2} and subsequently Proposition
  \ref{relaxingbound11}.
\end{proof}

\section{Uniform estimates for the normal coordinates}\label{appendixB}

In this appendix we discuss how, given a Riemannian manifold $(M, g)$ and a
compact set $K \subset M$, the Riemannian distance is topologically equivalent
to the Euclidean counterpart when working with local coordinates. The main
result of this appendix is Proposition \ref{topologicalequivalencenormalcoordinates}. We start
by establishing an estimate on manifolds with constant curvature.

\begin{proposition}
\label{estimatesformanifoldswithconstantcurvature}Let $(M, g)$ be a simply
  connected $d-$dimensional Riemannian manifold with constant sectional curvature
  $\mathcal{S}$ and let $p \in M$. Then
  \begin{enumerate}
    \item if $M$ has positive sectional curvature, it is locally isometric to
    an $d$-sphere $\mathbb{S}^d$,
    
    \item if $M$ has negative sectional curvature, it is locally isometric to
    an $d$-dimensional hyperbolic hyperplane $\mathbb{H}^d$,
    
    \item if $M$ has vanishing sectional curvature, it is locally isometric to
    the hyperplane $\mathbb{R}^d$.
  \end{enumerate}
  In addition, for every $p \in M$ and every $q, r \in B (p,
  R_{\mathcal{S}})$, with $R_{\mathcal{S}} \assign \frac{\pi}{8 \hspace{0.17em}
  \sqrt{\lvert \mathcal{S} \rvert}}$, there exist $C, C' > 0$ depending both
  on $p$ and on $\mathcal{S}$ such that,
  \begin{equation}
    \label{boundformanifold with constant curvature} \left\{\begin{array}{ll}
      & \lvert \exp^{- 1}_p (q) - \exp^{- 1}_p (r) \rvert \leqslant C'  d_{\mathbb{H}^d} (q,
      r)  \hspace{0.27em} \hspace{0.27em} \text{for the hyperbolic space}
      \hspace{0.27em} \mathbb{H}^d,\\
      & \lvert \exp^{- 1}_p (q) - \exp^{- 1}_p (r) \rvert \geqslant C d_{\mathbb{S}^d} (q, r)
      \hspace{0.27em} \hspace{0.27em} \text{for the sphere} \hspace{0.27em}
      \mathbb{S}^d.
    \end{array}\right. ,
  \end{equation}
where $d_{\mathbb{H}^d}$ and $d_{\mathbb{S}^d}$ denote the Riemannian distances on the hyperbolic space and on the sphere respectively.
\end{proposition}

\begin{proof}
  As the first part of the Theorem is a standard result in differential
  geometry, see {\cite[Thm. 3.82]{GHF}}, we only prove the first inequality
  in Equation~\eqref{boundformanifold with constant curvature}, the second one
  descending from a similar argument. The first bound in Equation
  \eqref{boundformanifold with constant curvature} can be proven by working on the
  Poincar{\'e} disc model of the hyperbolic space with line element $\mathd
  s^2 = \frac{4}{\sqrt{\lvert \textrm{$\mathcal{S}$} \rvert}  (1 - r^2)^2} 
  (\mathd r^2 + r^2 \mathd \Omega_{d - 1})$, where $\mathd \Omega_{d - 1}$ is
  the volume element of the unit $(d - 1)$-dimensional sphere. The distance
  between two points $q$ and $r$ reads
  \[ d_{\mathbb{H}^d} (q, r) = \frac{1}{\sqrt{\lvert \mathcal{S} \rvert}} \cosh^{- 1} \left(
     1 + 2 \frac{\lvert q - r \rvert^2}{(1 - \lvert q \rvert) (1 - \lvert r
     \rvert)} \right), \]
  where we denoted by $\lvert \cdummy \rvert$ the Euclidean distance.
  Considering $q, r \in B \left( p, \frac{\pi}{8 \hspace{0.17em} \sqrt{\lvert
  \mathcal{S} \rvert}} \right)$, it holds that
  \[ \cosh \left( \sqrt{\lvert \mathcal{S} \rvert} \lvert q - r \rvert \right)
     \lesssim_{\mathcal{S}} 1 + 2 \lvert q - r \rvert^2 \lesssim_\mathcal{S} \left( 1 + 2
     \frac{\lvert q - r \rvert^2}{(1 - \lvert q \rvert) (1 - \lvert r \rvert)}
     \right), \]
  which entails
  \[ \lvert q - r \rvert \lesssim_{\mathcal{S}} \cosh^{- 1} \left( 1 + 2
     \frac{\lvert q - r \rvert^2}{(1 - \lvert q \rvert) (1 - \lvert r \rvert)}
     \right) . \]
  The sought identity descends noting that, being the Poincar{\'e} disk a
  subset of $\mathbb{R}^d$, the exponential map acts therein as the identity.
\end{proof}

The following theorem allows us to compare, under suitable assumptions, the arc
lengths of curves in two different Riemannian manifolds. The proof can be found in
{\cite[Cor. 1.35]{CE}}.

\begin{theorem}
  \label{corollaryofrauch1}Let $(M,g), (M_0, g_0)$ be Riemannian manifolds with $\dim M_0
  \geqslant \dim M$ and let $m \in M$ while $m_0 \in M_0$. Assume that for all
  pair of vectors $v, w \in T_m M$, $v_0, w_0 \in T_{m_0} M_0$,
  $\mathcal{S}_{M_0} (v_0, w_0) \geqslant \mathcal{S}_M (v, w)$ where
  $\mathcal{S}_{\cdot}$ denotes the sectional curvature. Let $\rho > 0$ be such
  that $\exp_m$ restricted to the ball $B (0, \rho)$ is an embedding while
  $\exp_{m_0}$ restricted to $B (0, \rho)$ is non-singular. Let $\iota : T_m M \to
  T_{m_0} M$ be a linear injection preserving inner products. Then for any
  curve $c : [0, 1] \to \exp_m (B (0, \rho))$, we have
  \begin{equation}
    \label{rauch1eq} L [c] \geqslant L [\exp_{m_0} \circ \iota \circ \exp_m^{- 1}
    (c)] \assign L [c_0 (t)],
  \end{equation}
  where we denoted the arc length of a curve $\gamma$ with $L [\gamma (t)]$.
\end{theorem}

{\noindent}Theorem \ref{corollaryofrauch1} entails the following useful estimate.

\begin{proposition}
  \label{topologicalequivalencenormalcoordinates}Let $(M, g)$ be a
  $d$-dimensional Riemannian manifold. Consider $\{\Omega_n\}_{n \in \mathbb{N}}$, an exhaustion by compacts of $M$ as per Definition \ref{Def: Compact Exhaustion}. Endow $M$ with a finite good cover $(U_i,
  p_i)_{i \in I}$, see Definition \ref{Rem: Good Covering}, so that each $U_i$
  is a geodesically convex open subset with $\tmop{diam} (\exp_{p_i}^{- 1}
  (U_i)) \leqslant \min \{ \frac{\mathcal{R}_C (p)}{2}, \frac{\pi}{8
  \sqrt{\lvert \mathcal{S} \rvert}} \}$. Here $\lvert \mathcal{S} \rvert$ is
  the operator norm of the sectional curvature tensor evaluated in the compact
  set $\overline{\Omega_{\bar{n}}}$, where $\bar{n}$ is the smallest integer $n \in \mathbb{N}$ such
  that $U_i \subset \Omega_n$. Fixed $j \in I$, and a compact set $K \subset
  M$ such that $U_j \subset K$, there exist two constants $\mathfrak{C}_K, \mathfrak{C}_K' > 0$, depending
  solely on $\bar{n}$ and $K$, \ such that
  \begin{equation}
    \label{topologicalequivalencenormalcoordinateseq} \mathfrak{C}_K' d_g (q, r)
    \leqslant \lvert \exp^{- 1}_{p_j} (q) - \exp^{- 1}_{p_j} (r) \rvert \leqslant
    \mathfrak{C}_K d_g (q, r), \qquad \forall q, r \in U_j .
  \end{equation}
\end{proposition}

\begin{proof}
For the sake of clarity, we shall prove separately the two inequalities in Equation \eqref{topologicalequivalencenormalcoordinates}. 

\vspace{.2cm}

\noindent $(\leqslant)$ -- Fix $j \in I$ and consider the geodesically convex open set $U_j \subset \Omega_n$ for a given $n \in \mathbb{N}$. We start by showing that, for any $q,r \in U_j$, 
  \[ \lvert \exp^{- 1}_{p_j} (q) - \exp^{- 1}_{p_j}
     (r) \rvert \leqslant \mathfrak{C}_K d_g (q, r) . \]
  We fix $p_0 \in \mathbb{S}^d$ and a point $p_j \in U_j \subset \Omega_n$. For the sake of notational ease, we
  denote
  \[ \exp_{p_j} \equiv \exp_j . \]
 The sectional curvature restricted to
  $\overline{\Omega_n}$ is bounded from above. Since $\overline{\Omega_n}$ is
  compact, we choose the radius of the model manifold $\mathbb{S}^d$ so that
  $\mathcal{S}_{\mathbb{S}^d} \geqslant \mathcal{S}_{\Omega_n}$ in order for
  the hypotheses of Theorem \ref{corollaryofrauch1} to be abode by. Considering
  any bijection $\iota : T_{p_j} M \to T_{p_0} \mathbb{S}^d$, which preserves the
  inner product, the map $\varphi \assign \exp_{p_0} \circ \iota \circ \exp_j^{-
  1} : U_j \to \exp_{p_0} (\iota \circ \exp_j^{- 1} (U_j))$ is a diffeomorphism
  between geodesically convex sets. This is due to the fact that the sets $\{U_i\}_{i \in I}$
  are chosen as in Definition \ref{Rem: Good Covering} and the diameter
  of the geodesic neighbourhood $\exp_{j}^{- 1}
  (U_j)$ is smaller than $\frac{\pi}{8 \sqrt{\lvert
  \mathcal{S} \rvert}}$ per hypothesis. Hence, given two points $q, r \in U_j$, we consider a
  minimizing geodesic $c : [0, 1] \to U_j$, so that $L [c] = d_g (q, r)$. As a
  consequence of Theorem \ref{corollaryofrauch1},
  \begin{equation}
    \label{comparison result 0.1} d_g (q, r) = L [c] \geqslant L [c_0] \geqslant
    d_0 (\varphi (q), \varphi (r))
  \end{equation}
  where $d_0$ is the Riemannian distance on $\mathbb{S}^d$, $c_0 \assign
  \varphi \circ c$, while in the last inequality we used that, by definition, $d_0 (\varphi
  (q), \varphi (r))$ is the value minimizing the length of a curve on
  $\mathbb{S}^d$ connecting the corresponding two points. By Equation
  \eqref{boundformanifold with constant curvature},
  \[ d_0 (\varphi (q), \varphi (r)) \gtrsim \lvert \exp_{p_0}^{- 1} (\varphi
     (q)) - \exp_{p_0}^{- 1} (\varphi (r)) \rvert, \]
  and, recalling that $\varphi = \exp_{p_0} \circ \iota \circ \exp_j^{- 1}$, we
  have
  \[ \lvert \exp_{p_0}^{- 1} (\varphi (q)) - \exp_{p_0}^{- 1} (\varphi (r))
     \rvert = \lvert \exp_j^{- 1} (q) - \exp_j^{- 1} (r) \rvert. \]
  Combining these two equations together, we obtain
  \begin{equation}
    \label{comparison result 0.2} \lvert \exp_j^{- 1} (q) - \exp_j^{- 1} (r)
    \rvert \lesssim d_0 (\varphi (q), \varphi (r)) .
  \end{equation}
  Plugging now Equation \eqref{comparison result 0.2} into Equation
  \eqref{comparison result 0.1}, we get 
  \[ d_g (q, r) \gtrsim \lvert \exp^{- 1}_j (q) - \exp^{- 1}_j (r) \rvert . \]

\vspace{0.2cm}

\noindent $(\geqslant)$. Adopting a similar strategy, we shall now prove the first inequality in Equation
\eqref{topologicalequivalencenormalcoordinateseq}. Once again, we shall invoke Theorem \ref{corollaryofrauch1}, although this time the role of $M$ and the model manifold $M_0 = \mathbb{H}^d$ should be reversed. 

Let $p_0 \in \mathbb{H}^d$ and $p_j \in U_j \subset \Omega_n$, for a given $n \in \mathbb{N}$. By hypothesis, we know that the sectional curvature of $U_j$ is bounded from below by a negative constant depending on $n$. We choose the metric on the hyperbolic space $\mathbb{H}^d$ in such a way it has a negative curvature given by this constant and we fix two arbitrary points $q, r \in U_j$. Denoting by $\iota : T_p M \to T_{p_0} \mathbb{H}^d$ any bijection preserving the inner product, let us consider the diffeomorphism $\varphi \assign \exp_{p_0} \circ \iota \circ \exp_j^{- 1} : U_j \to \exp_{p_0} (\iota \circ \exp_j^{- 1} (U_j))$. We now choose $c : [0, 1] \to U_j$, so that $c_0 \assign \varphi \circ c$ is a minimizing geodesic between the points $\varphi (q)$ and $\varphi (r)$. Theorem \ref{corollaryofrauch1}
  entails that
  \begin{equation}
    \label{comparison result 0.11} d_0 (\varphi (q), \varphi (r)) = L [c_0]
    \geqslant L [c] \geqslant d_g (p, q),
  \end{equation}
  where, in the last inequality, we used that the geodesic distance is the minimum
  value of the length of a curve connecting the two endpoints. Equation
  \eqref{boundformanifold with constant curvature} implies that
  \[ \lvert \exp_{p_0}^{- 1} (\varphi (q)) - \exp_{p_0}^{- 1} (\varphi (r))
     \rvert \gtrsim d_0  (\varphi (q), \varphi (r)), \]
  and, since $\varphi = \exp_{p_0} \circ \iota \circ \exp_j^{- 1}$, we
  have
  \[ \lvert \exp_{p_0}^{- 1} (\varphi (q)) - \exp_{p_0}^{- 1} (\varphi (r))
     \rvert = \lvert \exp_j^{- 1} (q) - \exp_j^{- 1} (r) \rvert, \]
  hence, combining the previous two equations, 
  \begin{equation}
    \label{comparison result 0.21} \lvert \exp_j^{- 1} (q) - \exp_j^{- 1} (r)
    \rvert \gtrsim d_0  (\varphi (q), \varphi (r)) .
  \end{equation}
  Finally plugging Equation \eqref{comparison result 0.21} into Equation
  \eqref{comparison result 0.11}, we obtain
  \[ \lvert \exp_j^{- 1} (q) - \exp_j^{- 1} (r) \rvert \gtrsim d_g (q, r) . \]
\end{proof}

{\noindent}The following lemma is a direct consequence of {\cite[Thms. I.3.1,
I.3.2]{Chav}}, thus we omit the proof.  

\begin{lemma}
  \label{Lem:smoothness of diffeomorphism}Let $(M,g)$ be a Riemannian manifold and let $K
  \subset M$ be a compact subset. Denote by $\mathcal{R}_C (K)$ the convexity radius of $K$
  and fix a point $p \in K$. For $q \in B (p, \mathcal{R}_C (K) / 4)$, consider the
  diffeomorphism $\exp_q^{- 1} \circ \exp_p : B (0, \mathcal{R}_C (K) / 4) \to
  B (0, \mathcal{R}_C (K) / 2)$, and consider the transition map as a function
  of $q, s \in B (p, \mathcal{R}_C (K) / 4)$
  \begin{equation}
    \exp_{\cdot}^{- 1} \circ \exp_p (\cdot) : (q, s) \mapsto \exp_q^{- 1}
    \circ \exp_p (s),
  \end{equation}
  Then this map is smooth with respect to $q, s \in B (p, \mathcal{R}_C (K) / 4)$.
\end{lemma}

\section{Scaling of test functions on manifolds}\label{AppendixC}

In this appendix we develop the tools needed to 
scale and recenter test functions on a Riemannian manifold $(M,g)$. Henceforth, with a slight abuse of notation, we shall denote $(M,g)$ simply by $M$, leaving implicit the Riemannian metric $g$. In addition, throughout this section, we shall denote by $U$ a geodesically convex open
neighbourhood of $M$ with associated exponential map $\exp \equiv \exp_p$ centered at a point $p \in M$, and by $K \subset U$ a compact
set. In particular, we define $D_K^U \assign \mathrm{dist} (K, \partial U)$
and we indicate by $\mathcal{R}_C (K)$ the convexity radius of $K$ as per Definition \ref{Def:
Convexity Radius}.

We begin by proving two lemmata which yields a systematic way to perform the rescaling of test functions in the Riemannian setting. The idea is to relate the scaling properties of test functions defined on open subsets of the tangent space -- which is diffeomorphic to $\mathbb{R}^d$ -- with those of test functions living on geodesically convex open sets of $M$. 

\begin{lemma}
  \label{fromflattomanifold} In the aforementioned setting, fix a point $q \in K$ and let
  $\varphi \in \mathcal{B}^r$ for a given $r \in \mathbb{N}$ as in Equation \eqref{Eq: B^m}. Given any but
  fixed $\bar{\lambda} \in \left( 0, \min \{ \frac{\mathcal{R}_C
  (K)}{\mathfrak{C}'_K}, D^U_K \} \right)$ and $\lambda \in (0, 1)$, where
  $\mathfrak{C}'_K$ is the constant appearing in Equation \eqref{topologicalequivalencenormalcoordinateseq}, there exists
  $\phi_{[\lambda]} \in \bar{\lambda}^{- r - d}\mathcal{B}^r_{ q}(K)$ such
  that
  \begin{equation}
      \label{AppC: expeq}
      \exp_{\ast} \left( \varphi^{\lambda \hspace{0.17em}
     \bar{\lambda}}_{\exp^{- 1} (q)} \right) = \phi^{\lambda}_{[\lambda], q}.
  \end{equation}
Here the exponential map $\exp^{-1}$ is assumed to be centered at a fixed point $p \in M$. Moreover, the function $\phi_{[\lambda], q}^{\lambda}$ on the right hand-side is as per Equation \eqref{evaluationonmanifolds}. 
\end{lemma}

\begin{proof}
  Since $q \in K$, it holds that $\exp^{- 1} (q) \in T_p M$ and, therefore,
  by Equation \eqref{Eq: Rescaled Test-Function}, we can define
  $\varphi_{\exp^{- 1} (q)}^{\tilde{\lambda}}$ where $\varphi \in
  \mathcal{B}^r$ while $\tilde{\lambda} \assign \lambda \bar{\lambda}$. Note that $\text{supp} \, (\varphi_{\exp^{- 1} (q)}^{\tilde{\lambda}}) \subset B(\exp^{-1}(q), \tilde{\lambda})$.  Proposition \ref{topologicalequivalencenormalcoordinates} entails that the support of
  $\exp_{\ast} \varphi_{\exp^{- 1} (q)}^{\tilde{\lambda}}$ is contained in the geodesic open ball $B(q, \frac{\tilde{\lambda}}{\mathfrak{C}'_K})$ and, hence, in $B(q, R)$ with $R \lesssim_{K} \lambda \mathcal{R}_C (K)$. Thus, for each such $\tilde{\lambda}$ and for any $r \in U$, we define $\phi_{[\lambda], q}^{\lambda} (r)
  \assign [\exp_{\ast} \varphi_{\exp^{- 1} (q)}^{\lambda \bar{\lambda}}] (r)$ and 
   \begin{equation}\label{Eq: Aux1}
       \phi_{[\lambda], q} (r) \assign \lambda^d \exp_{q \ast} (\exp_{\ast}
       \varphi_{\exp^{- 1} (q)}^{\lambda \bar{\lambda}}) (\lambda \exp_q^{- 1}
       (r)) .
     \end{equation} 
It follows that $\text{supp}(\phi_{[\lambda], q}) \subset B
  (q, \rho_K)$, with $\rho_K := \min\{\frac{R_C(K)}{2}, 1 \}$ as per Equation~\eqref{Bronmanifold}. 
  Observe that $\exp$ denotes the exponential map centered at $p \in M$ associated with a local
  trivialization $U$, whilst $\exp_q$ that yielding local coordinates centered at $q \in K$. Note in addition
  that $\| \phi_{[\lambda], q} \|_{C^r} = \| \lambda^d \exp_{\ast}
  \varphi_{\exp^{- 1} (q)}^{\lambda \bar{\lambda}} \|_{C^r} \leqslant
  \bar{\lambda}^{- r - d}$.
\end{proof}

\begin{lemma}
  \label{frommanifoldtoflat} Under the same assumptions of Lemma \ref{fromflattomanifold}, let
  $K \subset U$ be a compact set and fix a point $q \in K$. Consider a test function $\varphi_p \in
  \mathcal{B}^r_{p}(K)$. Set $\rho_K \assign \frac{\mathcal{R}_C (K)}{2}$ and
  $\mathfrak{C}_K$ the constant appearing in Equation \eqref{topologicalequivalencenormalcoordinateseq}. Then, fixing
  $\bar{\lambda} \in (0, \min \{\rho_K, D_K^U \mathfrak{C}_K, \mathfrak{C}_K
  \})$, for every $\lambda \in (0, 1)$ there exists a test function
  $\psi_{[\lambda]} \in \bar{\lambda}^{- r - d} \mathcal{B}^r$ such that
  \[ \exp^{- 1}_{\ast} (\varphi^{\lambda \bar{\lambda}}) = (\psi_{[\lambda]})_{\exp^{- 1}(q)}^{\lambda}. \]
  Note that, on the left hand-side, the scaling is done with respect to
  the the map $\exp_q$, while, on the right hand-side, with respect to the
  map $\exp$ centered at a point $p \in M$.
\end{lemma}

\begin{proof}
  We follow a scaling procedure analogous to the one
  used in Lemma \ref{fromflattomanifold} but we reverse the steps. We work
  with $\lambda \in (0, 1)$. Hence, by Proposition
  \ref{topologicalequivalencenormalcoordinates}, $\exp^{- 1}_{\ast}
  (\varphi_q^{\lambda\bar{\lambda}})$ is supported in a
  geodesic ball of radius $R \leqslant \lambda \min \{D_K^U,
  1\}$, centered in $\exp^{- 1} (q)$. Thus, we define $(\psi_{[\lambda]})_{\exp^{- 1}(q)}^{\lambda} \assign \exp^{- 1}_{\ast} (\varphi_q^{\lambda
  \bar{\lambda}})$ and, for each $\lambda \in (0, 1)$, we consider the
  function
  \[ \psi_{[\lambda]} \assign \lambda^d (\exp^{- 1}_{\ast} (\varphi_q^{\lambda
     \bar{\lambda}}))^{\lambda^{- 1}}_{- \exp^{- 1} (q)}. \]
 Observe that $\|
  \psi_{[\lambda]} \|_{C^r} = \| \lambda^d \exp^{- 1}_{\ast}
  (\varphi_q^{\lambda \bar{\lambda}}) \|_{C^r} \leqslant
  (\bar{\lambda})^{- r - d}$, therefore $\psi_{[\lambda]} \in
  \bar{\lambda}^{- r - d} \mathcal{B}^r$.
\end{proof}

\noindent Finally, let us discuss 

\begin{lemma}[Recentering a test function]
  \label{recenteringonmanifolds} 
Let $M$ be a Riemannian manifold, and let $U
  \subset M$ be a geodesic convex set with associated exponential map
  $\exp$. Let $K \subset U$ be a compact set. There exist positive real
  constants $\tmop{cst}_K, \tmop{cst}_K', \widetilde{\tmop{cst}}_K$, such that,
for every $q \in K$ with $d_g (p, q) \leqslant \widetilde{\tmop{cst}}_K$, every $\phi_p \in
  \mathcal{B}^r_{p}(K)$ can be recentered as
  \[ \phi^{\lambda}_p = \xi^{\lambda''}_{[\lambda], \hspace{0.17em} q}, \]
  with $\lambda'' = \lambda_1 / \tmop{cst}_K'$, $\lambda_1 = \lvert \exp^{- 1}
  (p) - \exp^{- 1} (q) \rvert + \lambda / \tmop{cst}_K$ while $\xi_{[\lambda],
  \hspace{0.17em} q} \in (\tmop{cst}_K' \lambda / \lambda_1)^{- r - d}
  \mathcal{B}^r_q (K)$. In addition $\lambda'' \lesssim 1$, and, consequently, $(\tmop{cst}_K' \lambda / \lambda_1)^{- r - d} \lesssim_K \left(
  \frac{\lambda}{\lambda + d_g (p, q)} \right)^{- r - d}$.
\end{lemma}

\begin{proof}
  We shall denote $\tmop{cst}_K \assign \min \{\rho_K, D^U_K
  \mathfrak{C}_K, \mathfrak{C}_K \}$ and $\tmop{cst}_K' \assign \min \{
  \frac{\rho_K}{\mathfrak{C}_K'}, D_K^U \}$, where $\rho_K$, $\mathfrak{C}_K$
  and $\mathfrak{C}_K'$ are as per Definition
  \ref{evaluationonmanifolds} and Proposition
  \ref{topologicalequivalencenormalcoordinates}, respectively. We start by
  fixing $\lambda \in (0, \min \{\tmop{cst}_K, \tmop{cst}_K' \}/ 2)$, and $p,
  q \in K$ such that $d_g (p, q) \leqslant \frac{\tmop{cst}_K'}{2\mathfrak{C}_K}
  \assign \widetilde{\tmop{cst}}_K$. Additionally, denote by $\lambda' \assign
  \lambda / cst_K$. Then, fixing $\varphi \in \mathcal{B}^r_{p}(K)$ and using
  Lemma \ref{frommanifoldtoflat}, we have that
  \begin{equation}
    \label{recenteringonmanifold1} \exp^{- 1}_{\ast} (\varphi^{\lambda'
    \tmop{cst}_K}) = \psi_{[\lambda'], \exp^{- 1} (q)}^{\lambda'},
  \end{equation}
  with $\psi_{[\lambda']} \in \tmop{cst}_K^{- r - d} \mathcal{B}^r$. Using
  the results on the recentering of test functions on $\mathbb{R}^d$, see {\cite[Prop. 6.2]{CZ20}} and the proof of Proposition \ref{relaxingbound2}, we have 
  \begin{equation}
    \label{recenteringonmanifold2} \psi_{[\lambda], \exp^{- 1} (q)}^{\lambda'}
    = \zeta^{\lambda_1}_{\exp^{- 1} (q)},
  \end{equation}
  where $\zeta \assign \varphi^{\lambda_2}_w$, $\lambda_1 \assign \lvert
  \exp^{- 1} (p) - \exp^{- 1} (q) \rvert + \lambda'$, $\lambda_2 \assign
  \frac{\lambda'}{\lambda_1}$ and $w \assign \frac{\exp^{- 1} (p) - \exp^{- 1}
  (q)}{\lambda_1}$. Note that, since $\lambda_2 + \lvert w \rvert \leqslant
  1$, $\zeta \in (\tmop{cst}_K  \hspace{0.17em} \frac{\lambda'}{\lambda_1})^{-
  r - d}  \hspace{0.27em} \mathcal{B}^r$. We now apply Proposition
  \ref{fromflattomanifold} to pull back the test function on the manifold
  via the map $\exp$. To this end, we need
  $\lambda_1 \leqslant \tmop{cst}_K'$, which is guaranteed by the choice of the
  bounds on $d_g (p, q)$ and $\lambda$. Setting $\lambda'' \assign \lambda_1 /
  \tmop{cst}_K'$ we have
  \begin{equation}
    \label{recenteringonmanifold3} (\tmop{cst}_K \lambda' / \lambda_1)^{- r -
    d} \exp_{\ast} (\zeta^{\lambda'' \tmop{cst}_K'}_{\exp^{- 1} (q)}) = \left(
    \hspace{0.17em} \frac{\lambda}{\lambda_1} \right)^{- r - d}
    \eta_{[\lambda], q}^{\lambda''}.
  \end{equation}
  Renaming $\xi^{\lambda''}_{[\lambda], \hspace{0.17em} q} \assign \left(
  \frac{\lambda}{\lambda_1} \right)^{- r - d} \eta_{[\lambda''] 
 q}^{\lambda''}$, we have that $\xi_{[\lambda],
  \hspace{0.17em} q} \in (\tmop{cst}_K' \lambda / \lambda_1)^{- r - d}
  \mathcal{B}^r_{q}(K)$. Combining Equations \eqref{recenteringonmanifold1},
  \eqref{recenteringonmanifold2} and \eqref{recenteringonmanifold3} we obtain the sought result, namely
  \[ \phi^{\lambda}_p = \xi^{\lambda''}_{[\lambda], q}. \]
\end{proof}

{\noindent}We conclude this appendix by proving a generalization of Proposition
\ref{relaxingbound11} to the Riemannian scenario. 

\begin{proposition}[Extending the bound on Riemannian manifolds]
  \label{extendingthebound1onmanifold} Let $M$ be a Riemannian manifold and
  let $F$ be a germ of distributions thereon, see Section \ref{Sec:Germs-of-distributions-manifold}. Assume that, for every compact $K \subset M$, there exist
  $\eta_K \leqslant \rho_K$ with $\rho_K := \min \{\frac{R_C(K)}{2}, 1\}$, $R_K, C_K > 0$, $\gamma \in \mathbb{R}$ and $\alpha_K <
  \gamma$ such that, for $\bar{\lambda} \in (0, \eta_K)$
  \begin{equation}
    \label{almostcoherence} \lvert (F_p - F_q) (\varphi_q^{\lambda}) \rvert
    \leqslant C_K \lambda^{\alpha_K} (\lambda + d_g (p, q))^{\gamma - \alpha_K},
  \end{equation}
  uniformly for $p, q \in K$ such that $d_g (p, q) \leqslant R_K$, $\lambda \in
  (0, \bar{\lambda}]$, $\varphi_q \in \mathcal{B}^r_{q} (K)$. Then $F \in
  \mathcal{G}^{\text{\tmtextbf{$\alpha$}}_{\gamma}}_{\tmmathbf{r},
  \text{\tmtextbf{$R$}}} (M)$.
\end{proposition}

\begin{proof}
  Define $\tilde{C}_K \assign \min \{R_K, \rho_K \}$. We first prove
  that Equation \eqref{almostcoherence} holds true for all $\bar{\lambda} \in
  (0, \eta_K + \tilde{C}_K)$. The general case then follows by iteration of this procedure. For the sake of clarity, we divide the proof in three steps. 
  \begin{enumerate}
    \item Denote by $\tilde{K}$ the $\tilde{C}_K-$enlargement of
    $K$ as per Equation \eqref{Eq: K-enlargement}, which is still a compact set. As a matter of fact, we can cover
    $K$ with a finite number of balls $\{B_i\}_{i \in I}$ such that $ B_i \assign B(p_i, \frac{3}{4}
    \mathcal{R}_C (K))$ with $p_i \in K$. Let $\{K_i\}_{i \in I}$ be a collection of compact sets such that $K_i \subset B_i$ and $\cup_{i \in I} K_i = K$. We
    call $\tilde{K}_i$ the $\tilde{C}_K-$ enlargement of $K_i$. These are
    compact sets, since they can be constructed via the action of the exponential map on the compact sets $(\exp^{- 1} (K_i)_{\tilde{C}_K}) \subset \mathbb{R}^d$, \textit{i.e.}, $\tilde{K}_i = \exp
    (\exp^{- 1} (K_i)_{\tilde{C}_K})$. Hence, $\tilde{K} = \cup_{i \in
    I} \tilde{K}_i$.
    
    \item If $\lambda \in (0, \eta_K)$, there is nothing to prove. Thus, let $\lambda \in (\eta_K, \eta_K + \tilde{C}_K)$. Consider 
    $\{U_j\}_{j \in J}$, a covering of $K$, made of balls of radius $\eta_K$ centered at
    points $q_j \in \tilde{K}$ and a partition of unity $\{\psi_j\}_{j \in J}$ subordinated to it. Then, for any $q \in K$, we can write
    \[ \varphi^{\lambda}_q = \sum_{j \in J} \varphi_j^{\eta_K}, \]
    where $\varphi_j := (\varphi^{\lambda}_q \psi_{q_j})^{1/\eta_K}$. 
    We stress that $\mathcal{R}_C (\tilde{K}) \geqslant \rho_K$ and that the
    number of points $q_j$ depends only on $\eta_K$ and, hence, on $K$.
    
    \item Finally, in the same spirit of the proof of Proposition \ref{relaxingbound2}, observe that, given $\lambda \in (\eta_K, \eta_K +
    \tilde{C}_K)$, it holds $1 \lesssim \lambda^{\alpha_K} (\lambda + d_g (p,
    q))^{\gamma - \alpha_K}$
    \[ \begin{array}{lll}
         \lvert (F_p - F_q) (\varphi_q^{\lambda}) \rvert = \left| (F_p - F_q)
         \left( \sum_{j \in J} \varphi_j^{\eta_K} \right) \right| & \leqslant &
         \sum_{j \in J} (\lvert (F_p - F_{q_j}) (\varphi_j^{\eta_K}) \rvert +
         \lvert (F_{q_j} - F_q) (\varphi_j^{\eta_K}) \rvert)\\
         & \lesssim_{K, \eta_K} & C_K\\
         & \lesssim & C_K \lambda^{\alpha_K} (\lambda + d_g (p, q))^{\gamma -
         \alpha_K} .
       \end{array} \]
  \end{enumerate}
  
\end{proof}


\begin{thebibliography}{10}
	
\bibitem[BCD11]{Bahouri}
H. Bahouri, J.-Y. Chemin, R. Danchin,
\textit{``Fourier Analysis and Nonlinear Partial Differential Equations''},
 Fourier Analysis and Nonlinear Partial Differential Equations, (2011)
Springer, 540p.
	
\bibitem[Ber03]{Ber}
M. Berger,
\textit{``A Panoramic View of Riemannian Geometry''},
(2003), Springer, 824p.

\bibitem[BCD+24]{BCD24}
A. Bonicelli, B. Costeri, C. Dappiaggi and P. Rinaldi, 
 \textit{``A microlocal investigation of stochastic partial differential equations for spinors with an application to the Thirring model''}, Math. Phys. Anal. and Geom. \textbf{27}, 3 (2024) 16.

\bibitem[BDR23]{BDR23}
A. Bonicelli, C. Dappiaggi and P. Rinaldi,
 \textit{``An algebraic and microlocal approach to the stochastic nonlinear Schrödinger equation}, Ann. Henri Poinc. \textbf{27}, 7 (2023)  2443.

\bibitem[BDR24]{BDR24}
 A. Bonicelli, C. Dappiaggi and P. Rinaldi, 
 \textit{``On the stochastic Sine-Gordon model: an interacting field theory approach"}, Commun. Math. Phys. \textbf{405}, 12 (2024) 288.
	
\bibitem[Bou23]{B}
N. Boumal,
\textit{``An Introduction to Optimization on Smooth Manifolds''},
(2023), Cambridge University Press 338p.

\bibitem[BCZ24]{BCZ}
L. Broux, F. Caravenna, L. Zambotti,
\textit{`` Hairer's Multilevel Schauder Estimates Without Regularity Structures''},
Trans. Amer. Math. Soc. \tmtextbf{377} (2024), 6981.

\bibitem[BFD+15]{BFDY}
R. Brunetti, C. Dappiaggi, K. Fredenhagen and Y. Yngvanson, 
 \textit{``Advances in Algebraic Quantum Field Theory''}, Mathematical Physics Studies (2015) Springer, 455p. 

\bibitem[BF09]{BF09}
R. Brunetti and K. Fredenhagen, 
\textit{``Advances in Algebraic Quantum Field Theory''}, Lecture Notes in Phys., Vol. \textbf{786} Springer (2009), pp. 129-155
 
 \bibitem[CZ20]{CZ20}
F. Caravenna, L. Zambotti,
 \textit{`` Hairer's reconstruction theorem without regularity structures''},
 EMS Surv. Math. Sci \tmtextbf{7} (2020) 207.
 
  \bibitem[CE08]{CE}
 J. Cheeger and D. G. Ebin,
 \textit{``Comparison Theorems in Riemannian Geometry''},
 American mathematical Society (2008), 161p.
 
\bibitem[Chav06]{Chav}
 I. Chavel,
 \textit{``Riemannian Geometry, a modern introduction''},
 Cambridge studies in advanced mathematics, 2nd edition (2006).

\bibitem[DDR+20]{DDRZ}
C. Dappiaggi, N. Drago, P. Rinaldi and L. Zambotti, 
\textit{``A microlocal approach to renormalization in stochastic PDEs"}, Comm. Cont. Math. \textbf{24}, (2020), 7.

\bibitem[FR16]{FR}
K. Fredenhagen and K. Rejzner, 
\textit{``Quantum field theory on curved spacetimes: Axiomatic framework and examples"}, J. Math. Phys. \textbf{57}, 3 (2016), 031101.
  
 \bibitem[GHF04]{GHF}
S. Gallot, D. Hulin and J. Lafontaine,
\textit{``Riemannian Geometry''},
Springer (2004), 322p.

\bibitem[GIP15]{GIP}
M. Gubinelli, P. Imkeller, N. Perkowski, \textit{``Paracontrolled Distributions and singular PDEs"},
Forum of Mathematics, Pi, Volume 3 (2015) e6.

\bibitem[Hai14]{Hai} 
M. Hairer, \textit{``A Theory of Regularity Structures.''} 
Invent. math. {\bf 198} (2014), 269.

  \bibitem[HS23]{HS23}
M. Hairer, H. Singh
 \textit{``Regularity Structures on Manifolds and Vector Bundles''},
 arXiv:2308.05049 [math.PR].

 \bibitem[KPZ86]{KPZ} 
M. Kardar, G. Parisi, Y. Zhang. \textit{``Dynamic Scaling of Growing Interfaces.''} 
 Phys. Rev. Lett. {\bf 56} (1986), 889. 
 
 \bibitem[Lee12]{LeeS}
 J. Lee,
 \textit{``Introduction to Smooth Manifolds''},
 Springer (2012).
 
  \bibitem[Lee18]{LeeR}
 J. Lee,
 \textit{``Introduction to Riemannian Manifolds''},
 Springer (2018).
  
 \bibitem[RS21]{RS21}
P. Rinaldi, F. Sclavi,
\textit{``Reconstruction theorem for germs of distributions on smooth manifolds''},
Jour. of Math. Analysis and Appl. \tmtextbf{501} (2021)
125215.
\end{thebibliography}
\end{document}